\documentclass[reqno]{tran-l}

\usepackage{amsmath}
\usepackage{amssymb}
\usepackage{amsthm}
\usepackage{pdfsync}
\usepackage{graphicx}
\usepackage{comment}
\usepackage[left=1.3in, right=1.3in, top=0.75in, bottom=0.8in]{geometry}

\allowdisplaybreaks[3]

 \newtheorem{thm}{Theorem}[section]
 \newtheorem{cor}[thm]{Corollary}
 \newtheorem{lem}[thm]{Lemma}
 \newtheorem{prop}[thm]{Proposition}
 \theoremstyle{definition}
 
 \theoremstyle{remark}
 
 \numberwithin{equation}{section}
 \DeclareMathOperator{\RE}{Re}
 \DeclareMathOperator{\IM}{Im}

 \newcommand{\h}{\mathcal{H}}

 \newcommand{\abs}[1]{\left\vert#1\right\vert}
 \newcommand{\set}[1]{\left\{#1\right\}}

\begin{document}

\title{Schr\"{o}dinger's Equation is Universal, Dark Matter and Double Diffusion}

\author{ Johan G.B. Beumee, Herschel Rabitz}

\address{Flat 36, Bowles Lodge, All Saint's Road, Hawkhurst,
United Kingdom}

\email{jbeumee@gmail.com}

\thanks{Rights Reserved JoDe Group Financial 12 April 2021}

\thanks{The author thanks Mark Davis and Chris Rogers for comments.}


\date{}

\keywords{forward stochastic processes, backward stochastic processes, quantum mechanics, Schr\"{o}dinger equation, heatbath, diffusion, elastic collisions, Special Relativity, Dark Matter, Double Diffusion}


\commby{James F. & Fransina I. Beumee}


\begin{abstract}
This paper considers a main particle and an incident particle classical mechanics elastic collision preserving energy and momentum while ignoring the angular momentum, spin or other particle characteristics. The elastic collision is modelled using projections which reduce the collision to the one-dimensional main and incident particle surface. The main result of the paper shows that the colliding two particle classical Hamiltonian energy can be represented in four weighted individual particle in symmetric and anti-symmetric (osmotic) terms similar to the quadratic Nelson measure used in the derivation of the Schr\"{o}dinger wave function. However, these four energy terms representing the Hamiltonian of the particles have each a different mass-ratio weighting. Following Nelson, if the second particle behaviour can be captured in a potential and the ingoing and outgoing velocities of the main particle are modelled using stochastic differential equations the motion of the main particle satisfies the Schr\"{o}dinger's equation. The diffusion variance of this equation is replaced by a related ratio of masses and the assumed variance. The first example attempts to reconcile this result with quantum mechanics by considering the Schr\"{o}dinger equation in the presence of only one type of incident particle. The main particle energy levels become multiples of the incident particle and the energy expression for the entire system agrees with quantum mechanics but there are differences with the stochastic equation. As the relationship is classical the Schr\"{o}dinger equation can be used to represent corrections for Newton's equation and also suggests a user profile to be used in the search for Dark Matter making use of the altered Newton's equation. An alternative solution to the collision model also shows relativistic properties as the interactions suggest corrections to the Minkowski equation in Einstein's Special Relativity. It is also possible to use the classical Schr\"{o}dinger's equation both on the main and incident particle simultaneously leading to a correlated set of wave equations with different diffusion parameters. In principle this set of equations may lead to Dirac's equation.
\end{abstract}

\maketitle

\section*{Introduction}
\smallskip
Ever since Einstein's introduction of the molecular-kinetic theory of heat in 1905 the Brownian motion/Markovian formalism has been applied to a large variety of topics including the existence of atoms,
the classical particle diffusion and stochastic mechanics. The original derivation of the associated diffusion equation for the density distribution is due to Einstein~\cite{EINST1} where he used stochastic properties to measure the size of atoms. The second subject is more the domain of statistical mechanics and focusses on thermodynamic properties,  Goldstein~\cite{GOLDSTEIN1}, isothermal flows, Garbaczewksi~\cite{GARBA1}, transport equations (e.g. Master equation, Boltzmann's equation or Kramer's equation) and Markov Chains, for instance Posilicano~\cite{POSILICANO1}, van Kampen~\cite{KAMPEN1} or Gamba~\cite{GAMBA1}. In the third topic of stochastic mechanics Nelson~\cite{ENELSON3} showed in 1966 how to apply a set of symmetric and anti-symmetric velocity terms and employ an acceleration operator to derive Schr\"{o}dinger's equation. This result led to the stochastic interpretation of quantum mechanics see for instance the recent review by Carlen~\cite{ECARL2} or Nelson~\cite{ENELSON1},~\cite{ENELSON2}.

The collision method described in this paper concentrates entirely on an elastic classical collision between the main (mass $M$) and incident particle (mass $m$) ignoring angular momentum or spin. The main result of the paper looks at the classical particle collision and looks at the axis that the collision employs. As there is no other process involved at the very moment of collision the interaction between the two particles happens via the one-dimensional interaction path $\phi$ defined by the unit vector connecting the center of the main particle and the incident particle. The three dimension motions (extendable to n-dimensions) of $v_1,w_1,v_2,w_2$ projected onto this axis by the project operator $P(\phi)$  determine the actual collision. All other motion from $(I-P(\phi))v_1,(I-P(\phi))w_1,(I-P(\phi))v_2,(I-P(\phi))w_2$ do not affect the classical collision. Once the collision has been determined resolving the equation writing $v_2, w_2$ in terms of $v_1,w_1$ the remaining elements are added and the final result is obtained. This procedure is entirely classical and applies to any particle motion $x(t)$.

With collisions the pre - and post velocities of the main particle $v_1,v_2$ and incident particles $w_1,w_2$ are modelled implicitly from position $x=x(t)$ by investigating the previous collision position $x(t-\tau_1)$ and the next collision coordinate at $x(t+\tau_2)$. At time $t-\tau_1$ the main particle velocity equals $v_1$ so that $x(t)=x(t-\tau_1)+v_1(x(t-\tau_1),t-\tau_1)\tau_1$ while the incident particle velocity equals $w_1$. Once the collision happened at time $t$ the new collision speeds are $v_2$ for the main article and $w_2$ for the incident particle hence the new main particle position becomes $x(t +\tau_2)=x(t)+v_2(x(t),t)\tau_2$. The $x(t)$ process shows the continuous progress of the main particle subject to equations of continuity whereas the behaviour of the incident particle $w_1, w_2$ can be obtained either from a heatbath or it can be obtained from an equation correlated to $x(t)$.  Here the post - collision time $\tau_2$ and the pre - collision time $\tau_1$ can be represented by a second (or higher) order gamma distribution so that $E[\tau_1]=E[\tau_2]=\overline{\tau}$ for a stationary process of collision $\tau$. If the inter-collision times decrease or increase the collision times become important in the calculations.

The easiest way of modelling $x(t)$ is to describe the path of the main particle as a solution to a stochastic differential equation and presume that the collisions only occur with a frequency distributed with a 2nd order (or higher) Gamma distribution. In this case the forward equation supplies the current position $x(t)$ and the velocity $v_2$ but with the current position $x(t)=x$ known the backward momentum $v_1(t)=x(t)-x(t-\tau_1)$ is a random variable that will have to be conditioned on Bayes' theorem to obtain $v_1$. That way the actual process can be modelled to some degree though not with the probability that the inter-collision time equals zero. In this paper we assume this often for the main distribution and model the incident particles as a potential, a constructed function and eventually as a stochastic process.

With this projection based collision representation consider the symmetric (average velocity) and anti-symmetric equations (osmotic velocity) of $v_2+v_1$, $w_2+w_1$, $v_2-v_1$ and $w_2-w_1$. With a fair bit of calculation it will be shown that the sum of the main particle $\abs{(v_2+v_1)}^2/4$, $\abs{(v_2-v_1)}^2/4$ and the incident particle $\abs{(w_2+w_1)}^2/4$, $\abs{(w_2-w_1)}^2/4$ can be made equal to the two particle Hamiltonian if all elements are weighted appropriately. The parts of the energy momentum squared $\abs{(v_2-v_1)}^2/4$ and $\abs{(w_2-w_1)}^2/4$ will be referred to as the osmotic velocity or momentum and they become zero if there is no diffusion (stochastic process) or become acceleration squared without diffusion. Notice that to represent the total energy of both systems all four terms in the equation are different emphasizing the effect of the interaction.

The main proof in the paper shows that the classical main particle and incident particle energy in the collision exchange of energy is invariant under the collision as long as the appropriate weighting has been applied to the four terms. The proof is in the back of the paper and uses energy integrals subjected to repeated partial integrals. The most interesting aspect of this fact is that Nelson~\cite{ENELSON1} showed the same argument for the main particle but he did not use the appropriate weighting for the individual momentum terms and he did not refer to energy and momentum conservation affecting the two particles. The paper shows that it is possible to condense the terms of the second particle into a potential (even a velocity dependent one if needed) so that the energy reduces to the first two symmetries of the main particle plus a potential for the incident particle one. With this approximation the energy equation almost look like Nelson's example used in the 1966 paper except for the weighting of the osmotic term.

First consider the case that incident particle is approximated with a potential. Then the solution to this energy equation conserving energy and momentum can employ stochastic differential equations and show that main particle position satisfies Schr\"{o}dinger's equation with a variance equal to $\eta = M^{3/2}\sigma^2/m^{1/2}$. If the incident particle has a finite energy $mc_w^2/2$ the solution of Schr\"{o}dinger's equation carries multiple incident particle energies. This is likely a first step towards quantum mechanics as the first particle now carries energy Eigen solutions that are multiples of $mc_w^2\overline{\tau}/2$ where $m$ is the incident particle mass and $c_w$ becomes the incident particle velocity squared and $\overline{\tau}$ shows the average time to collisions. Only proving Schr\"{o}dingers equation from the stochastic differential equation is not sufficient to explain quantum mechanics as Wallstrom~\cite{WALLSTROM1} pointed out.

The assignment of angular momentum and spin are often done as an external requirement to the Schr\"{o}dinger equation but in this case the angular momentum can be modelled independently in the energy term and so for that matter can the particle spin. Obviously using the present approach the energy model needs to add the angular momentum and spin components and needs to prescribe the both the angular momentum and spin state to energy transition. Once that has been done the solution to the energy equation will be a set of angular momentum and spin related Schr\"{o}dinger equations. In that case a discreteness for measurements for angular momenta and spin will be the consequence and the Wallstrom objection will be countered. Since the derivation of the energy expression is classical and the stochastic differential equation is only a choice of convenience it is clear that better solutions may be available. In particular binomial or multinomial transition processes that may explain Wallstrom's and Nelson's objections to the stochastic approach and also in the future explain the re-normalization problem.

Another example of the use of this energy equation is the application to galactic clusters to explain and modify Newton's equation due to diffusion and use this result to find a solution to Dark Matter. In this case the main particle is a typical solar system and the incident particles are the remaining stars in the galactic system. The force experienced by a typical star system is subject to two incident particle contributions. The first force replicates the Newtonian distribution as a function of the distance to the radial while the second force is provided as the osmotic term and depends on the diffusion constant. The diffusion here is the result of an elongation of star trajectories around the galactic centre due to the randomly provided forces caused by the large number of different star interactions. The solution is Schr\"{o}dinger's equation with a double potential and an unknown diffusion rate. Interestingly, the Schr\"{o}dinger equation with only the Newtonian potential have been studied by a variety of different people including Penrose~\cite{PENROSE1} who investigated the behaviour of particles as a function of their own gravity. The Newton potential enriched with its osmotic additional terms has not been studied yet. In addition, Milgrom~\cite{MIGROM1}, ~\cite{MILGROM} suggested to change the Newtonian function subject to radial distance of the galactic cluster which this example seems to suggest as well due to the diffusive forces in the osmotic terms.

The stochastic mechanics results are important for this paper as part of the examples suggests using the energy formula and substituting stochastic differential equations for the main particle trajectory. The forward and backward time step perspective for continuous stochastic processes can be gleaned from the extensive work on stochastic processes by Nelson~\cite{ENELSON1}, Carlen~\cite{ECARL1}, ~\cite{ECARL3}, Guerra~\cite{GUERRA1}, ~\cite{GUERRA2} and see the references in Carlen~\cite{ECARL2}. Carlen in particular emphasized that this construction for the examples would almost always have weak solutions as long as the potentials have the Kato-Rellich property. This energy invariant is particularly useful as it is possible to enter other stochastic processes like a binomial process, a set of trinomial trees or a multinomial process. Taking the forward and backward processes into account it is possible to use Feynman's method for generating solutions to the Dirac equation in a real setting rather than in complex space using the Abelian Gauge Theory based on Dirac's equation.

Another interesting result is based on the same set of velocity equations for $v_1,v_2$ and $w_1,w_2$ but solves the set of equations differently by looking at the Eigenvectors and Eigenvalues. Using this method it is quite easy to obtain the inputs and outputs simultaneously in the first Eigenvector as the average velocity field $a$ covering all input and output velocities. The second Eigenvector provides the individual velocity changes in terms of velocity differences $(v_1-w_1)$ and $(v_2-w_2)$. In this example it can be shown a formula similar to the Minkowski surface describing Einstein Special Relativity solutions. The example then shows that if the momenta $(v_1-w_1)$ and $(v_2-w_2)$ are independent of $a$ or if the incident particle consist of light (photon mass $m=0$) the original Minkowski geometry holds. For large masses where the mass of the main particle $M$ is much larger than the incident particle it becomes clear that the effect of fluctuations equals changes in overall average velocity $a$ is affected by $w_2$ and $w_1$. Though these energy conservation formula provide relativistic behaviour there is no reference to the speed of light only to the speed of the incident particles $w_1,w_2$.

The last example in this paper takes the energy conservation formula and substitutes one stochastic differential equation for the main particle with variance $\eta$ and a second stochastic process describing the incident particle using a much smaller variance $\eta_w=m^2\eta/M^2$. The example applies stochastic differential equations to both the main and incident process where the main and incident particles were modelled exchanging forces based on two potentials satisfying a condition. The double stochastic process is resolved into two Schr\"{o}dinger equations and the energy minimization technique demands that any potentials between the world of the main particle and the incident particle are symmetrically represented. In principle these equations describe the previous example and should accommodate the relativistic case as well for massless photons. If the incident particle do have mass the equations can only accommodate fluctuations that obey a slightly different requirement than the Minkowski surface. Further research on how this set of equations compares to Dirac's equation will be deferred to the future.

Section 1 investigates a two particle elastic collision process ignoring angular momentum or spin and calculates the pre - collision and post - collision kinetic energies and momentum of the main particle being diffused by the incident particle. To represent the elastic collision the collision process is assumed one-dimensional along the main and incident particle centres and uses the (random) projection operator to calculate the post - collision velocities. From the energies then a pre - collision and post - collision Hamiltonian is constructed. As the main and incident particle are isolated the energies and momentum for both particle pre - and post - collision Hamiltonian's are conserved. Any energy from the main particle is radiating energy into the heatbath or the main particle is absorbing energy from its heatbath surroundings.

The main result of the paper in Section 1 shows that for elastic collisions the velocities of the classical main and incident particle can be written in a specific Nelson measure for each particle. Specifically the combined energy of both particles is written as four terms with all terms having different weighting terms depending on the mass ratio between the main and incident particle. Hence the Nelson measure is intended to represent energy and momentum conservation for classical Hamiltonian collision problems. This result is very promising as Nelson showed how to convert these symmetrical and anti-symmetric expressions into Schr\"{o}dinger's equation with the use of stochastic processes. However, notice that the energy and momentum constraints demanded a specific weighting of the individual terms than Nelson originally assumed. The other result in this Section uses the same momenta equations but solves the system employing Eigenvectors and Eigenvalues. As mentioned above the first Eigenvector equals the average velocity field $a$ covering all input and output velocities and the second Eigenvector provides the individual velocity changes in terms of velocity differences $(v_1-w_1)$ and $(v_2-w_2)$.

Section 2 shows the solution of the particle using the solution to the generic stochastic differential equation employing both the usual forward equations and the backward equations. The two sets of equations produce that the probability density for the process and they show that the probability density satisfies the useful continuity equity. In addition the inter-collision times are distributed with a second order Gamma distribution and this Section shows the estimated pre - and post collision speeds together with the estimated drift term. In future calculations there will be constant energy terms that depend on the average inter-collision rates. The Section also shows how to use the differential equations to estimate the equations for functions of the main particle stochastic process.

Section 3 shows how the double particle energy equation can be solved assuming that the main particle can be represented via the forward stochastic differential equation describing the forward steps and backward process representing the backward step assuming the current position $x(t)=x$. To deal with the incident particles this part of the energy representation is presented as a potential representing the average energy of the incident particle. The energy terms now represent the average main particle velocity, the osmotic velocity weighted by the mass ratio, the incident particle potential and the constant energy term. The minimum solution for the main particle in the form of a stochastic process the only probability density that renders the double term pre - and post - collision energies for the main particle and the incident potential is the squared modulus of the wave function obtained from Schr\"{o}dinger's equation. However for time dependent potentials there is an additional term in the potential section of the Schr\"{o}dinger equation showing the energy dependence on the changing potential. Another difference with quantum mechanics is that the drift term of the associated stochastic equation depends on the mass ratio due to the mass ratio dependency and the fact that the variance of this equation only refers to the individual masses and the main particle volatility.

The first example in Section 4 shows an adjustment employs to adjust to Newton's law for our solar system using the Earth’s mass $M$ and comparing the diffusion result to Newton’s calculations of the Earth’s orbit around the sun. The stochastic differential of the Earth's orbit leads to the Schr\"{o}dinger equation for the position of the Earth subject to the potential of the sun. The variance here in the equation obviously relies on the sun's and Earth's mass ratio and reflects all uncertainties of mass inside the sun, uncertainties in the earth's orbit and gas deposits throughout the solar system. Another example using the same technique shows the Schr\"{o}dinger equation for the main and incident particle but where there is no interacting potential. In this case the Schr\"{o}dinger equation shows that the main particle Eigen values have Eigen values showing multiples of the assumed incident particle energy. This example agrees with quantum mechanics assuming that the environment is providing an environment of small incident particles and calling the diffusion rate proportional to $\hbar$. In the next example there is an unexpected result as it shows that the correlation between main particles and incident particles increases as the speed of the main particle increases.

Another example uses the incident particle energy distribution to model the attraction on a mass residing within a galaxy and notices that the potential forces within the galactic cluster or galaxy can be changed to compensate for the diffusion of forces. The solution to this case of Dark Matter employs the Schr\"{o}dinger equation to demonstrate how Newton's equation has to be adjusted due to the noticeable amount of diffusion in the system. Random star behaviour, position sizes, observational uncertainties and unexplained stellar gas causes uncertainty and the interaction potentials need to be modified to incorporate Newton's formula. The uncertainty causes small changes to the mass and velocity distribution and the uncertainty is best described using Schr\"{o}dinger's equation. The profile for the galactic radial speed of the star can be adjusted by the assumed variance. Though the radial velocity profiles will flatten as a result of this it is not clear whether this explains Dark Matter.

In the next example there is the relationship between the mean inter-particle collision time and the particle energy can be expressed in the form of the geometric Minkowski invariant which is now based on the average velocity. The original equation incorporates fluctuations due to impacted heatbath particles and is based on a correlation between the heatbath particles and the average velocity. Once the correlation is removed the original Minkowski minimal surface reappears. The reason that the original equation is of considerable interest is that it provides a perfect way of incorporating individual particle interaction which provides an introduction to relativistic mechanics and quantum mechanics.

Section 5 describes the two particle Hamiltonian elastic collision using the all four particle representation from Section 2 representing both the main particle and incident particle energy. Previously, the incident particle energy was approximated with a combined potential which forced the main particle position satisfying a Schr\"{o}dinger equation so it is logical to extend the same treatment to the incident particle. The last Theorem in the paper shows that both the main and incident particles can be represented by an intertwined set of Schr\"{o}dinger equations. In the two equations one potential is positive and the other is negative but the variances of the two equations are different. The two potentials might be connected as the first equation for the main particle finds a residue potential employing the probability of the incident particle while the incident particle equation depends on the main particle probability density. There is another energy condition required involving both potentials.

\section{Energy/Momentum Conservation}
This Section introduces classical elastic collisions for the main particle of mass $M$ and the incident particle $m$ where typically $(M>>m)$. The first result shows that the fully conserved energy of the collision can be expressed in four terms representing the Hamiltonian by two average speed functions and two osmotic terms. The second result shows that the velocity movement can be solved using two Eigenvalues and two Eigenvectors.

By assumption the particles exchange energy and momentum during the collision hence $v_2\neq v_1$ and $w_2 \neq w_1$ and it is assumed that there are no other processes involved.  For a purely elastic collision the momentum and the energy of the entire system during the collision must be conserved as the particles collide. The energy and momenta for the main particle $v_2,v_1\in\Re^n$ and the incident particle $w_1,w_2\in\Re^n$ have the terms
\begin{subequations}
\label{l:ENERG1P}
\begin{align}
\label{l:ENERG1a}
&\mathcal{P}_1=p_1+q_1=Mv_1+mw_1,
\\
\label{l:ENERG1b}
&\mathcal{P}_2=p_2+q_2=Mv_2+mw_2,
\\
\label{l:ENERG1c}
&\h_{1}=\frac{1}{2}\left(M\abs{v_2}^2+m\abs{w_2}^2\right),
\\
\label{l:ENERG1d}
&\h_{2}=\frac{1}{2}\left(M\abs{v_1}^2+m\abs{w_1}^2\right),
\end{align}
\end{subequations}
then the elastic collision momentum and energy conservation requires that $\mathcal{P}_2=\mathcal{P}_1$ and $\h_{2}=\h_{1}$.

In one dimension it is easy to turn the $v_2, w_2$ into a linear function of the inputs $v_1$ and $w_1$ but in $n$ dimensions there are $n-2$ additional arbitrary terms for the solution. However, following Figure \ref{COLL3} at collision time let $\phi$ be the random unit vector stretching from the center of the main particle to the center of the incident particle exchanging momentum and energy. As the picture suggests for elastic collisions the exchange of momentum and energy has to happen across the unit-vector $\phi$. All collision motion is exchanged along $\phi$ while all motion transfer in the plane orthogonal to $\phi$ is zero for the main and the incident particle. The matter of spin and angular momentum conservation will be addressed later.

\begin{figure}
\centering
\includegraphics[width=4in]{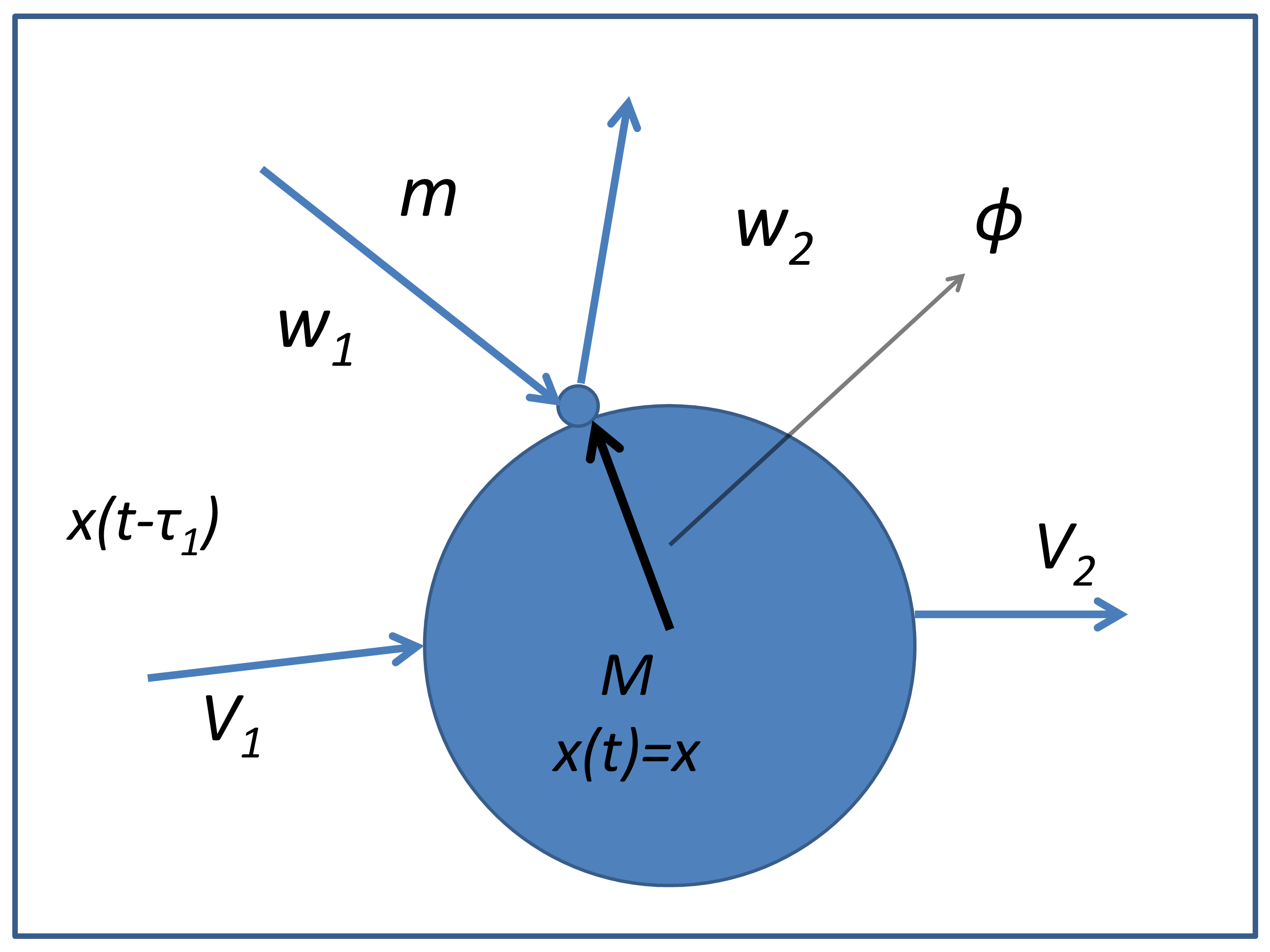}
\caption{Center of gravity line for elastic $M,m$ elastic collision}
\label{COLL3}
\end{figure}

To conveniently arrange for motion along a unit-size vector $\phi$ ($\abs{\phi}^2=\phi^T\phi=1$) consider the projection $P(\phi) = \phi\phi^T$ ($P(\phi)$ is an n-dimensional matrix) which maps all vectors onto the vector $\phi$. So $P(\phi)v_1 =\phi(\phi^Tv_1)$ is a vector project of $v_1$ on $\phi$ while $(I-P(\phi))v_1$ employing $I$ as the (n-dimensional) unit matrix is the remaining vector perpendicular to $\phi$. Clearly $P(\phi)$ is a projection because
\begin{align*}
P(\phi)P(\phi)=(\phi\phi^T)(\phi\phi^T)=\phi(\phi^T\phi)\phi^T=
\phi\phi^T=P(\phi),
\end{align*}
since $\phi$ is a unit vector. Also
\begin{align*}
P(\phi)(1-P(\phi))=(P(\phi)-P(\phi)^2) = (P(\phi)-P(\phi)) = 0,
\end{align*}
so operator $(1-P(\phi))$ is always orthogonal to $P(\phi)$ i.e. $(1-P(\phi))v\perp P(\phi)w$ for all vectors $v,w$.

To implement the conservation law along the $\phi$ axis decompose the $v_1,v_2,w_1$ and $w_2$  vectors as follows
\begin{align*}
v_1&= P(\phi)v_1 + (I-P(\phi))v_1,
\\
v_2&= P(\phi)v_2 + (I-P(\phi))v_2,
\\
w_1&= P(\phi)w_1 + (I-P(\phi))w_1,
\\
w_2&= P(\phi)w_2 + (I-P(\phi))w_2,
\end{align*}
then there is no momentum or energy change in the direction orthogonal to $\phi$ so clearly
\begin{align*}
(1-P(\phi))v_2 & = (1-P(\phi))v_1,
\\
(1-P(\phi))w_2 & = (1-P(\phi))w_1,
\end{align*}
while along $\phi$ the conservation laws for momentum $\mathcal{P}_1, \mathcal{P}_2$ and energy $\h_1, \h_2$ imply
\begin{align*}
MP(\phi)v_1 + mP(\phi))w_1 &= MP(\phi)v_2 + mP(\phi))w_2,
\\
Mv_1^TP(\phi)v_1 + mw_1^TP(\phi))w_1 &= Mv_2^TP(\phi)v_2 + mw_2^TP(\phi))w_2.
\end{align*}

The last equation implies
\begin{align*}
M\phi^Tv_1 + m\phi^Tw_1 &= M\phi^Tv_2 + m\phi^Tw_2,
\\
M(\phi^Tv_1)^2 + m(\phi^Tw_1)^2 &= M(\phi^Tv_2)^2 + m(\phi^Tw_2)^2,
\end{align*}
making use of the $P(\phi)$ projection property again by removing $\phi$ from the equation while the
second equation was created by substitution. Now the conservation problem is one-dimensional and as the Proposition \eqref{l:EQMOT1} shows the solution equals
\begin{align}
\label{App14}
\begin{pmatrix}
\phi^Tv_2 \\
\phi^Tw_2
\end{pmatrix}
&= \begin{pmatrix}
\cos(\theta) & \gamma_m\sin(\theta) \\
\frac{\sin(\theta)}{\gamma_m} & -\cos(\theta)
\end{pmatrix}
\begin{pmatrix}
\phi^Tv_1\\
\phi^Tw_1
\end{pmatrix},
\end{align}
where $\gamma_m^2 = m/M$, $\sin\left(\theta\right)
=2\gamma_m/\left(1+\gamma_m^2\right)$,
$\cos\left(\theta\right)=\left(1-\gamma_m^2\right)/\left(1+\gamma_m^2\right)$.

Notice again that there is no angular momentum modelled here. In principle each particle can have a three-dimensional angular momentum vector that upon collision can change interacting with the linear momentum and energy. Angular momentum assumes that the particles can twist or turn and the collision transfers this to change the post collision energy and momentum particles. Without the angular momentum Appendix A shows that \eqref{App14} can be used to obtain the post-velocities as follows.
\begin{thm}
\label{l:THEOREM3}
In n-dimensions ignoring angular momentum for an elastic collision it must be that $\h_1=\h_2$ and $\mathcal{P}_1=\mathcal{P}_2$ conserving energy and momentum between $v_2$ and $w_2$. Then for the main and incident input velocities $v_1,v_2$ and $w_1,w_2$ it follows that
\begin{align}
\label{l:EQMOT9}
\begin{split}
\begin{pmatrix}
v_2 \\
w_2
\end{pmatrix}
&= \begin{pmatrix}
\cos(\theta) & \gamma_m\sin(\theta) \\
\frac{\sin(\theta)}{\gamma_m} & -\cos(\theta)
\end{pmatrix}
\begin{pmatrix}
v_1\\
w_1
\end{pmatrix}
+
\begin{pmatrix}
\gamma_m \Phi
\\
-\frac{1}{\gamma_m} \Phi
\end{pmatrix}
\\
&=\begin{pmatrix}
I-\gamma_m\sin(\theta)P(\phi) & \gamma_m\sin(\theta)P(\phi) \\
\frac{\sin(\theta)}{\gamma_m}P(\phi) & I-\frac{\sin(\theta)}{\gamma_m}P(\phi)
\end{pmatrix}
\begin{pmatrix}
v_1
\\
w_1
\end{pmatrix},
\end{split}
\end{align}
where $I$ is the (n-dimensional) unit matrix and $\Phi = \sin(\theta) (1-P(\phi))(v_1-w_1)$ is an n-dimensional vector so that $\Phi^T(v_2-v_1)=0$. Here $P(\phi)$ is the usual random projection in the direction of $\phi$ with $\gamma_m^2 = m/M$, $\sin\left(\theta\right) =2\gamma_m/\left(1+\gamma_m^2\right)$, $\cos\left(\theta\right)=\left(1-\gamma_m^2\right)/\left(1+\gamma_m^2\right)$.
\end{thm}
\begin{proof}
For a detailed proof see Appendix \eqref{l:AppendixA}. The brute force proof of equation \eqref{l:EQMOT9} is substituting the expression $\Phi = \sin(\theta)(1-P(\phi))(w_1-v_1)$ into equation \eqref{l:EQMOT9} and then checking the equation straightforwardly. In the proof in Proposition \eqref{l:PROP1} the collision along $\phi$ is established first which is straightforward as this collision is one-dimensional along $\phi$. The remaining orthogonal motions $(I-P(\phi))v_2, (I-P(\phi))v_1, (I-P(\phi))w_2$ and $(I-P(\phi))w_1$ are assumed to be unaltered by the collision as they are parallel to the collision axis $\phi$. Eventually these orthogonal velocities are then substituted into the one-dimensional result to derive \eqref{l:EQMOT9}.
\end{proof}

Notice that the terms $w_1$ and $w_2$ still have a direct random relationship to $v_1$ and $v_2$ as specified by equation \eqref{l:EQMOT9} in Theorem \eqref{l:THEOREM3}. In fact from this equation it is clear that
\begin{align*}
v_2-v_1&=(\cos(\theta)-1)v_1+\gamma_m\sin(\theta)w_1+\gamma_m\Phi
=\gamma_m\sin(\theta)(w_1-v_1)+\gamma_m\Phi,
\\
w_2-w_1&=\left(\frac{\sin(\theta)}{\gamma_m}-1\right)v_1-(\cos(\theta)+1)w_1-\frac{1}{\gamma_m}\Phi
=\frac{\sin(\theta)}{\gamma_m}(v_1-w_1)-\frac{1}{\gamma_m}\Phi,
\end{align*}
so that the momentum conservation emerges $w_2-w_1 = -\frac{v_2 - v_1}{\gamma_m^2}$. Similarly from \eqref{l:EQMOT9} it is clear that
\begin{align*}
v_2+v_1&=(\cos(\theta)+1)v_1+\gamma_m\sin(\theta)w_1+\gamma_m\Phi
=\frac{\sin(\theta)}{\gamma_m}(v_1+\gamma_m^2w_1)+\gamma_m\Phi,
\\
w_2+w_1&=\frac{\sin(\theta)}{\gamma_m}v_1-(\cos(\theta)-1)w_1-\frac{1}{\gamma_m}\Phi
=\frac{\sin(\theta)}{\gamma_m}(v_1+\gamma_m^2w_1)-\frac{1}{\gamma_m}\Phi,
\end{align*}
so that
\begin{align*}
w_2+w_1 & = v_2 + v_1  - \frac{2}{\sin(\theta)}\Phi,
\end{align*}
with $\Phi$ is defined in \eqref{l:EQMOT9}. This equation shows that $v_1,v_2$ will depend on $P(\phi)$ via $\Phi$.

Using the projection $P(\phi)$ for elastic collisions the linear relationship between $v_1$, $w_1$ and $v_2$, $w_2$ due to energy conservation $\h_2=\h_1$ and momentum conservation $\mathcal{P}_2=\mathcal{P}_1$ can be represented as
follows.
\begin{prop}
\label{l:THEOREM4}
Using Theorem \eqref{l:THEOREM3} the elastic collision solution to the four equations in \eqref{l:ENERG1P} can also be represented as
\begin{align}
\label{l:EQMOT1P}
\begin{split}
\begin{pmatrix}
v_2 \\
w_2
\end{pmatrix}
&=
\begin{pmatrix}
v_1 \\
w_1
\end{pmatrix}
+
\begin{pmatrix}
\gamma_m \sin(\theta)P(\phi)(w_1-v_1) \\
- \frac{1}{\gamma_m}\sin(\theta)P(\phi)(w_1-v_1)
\end{pmatrix}
\\
&=
\begin{pmatrix}
v_1 \\
w_1
\end{pmatrix}
+
\begin{pmatrix}
\gamma_m\left(\sin(\theta)(w_1-v_1)-\Phi\right) \\
-\frac{1}{\gamma_m}\left(\sin(\theta)(w_1-v_1)-\Phi\right))
\end{pmatrix},
\end{split}
\end{align}
with $\phi$ showing the (random, unit-size) connection between the particles and $\gamma_m^2=m/M$,
$\sin\left(\theta\right) = 2\gamma_m/\left(1+\gamma_m^2\right)$,
$\cos\left(\theta\right)=\left(1-\gamma_m^2\right)/\left(1+\gamma_m^2\right)$.

\begin{proof}
For a detailed proof see Appendix $\ref{l:AppendixC}$. The first equation is just a rewrite of the second term in \eqref{l:EQMOT9} while the second term is created from the first term using the definition $\Phi$.
\end{proof}
\end{prop}

The results are suggestive as it indicates the motion of the particle. Looking at equation \eqref{l:EQMOT1P} it is possible to take an expectation in order to derive behavior. For notation assume that $E[.\vert x]$ is the conditional expectation over the backward and forward fluctuations ($d^+ z(t), d^- z(t)$) and the motion of the projection $P(\phi)$ leaving the position of the particle in $x(t) = x$. Presumably the vector $\phi$ is not related to the forward and backward fluctuations.

Take the conditional expectation $E[. \vert x]=E[. \vert x(t)=x]$ over equation \eqref{l:EQMOT1P} then
\begin{align}
\label{l:EQMOT3}
\begin{split}
\begin{pmatrix}
E[v_2 \vert x] \\
E[w_2\vert x]
\end{pmatrix}
&=
\begin{pmatrix}
E[v_1\vert x] \\
E[w_1\vert x]
\end{pmatrix}
+
\begin{pmatrix}
\gamma_m \sin\left(\theta\right) E\left[P(\phi)\right](E[w_1\vert x]-E[v_1\vert x]) \\
- \frac{1}{\gamma_m}\sin\left(\theta\right)E\left[P(\phi)\right](E[w_1\vert x]-E[v_1\vert x])
\end{pmatrix}
\\
&\approx
\begin{pmatrix}
E[v_1\vert x] \\
E[w_1\vert x]
\end{pmatrix}
+
\begin{pmatrix}
\gamma_m \sin\left(\theta\right) (E[w_1\vert x]-E[v_1\vert x]) \\
- \frac{1}{\gamma_m}\sin\left(\theta\right)(E[w_1\vert x]-E[v_1\vert x])
\end{pmatrix}
\approx
\begin{pmatrix}
E[v_1\vert x] \\
E[w_1\vert x]
\end{pmatrix}
\end{split}
\end{align}
if and only if $E[v_1\vert x]=E[w_1\vert x]$ and $E[P(\phi)]=I$. If the mean $E[v_1\vert x]$ is larger than $E[w_1\vert x]$ equation $\eqref{l:EQMOT3}$ will slow it down. If the average motion is not as large as $E[v_1\vert x]$, the particle $M$ is speeded up. Though the motion of $w_1$ is random after a large number of collisions the mean of the main particle will adjust itself to the mean of the incident particle. The behavior of $w_1,w_2$ convergence depends on the standard deviation as provided by $P(\phi)$ and $\gamma_m^2 = m/M$.

The simple equation \eqref{l:EQMOT9} can be solved easily once $P(\phi)$ is established but it is
possible to find an alternative representation.

\begin{thm}
\label{l:THEOREM2}
The requirements on \eqref{l:ENERG1a},\eqref{l:ENERG1b},\eqref{l:ENERG1c} and \eqref{l:ENERG1d} in the form $\h_1(x)=\h_2(x)$ and $\mathcal{P}_1(x)=\mathcal{P}_2(x)$ creates equation \eqref{l:EQMOT9} which can be solved implicitly as
\begin{align}
\label{l:EQMOT4}
\begin{split}
v_1 & =a-\gamma_m^2b,
\\
w_1 & =a+b,
\\
v_2 & =a+\gamma_m^2b^\bot,
\\
w_2 & =a-b^\bot,
\end{split}
\end{align}
where $b^\bot=b + \frac{1}{\gamma_m}\Phi$ and $\abs{b}^2 = \abs{b^\bot}^2$. Here the vector $a$  becomes the average system velocity while $b$ and $b^\bot$ constitute the interactions between the main and incident particles. The vectors $b, b^\bot$ and $a$ are n-dimensional. The only restriction here is that the size of $b$ and $b^\bot$ vectors is conserved in the elastic collision.

\begin{proof}
The details of this Theorem can be found in Appendix \eqref{l:AppendixB}. Equation  $\eqref{l:EQMOT9}$ is linear and can be solved using eigenvalues and eigenvectors of the collision matrix. The energy conservation condition then specifies that $b$ and $b^\bot$ have identical sizes.
\end{proof}
\end{thm}

Notice from equations \eqref{l:EQMOT4} it follows that
\begin{align}
\label{l:EQMOT5}
\begin{split}
v_1-w_1 &= -(1+\gamma_m^2)b,
\\
v_1+\gamma_m^2w_1 &= (1+\gamma_m^2)a,
\end{split}
\end{align}
and similarly that
\begin{align}
\label{l:EQMOT6}
\begin{split}
v_2-w_2 &= -(1+\gamma_m^2)b^\bot,
\\
v_2+\gamma_m^2w_2 &= (1+\gamma_m^2)a,
\end{split}
\end{align}
and by Appendix \eqref{l:AppendixB} it is clear that $\abs{b}^2 = \abs{b^\bot}^2$. As a result
\begin{align*}
\begin{split}
\begin{split}
\abs{v_1-w_1}^2 &= \abs{v_2-w_2}^2,
\\
v_2+\gamma_m^2w_2 &= v_1+\gamma_m^2w_1,
\end{split}
\end{split}
\end{align*}
or
\begin{align*}
\begin{split}
\begin{split}
\abs{v_1-w_1}^2 &= \abs{v_2-w_2}^2,
\\
a = \frac{Mv_2+mw_2}{M+m} &= \frac{Mv_1+mw_1}{M+m}.
\end{split}
\end{split}
\end{align*}
The second equation shows the conservation of the mean momentum $a$ through the collision while the first equation shows that the noise terms $b, b^\bot$ are conserved through the term
$(1+\gamma_m^2)^2\abs{b}^2=\abs{w_1-v_1}^2=\abs{w_2-v_2}^2=(1+\gamma_m^2)^2\abs{b^\bot}^2$.

This following Theorem is the main result of this paper describing how the energy and momentum conservation for the motion for the heatbath (incident) particle $m, w_1$ (post-collision $m,w_2$) and the main particle $M, v_1$ (post-collision $M,v_2$) can be expressed exclusively in terms of $v_2$ and $v_1$ not using any non-linear terms depending on $P(\phi)$.

\begin{thm}
\label{l:THEOREM5}
Let the pre- and post collision momentum for the two particles equal $\mathcal{P}_1(x)=Mv_1+mw_1$ and $\mathcal{P}_2(x)=Mv_2+mw_2$ and let the pre- and post collision combined energies be $\h_1(x)=\frac{M}{2}\abs{v_1}^2+\frac{m}{2}\abs{w_1}^2$,
$\h_2(x)=\frac{M}{2}\abs{v_2}^2+\frac{m}{2}\abs{w_2}^2$. During the elastic collision the pre-collision momentum equals the post-collision momentum $\mathcal{P}_2(x)=\mathcal{P}_1(x)$ and the pre-collision energy equals the post-collision energy $\h_1(x)=\h_2(x)$. It then follows that the conserved energy of both the main particle and the incident particle combined is the sum of both particles energy expressed in the Nelson measure so that
\begin{align}
\label{l:DRFT4}
\begin{split}
&\frac{2\h_2(x)}{M}=\frac{2\h_1(x)}{M}=
\abs{\frac{v_2+v_1}{2}}^2+\frac{1}{\gamma_m^2}
\abs{\frac{v_2-v_1}{2}}^2 + \gamma_m^2\abs{\frac{w_2+w_1}{2}}^2+
\gamma_m^4\abs{\frac{w_2-w_1}{2}}^2,
\end{split}
\end{align}
with $\gamma_m^2 = m/M$.

\begin{proof}
See Appendix \eqref{l:AppendixE} for a detailed proof. The derivation uses the properties of Theorem
\eqref{l:THEOREM3} using the properties of the vector $\Phi$ as specified in the same Theorem.
\end{proof}
\end{thm}

This expression has many interesting advantages as it expresses the classical energy of the two elastic colliding particles in terms of the Nelson measure applied to the main particle and the incident particle. The result is derived for the Hamiltonian so any results are both classical and quantum mechanical. But it is not easy to see that \eqref{l:DRFT4} can be extended to multiple immediate interactions or can be put in the form of a field if required. Presumably if there are an infinite number of particle additions the overall energy becomes a field and the energy contributions become ever smaller by design to guarantee a finite limit. If this design succeeds the re-normalization problem in quantum mechanics may be successfully addressed. In 1966 Nelson~\cite{ENELSON3} used the first two (unweighed) terms to construct a similar Hamiltonian for the underlying main particle and derived the Schr\"{o}dinger's wave equation employing stochastic processes. He noticed that the one particle Hamiltonian plus potential was satisfied but he did not mention the energy conservation of the main and incident particle.

In much of our analysis Gaussian processes will be employed but the actual form of the velocity distribution is irrelevant both for the main and independent particle so there is no objection using other velocity distributions or employing different probability transfer models. The energy terms in \eqref{l:DRFT4} conveniently separates the first term representing the ongoing motion of the main and incident particle from the weighted second term showing the dependence on the interaction. The average velocity terms relate to the main and incident particle emphasizing the average motion while the difference terms are the osmotic velocity with different weighting. Equation \eqref{l:DRFT4} describes both the averaged velocities $(v_1+v_2)/2$, $(w_1+w_2)/2$ and the osmotic diffusion terms $(v_1-v_2)/2$, $(w_1-w_2)/2$ at the same time. Notice that the weighting of the terms allow different emphasis to both terms. The main term has a large weighting on the dispersal term $v_2-v_1$ suggesting that the energy is concentrated in the motion of the main particle and experiences very little diffusion if $m << M$. Conversely for the incident particle the main weighting is on the motion $w_1+w_2$ of the particle assuming a small amount of translational energy weighted by $\gamma_m^2$ showing that its contribution is bigger. The osmotic energy of the incident particle by $\gamma_m^4$ shows that its contribution is quite large for smaller $m<<M$.

The main particle $M$ collides at time $t$ with the incident particle and then moves on the next collision at time $t+\tau_2$. So between the collision between $M$ and the incident particle $m$ at time $t$ and $t+\tau_2$ the probability density $\rho=\rho(x,t)$ applies and when the next collision occurs this probability density is the initial condition for the new wave function for the main particle. However, the new incident particle may have a different set of initial conditions for a heatbath or its initial condition could have been provided by a stochastic process. In the case of photons moving in an empty container it may be true that the initial conditions prevail over many meters or in deep space they may prevail over many parsecs. But inside a heatbath the main - incident particle collisions are very frequent and the initial conditions for the incident particle will be changed continuously.

The equation can be applied to classical mechanics, electro-dynamics, quantum mechanics and gravity. A convenient form of the energy of the main particle is to imagine that the third and fourth terms in the energy for the incident particle can be represented as a potential. This conveniently creates a two term main particle energy and a conventional potential distribution. Applying this to Newton in the application show changes to the Newtonian law within the solar system due to the diffusion terms. In another example below shows the dispersion of mass assuming that the incident particle shows the Newtonian diffusion drift inside galactic clusters. This case shows how the effective distribution of the star density round our galaxy can be influenced by the random behaviour of the gravity field. The diffusion used here relates to the interaction of the mass diffusing with its environment and does not relate to gravitational waves.

\section{Diffusion and Energy}
This Section represents the main behaviour of the main particle in terms of velocities, momenta, energy modelling employing stochastic processes while the incident particles contributions are represented in the form of a potential. As the main particle $M$ at $x(t)$ driven by Gaussian processes moves through the medium there are collisions between the main and incident particle $m$ at collision times $t\in\set{t_j,j=0,1,...}$ and at points for the main particle $x(t)\in\set{x(t_j),j=0,1,...}$. With this assumption the inter-collision times are defined and a distribution for the inter-collision times is provided.

Assume that the inter-collision times $\tau_1,\tau_2$ are modeled as independent random variables satisfying a second order Gamma distribution then the mean time to the next distribution equals $E[\tau_2]=E[\tau_1]=\overline{\tau}=2/\beta$ with variance $var(t_{j+1}-t_j)=2/\beta^2$. This justifies the pre - and post velocities of the main and incident particles. Let $b^+=b^+(x,t), b^-=b^-(x,t)$ as the forward and backward step for the position of $M$ employing the n-dimensional stochastic process $x(t)\in\Re^n$ and assumes that at collision time $t$ the process is approximated as
\begin{subequations}
\begin{align}
\label{l:DRFTDEF0a}
&d^+x(t)=x(t+\tau_2)-x(t)=
b^+\tau_2+\sigma d^+z(t) = b^+(x(t),t)\tau_2+\sigma d^+z(t),
\\
\label{l:DRFTDEF0b}
&d^-x(t)=x(t)-x(t-\tau_1)=
b^-\tau_1+\sigma d^+z(t)=b^-(x(t),t)\tau_1+\sigma d^-z(t),
\end{align}
\end{subequations} where $d^+z(t),d^-z(t)$ are the independent n-dimensional Gaussian process $z(t)$ so that $d^+z(t) = z(t+\tau_2) -z(t)$ and $d^-z(t) = z(t) - z(t-\tau_1)$. From this process the forward and backward velocities $v_2(t)\in\Re^n$ and $v_1(t)\in\Re^n$ can be approximated as
\begin{gather}
\label{l:DRFTDEF1}
\begin{split}
v_2\left(x(t),t\right)=
\frac{x(t+\tau_2)-x(t)}{\tau_2}=
b^++\frac{1}{\tau_2}\sigma d^+z(t)=b^+(x(t),t)+\frac{1}{\tau_2}\sigma d^+z(t),
\\
v_1\left(x(t),t\right)=
\frac{x(t)-x(t-\tau_1)}{\tau_1}=
b^-+\frac{1}{\tau_1}\sigma d^-z(t)=b^-(x(t),t)+\frac{1}{\tau_1}\sigma d^-z(t),
\end{split}
\end{gather}
which is now properly defined. Usually $n$ is chosen as equal to three most of the time but an n-dimensional representation provides a mild extension. Imagine that the one dimensional time dimension to the next collision has the Gamma density $f_\beta (t)=\beta^2 te^{-\beta t},t \geq 0$. Then the distribution for the particle speed is defined as $d^+z(t)/\tau_2$, $d^-z(t)/\tau_1$ which are proper random variables with a distribution
$f_{\frac{d^+z(t)}{\tau_2}}(v), f_{\frac{d^-z(t)}{\tau_2}}(v)\sim \left(\beta+\frac{v^2}{2}\right)^{-5/2}$.

As it turns out the stochastic differential equations (infinitesimal time steps) \eqref{l:DRFTDEF0a} have many strong and weak solutions but Nelson~\cite{ENELSON1} and Carlen~\cite{ECARL1} showed that the backward increment \eqref{l:DRFTDEF0b} for a continuous stochastic process almost always provides a similar solution to the solution of \eqref{l:DRFTDEF0a}. The stochastic equation is a time-reversed Markovian process with a derived drift and Gaussian increment. Specifically the following applies.

\begin{thm}
\label{l:TIMEERG2}
Let $\rho=\rho(x,t)$ be the n-dimensional distribution for finding the main particle at position $x$ at time $t$ then for a diffusion process $x(t)\in\Re^n$ the forward motion of the particle $x(t+\tau_2)-x(t)$ and the backward step $x(t)- x(t-\tau_1)$ conditional on $x(t)=x$ are given by
\begin{subequations}
\begin{align}
\label{l:DRFT2aa}
d^+ x = x(t+\tau_2)-x(t)&=b^+\tau_2+\sigma d^+z(t),
\\
\label{l:DRFT2ab}
\begin{split}
d^- x = x(t)-x(t-\tau_1) & = b^-\tau_1+\sigma d^-z(t),
\\
b^-&=b^+- \sigma^2\nabla
\log{\rho},
\end{split}
\end{align}
The shocks $d^+z(t)$ and $d^-z$(t) are independent Gaussian increments. The probability density is defined as the solution of
\begin{align}
\label{l:DRFT3aa}
\frac{\partial \rho}{\partial t}&= - \nabla^T\left(b^+\rho\right)
+\frac{\sigma^2}{2}\Delta\rho,
\\
\label{l:DRFT3ab}
\frac{\partial \rho}{\partial t}&= -\nabla^T\left(b^-\rho\right)
- \frac{\sigma^2}{2}\Delta\rho,
\end{align}
so that
\begin{align}
\label{l:DRFT3c}
\frac{\partial \rho}{\partial t}= -
\nabla^T\left(\left(\frac{b^++b^-}{2}\right)\rho\right),
\end{align}
which is the continuity equation. Here $\nabla^T=\left(\frac{\partial}{\partial
x_1},...,\frac{\partial}{\partial x_n}\right)$ and $\Delta=\sum_{i} \frac{\partial^2}{\partial x_i^2}$.
\end{subequations}
\end{thm}

\begin{proof}
Carlen ~\cite{ECARL1} shows a large number of solutions for the case where forward drift can be associated with specific types of potentials.
\end{proof}

Using backward and forward representation of the motion to prove the approximations \eqref{l:DRFTDEF1} of the main particle using the distribution to the next collision the following is now straightforward.
\begin{lem}
\label{l:ENERG8}
Let the n-dimensional position process $x(t)$ and the inter-collision times be second order Gamma
distributed. Then at the collision time $t$ the velocities can be measured as follows
\begin{align}
\label{l:DRFT6}
\begin{split}
E[v_2(t)\vert x(t)] & = b^+, E[v_1(t)\vert x(t)]=b^-
\\
\frac{M}{2}E[\abs{v_2(t)}^2\vert x(t)] & = \frac{M}{2}\abs{b^+}^2+ \frac{nM\sigma^2}{\overline{\tau}}
\\
\frac{M}{2}E[\abs{v_1(t)}^2\vert x(t)] & = \frac{M}{2}\abs{b^-}^2+\frac{nM\sigma^2}{\overline{\tau}}
\end{split}
\end{align}
where $\overline{\tau}=2/\beta$ is the mean
inter-particle collision time.

\begin{proof}
The expectations are straightforward as the n-dimensional terms on
the righthand side of the second and third equation in
\eqref{l:DRFT2ab} are finite curtesy of the fact that
\begin{gather*}
E\left[\abs{\frac{\Delta^+z}{\tau_2}}^2\right]=
nE\left[\frac{1}{\tau_2}\right]=\frac{2n}{\overline{\tau}},
\\
E\left[\abs{\frac{\Delta^-z}{\tau_1}}^2\right]=
nE\left[\frac{1}{\tau_1}\right]=\frac{2n}{\overline{\tau}}.
\end{gather*}
The equality in the first equation in (\ref{l:DRFT6})
is due to the fact that
\begin{gather*}
\int_{-\infty}^{\infty}(b^+-b^-)\rho dx=
\int_{-\infty}^{\infty}\left(\frac{\rho_x}{\rho}\right)\rho dx=0,
\end{gather*}
so that the average velocity of the main particle is the same
whether a forward or backward view is developed.
\end{proof}
\end{lem}

This theorem involving the backward equation \eqref{l:DRFT3ab} and its relationship to the forward eqaution can be extended to properly differentiable functions $g=g(x(t),t)$ as well so
\begin{thm}
Assume that $g=g(x(t),t)$ is a sufficiently differential function (n-dimensional) then using Theorem
\eqref{l:TIMEERG2} and Ito's theorem the function $g=g(x(t),t)$ with $x(t)\in\Re^n$ equals
\begin{subequations}
\begin{align}
\label{l:DRFT4aa}
\begin{split}
d^+ g & = g(x(t+\tau_2),t+\tau_2)-g(x(t),t) = \left(g_t + (\nabla g)^Tb^++\frac{\sigma^2}{2}\Delta g\right)\tau_2 + \sigma(\nabla g)^T d^+ z(t)
\\
& = \left(g_t + Dg\right)\tau_2 + \sigma(\nabla g)^T d^+ z(t)
\end{split}
\\
\begin{split}
d^- g & = g(x(t),t)-g(x(t-\tau_1),t-\tau_1) = \left(g_t + (\nabla g)^Tb^--\frac{\sigma^2}{2}\Delta g\right)\tau_1 + \sigma(\nabla g)^T d^- z(t)
\\
& = \left(g_t + D^*g\right)\tau_1 + \sigma(\nabla g)^T d^- z(t),
\end{split}
\end{align}
where $\rho(x,t)$ is the probability density for finding the main particle at position $x$ at time $t$ and $Dg=(\nabla g)^Tb^++\frac{\sigma^2}{2}\Delta g$, $D^*g=(\nabla g)^Tb^--\frac{\sigma^2}{2}\Delta g$. Again the shocks $d^+z(t)$ and $d^-z(t)$ are independent Gaussian increments and
$\nabla=\left(\frac{\partial}{\partial x_1},...,\frac{\partial}{\partial x_n}\right)$ and
$\Delta=\sum_{i}\frac{\partial^2}{\partial x_i^2}$.
\end{subequations}
\end{thm}

\begin{proof}
The results from Nelson ~\cite{ENELSON1} and Carlen ~\cite{ECARL1} are for where $\tau_1, \tau_2$ approaches zero so that \eqref{l:DRFT2aa} becomes a stochastic differential equation. If that is true equation \eqref{l:DRFT2ab} is also a stochastic differential equation and has a solution. Notice that the expressions for $Dg$ and $D^*g$ are the result of applying Ito's lemma.
\end{proof}

Notice that the discrete collision approximation \eqref{l:DRFT2aa} allows for a pre - and post - collision momentum and energy \eqref{l:DRFT6}. If the stochastic process is applied to molecular applications at room temperatures then $\abs{b^-}^2 >> n\sigma^2/\overline{\tau}$ while $\abs{b^+}^2 \approx n\sigma^2/\overline{\tau}$ only at lower temperatures. However in other applications like astrophysics or economics it remains to be seen what the drift energy is in proportion to the diffusion energy. In the next Section the momentum solution for the collision is provided and the minimum energy solution for both particle is found.

\section{The Schr\"{o}dinger equation}

This Section substitutes the assumptions of the stochastic equations introduced in the previous Section into the energy expression for the main and incident particle. This Section derives the Schr\"{o}dinger equation for the main particle to simplify the terms in \eqref{l:DRFT4} assuming that the movement $w_1, w_2$ of the incident particle can be cast into a potential.

Then define $\gamma_m^2f_{approx} \approx \gamma_m^2\abs{w_2+w_1}^2 /4 \approx \gamma_m^2 \abs{\frac{dw}{dt}}^2$ which is an independent non-random energy term (divided by $M$) possibly depending on time and position $x$. First define the approximate potential
\begin{align}
\label{l:DRFT18}
f_{approx}=f_{approx}(x,t)= E\bigg[\abs{\frac{w_2+w_1}{2}}^2+\gamma_m^2\abs{\frac{w_2-w_1}{2}}^2
\bigg|x(t)=x\bigg],
\end{align}
where the conditional expectation $E[.\vert x]$ averages over the $w_2,w_1$ fluctuations (changes in the state change) but not the state $x$ itself. The expression in Theorem \eqref{l:THEOREM5} can now be approximated as follows.

\begin{cor}
\label{CORR1}
Using an expected value for \eqref{l:DRFT4} the total energy $\h_1$ or $\h_2$ can now be written as
\begin{align}
\label{l:DRFT13}
\h_1 = \h_2 = \frac{M}{2} \left( \abs{\frac{v_2+v_1}{2}}^2+\frac{1}{\gamma_m^2}
\abs{\frac{v_2-v_1}{2}}^2\right) + V,
\end{align}
where $V=V(x,t)$ is a potential defined as
\begin{align}
\label{l:DRFT16}
V=V(x,t)=\frac{M\gamma_m^2}{2} f_{approx}.
\end{align}
The potential $V$ ignores the actual behaviour and osmotic behaviour of the incident particle $w_1,w_2$ by substituting a straightforward function.
\end{cor}

This is the form of the Nelson measure adopted by Edward Nelson in 1966 in ~\cite{ENELSON1} except for the mass weighting on the second term in the main particle. Notice that the potential $V$ does not depend on the mass $M$ explicitly as $M\gamma_m^2=m$ while only the random terms depend on $w_2$, $w_1$ and $\gamma_m$. The agenda for further developments is now to employ the collision representation of the Markovian process describing the path of the main particle in equation \eqref{l:DRFTDEF1} and substitute the respective forward and backward velocities into \eqref{l:DRFT13} using the energy term $V$. The total energy as a function of time and position using the forward and backward random processes can be investigated.

Using Theorem \eqref{l:THEOREM5}, definition \eqref{l:DRFTDEF1} and equation \eqref{l:DRFT13} it is
possible to provide more detail on the precise form of the total energy $\h_1=\h_2$.
\begin{prop}
\label{l:HAMILT12}
Let the three-dimensional forward drift $b^+=b^+(x,t)$ and backward drift $b^-=b^-(x,t)$ of the (three-dimensional) stochastic process $x(t)$ with density $\rho=\rho(x,t)$ satisfying Theorem \eqref{l:TIMEERG2} and assume that $\gamma_m^2$ is small and that the variance constant $\sigma$ is a constant. Then using the potential $V=V(x,t)$ the expectation of the pre-collision energy $E[\h_1]$ equals the post-collision average energy $E[\h_2]$ so that in three dimensions
\begin{align}
\label{l:HAMILT1}
\begin{split}
E\left[\h_1\right] & = E\left[\h_2\right] =
\\
&\frac{M}{2}\left(E\left[\abs{\frac{b^++b^-}{2}}^2\right]
+\frac{\eta^2}{4M^2} E\left[\abs{\frac{1}{\rho}\nabla
\rho}^2\right]\right)
+\int \rho Vdx +
\frac{3\eta}{\sin{(\theta)}\overline{\tau}},
\end{split}
\end{align}
where
\begin{align*}
\eta & = M\sigma^2/\gamma_m,
\\
E[V] & = \int \rho Vdx.
\end{align*}
Here $E[.]$ shows the density of $\rho$ for state space x taking into account the forward $b^+$ and backward $b^-$ mean drifts.
\begin{proof}
The terms that requires an explanation are the expectation of the kinetic energy term, the osmotic term and the resident constant variance term at the end of equation \eqref{l:HAMILT1}. Using the expressions for the main particle backward and forward drift processes defined as the expectation of \eqref{l:DRFT13} becomes
\begin{align*}
E\left[\abs{\frac{v_2+v_1}{2}}^2+ \frac{1}{\gamma_m^2}\abs{\frac{v_2-v_1}{2}}^2\right]
= &E\left[\abs{\frac{b^++b^-}{2}+ \frac{1}{2}\sigma\frac{\Delta^+z}{\tau_2}+\frac{1}{2}
\sigma\frac{\Delta^-z}{\tau_1} }^2\right]
\\
& +\frac{1}{\gamma_m^2} E\left[\abs{\frac{b^+-b^-}{2}
+\frac{1}{2}\sigma\frac{\Delta^+z}{\tau_2}-\frac{1}{2}
\sigma\frac{\Delta^-z}{\tau_1} }^2\right],
\end{align*}
which can be simplified to
\begin{align}
\label{l:HAMILT3}
\begin{split}
E&\left[\abs{\frac{b^++b^-}{2}}^2\right]
+\frac{1}{\gamma_m^2}
E\left[\abs{\frac{b^+-b^-}{2}}^2\right]
+\frac{3\sigma^2}{2}\left(1+\frac{1}{\gamma_m^2}
\right)E\left[\frac{1}{\tau}\right]
\\
=&E\left[\abs{\frac{b^++b^-}{2}}^2\right]
+\frac{\sigma^4}{4\gamma_m^2} E\left[\abs{\frac{1}{\rho}\nabla
\rho}^2\right]
+\frac{6\sigma^2}{\gamma_m\sin(\theta)\overline{\tau}},
\end{split}
\end{align}
assuming that $E[1/\tau] = 2 / \overline{\tau}$ and using the definition for $\sin(\theta)=2\gamma_m/(1+\gamma_m^2)$ provided above. Again here $\rho(x,t)$ is the probability density for $x(t)$ used in Theorem \eqref{l:TIMEERG2} together with the (three-dimensional) forward and backward drifts $b^+=b^+\left(x,t\right)$ and $b^-=b^-\left(x,t\right)$.

Clearly now
\begin{align*}
E[\h_1] & = E[\h_2] =
\frac{M}{2}E\left[\abs{\frac{v_2+v_1}{2}}^2+ \frac{1}{\gamma_m^2}
\abs{\frac{v_2-v_1}{2}}^2\right]
\\
&=\frac{M}{2}\left(E\left[\abs{\frac{b^++b^-}{2}}^2\right]
+\frac{\sigma^4}{4\gamma_m^2} E\left[\abs{\frac{1}{\rho}\nabla
\rho}^2\right]\right)
+\frac{3M\sigma^2}{\gamma_m\sin(\theta)\overline{\tau}}
\\
&=\frac{M}{2}\left(E\left[\abs{\frac{b^++b^-}{2}}^2\right]
+\frac{\eta^2}{4M^2} E\left[\abs{\frac{1}{\rho}\nabla
\rho}^2\right]\right)
+\frac{3\eta}{\sin(\theta)\overline{\tau}},
\end{align*}
which reconciles the terms and proves the Theorem.
\end{proof}
\end{prop}

Equation \eqref{l:HAMILT1} show the energy embedded in the main and incident particle but there are no other effects on the two particles so \eqref{l:HAMILT1} must be constant in time. This specifies exactly how the density for the space variable specified by $\rho(x,t)$ is defined. In fact, following E. Nelson here the following becomes clear.
\begin{thm}
\label{l:THEOREM6}
Assume that $\eta =M\sigma^2/\gamma_m$ for a constant variance $\sigma$, $\Delta=\frac{\partial^2}{\partial x^2_1}+...+\frac{\partial^2}{\partial x^2_3}$ and let $R=R(x,t)$, $S=S(x,t)$ and $\Psi=\Psi(x,t)$ be appropriate functions. Then in the absence of angular momenta or spin the Hamiltonian in \eqref{l:HAMILT1} (or averaged \eqref{l:DRFT13}) is time independent if the main particle has probability density $\rho=\rho(x,t)$ for the main particle position process $x(t)$ derived from the wave function $\psi=\psi(x,t)$ where
\begin{align}
\label{l:SCHROD1}
\begin{split}
\psi &=\psi=
exp\left[\frac{ MR+iMS}{\eta}\right]=exp\left[\frac{\gamma_m(R+iS)}{\sigma^2}\right],
\end{split}
\end{align}
satisfying Schr\"{o}dinger's equation
\begin{align}
\label{l:SCHROD5}
\begin{split}
i\eta\frac{\partial\psi}{\partial t}&=-\frac{\eta^2}{2M}\Delta\psi + (V+\Psi)\psi,
\\
\rho&=\,|\psi\,|^2,
\end{split}
\end{align}
for the function $\Psi$ satisfying
\begin{align*}
\frac{\partial V}{\partial t} = \big(\nabla S\big)^T \nabla\Psi=S_{x_j} \Psi_{x_j}.
\end{align*}
The forward and backward drift can be written as
\begin{align}
\begin{split}
\label{l:DRFT15}
b^\pm&=\left(\nabla S \pm \gamma_m\nabla R\right),
\\
& = \frac{\eta}{M}\left(\IM \pm \gamma_m\RE \right)\frac{\nabla\psi}{\psi},
\end{split}
\end{align}
and the total (average) energy equals
\begin{align}
\label{l:HAMILT4}
\begin{split}
E\left[\h_1\right]=&E\left[\h_2\right]=\frac{\sigma^4\delta^2M}{2}E\left[ \frac{\abs{\nabla
\psi}^2}{\abs{\psi}^2}\right]+ E[(V+\Psi)] + \frac{3\eta}{\sin(\theta)\overline{\tau}}
\\=&
\frac{\eta^2}{2M}\int \abs{\nabla
\psi}^2 dx + \int \rho \big(V+\Psi\big)dx + \frac{3\eta}{\sin(\theta)\overline{\tau}}.
\end{split}
\end{align}
\begin{proof}
The proof is presented in Appendix \eqref{l:APP5}. The main - and incident particle energy was  represented by \eqref{l:HAMILT1} and then it was shown that the energy became invariant to time as long as the probability density satisfied \eqref{l:SCHROD5}. It is dissimilar from the original proof in Nelson~\cite{ENELSON1}, ~\cite{ENELSON2} as he did not consider the $\gamma_m$ weighting term and employed operators in the proof. Carlen~\cite{ECARL1} demonstrated that the stochastic differential equation \eqref{l:DRFT2aa} and therefore \eqref{l:DRFT2ab} admits many strong and weak solutions.  In general he showed that a weak stochastic solution to \eqref{l:DRFT2aa} exists if the potential $V$ belongs to a class Kato-Rellich potential. Notice that the proof shows the average energy given the stochastic process defined by \eqref{l:DRFT3aa} and \eqref{l:DRFT3ab} but \eqref{l:DRFT13} suggests that Theorem \eqref{l:THEOREM6} holds without using averaging. Notice that the fixed contribution here could be ignored as long as the collision rate $1/\overline{\tau}$ does not change in time.
\end{proof}
\end{thm}

The Theorem above refers to any energy and momentum conserving set of particles in physics so it can be applied to all Hamiltonian collision problems, diffusion problems, quantum mechanics, electrodynamic problems and gravity. There is no reference here to a stochastic interpretation of quantum mechanics though it is tempting to formulate one. Obviously equation \eqref{l:SCHROD1} refers to the Schr\"{o}dinger wave function if the variance equals $\eta = \hbar$ but one difference is that the $\gamma_m R(t)$ refers to the distribution width in the distribution rather than $R(t)$. The Theorem above derives that $\eta=M\sigma^2/\gamma_m$ so that means for the variance of the process that
\begin{align}
\label{APP29}
\frac{\sigma^2}{2} = \frac{\eta\gamma_m}{2M} = \frac{\hbar}{2M}\gamma_m,
\end{align}
for small $\gamma_m=\sqrt{m/M}$. This suggests that the diffusion term is dependent on the size and mass ratio $m/M$ of the incident particles generating the diffusion. This is another difference as Nelson originally suggested that $\sigma^2 = \hbar / M$ suggesting a usually larger variant for the stochastic differential equations.

Equation \eqref{l:HAMILT1} refers to the constant energy term in the energy function $3\eta/\sin(\theta)\overline{\tau}$ which is due to the stochastic process limit. As the average inter-collision time $\overline{\tau}$ decreases the path of the particle approaches the equations \eqref{l:DRFT2aa} and \eqref{l:DRFT2ab} however the constant energy contribution explodes. Conveniently this contribution is ignored which is reasonable until relativistic conditions return or the inter-collision times change in time for other reasons.

\section{The Diffusion Process}

This Section addresses the many consequences of Theorem \eqref{l:THEOREM3} and Theorem \eqref{l:THEOREM6}. The first examples shows how Newton's information can be altered for the solar system due to diffusion and the second example addresses quantum mechanics. Another example shows how the changed Newton formula may attribute to the radial velocity of stars around the galactic centre in an effort to find Dark Matter. The last example shows how Theorem \eqref{l:THEOREM2} shows
a relationship with relativity specifically with Minkowski's equation.

\medskip

\textbf{Newtonian Mechanics Adjusted.}
First apply \eqref{l:DRFT4} to our solar system using the Earth's mass $M$ and compare the result to Newton's calculations of the Earth's orbit around the sun. If there is no diffusion at all then equations \eqref{l:DRFT4} and \eqref{l:DRFT13} can be simplified by assuming that the movement $w_1$, $w_2$ of the incident particle are non-stochastic and valued $w(t)$ so then
\begin{align*}
w_2 - w_1 \approx 0,
\\
v_2 - v_1 \approx 0.
\end{align*}
Imagine that the potential field generating the energy of the incoming particle $w(t)$ is generated by a central potential of the sun and that the main particle is the Earth with mass $M$ residing a distance $r=\abs{x(t)-x_\odot}$ away from the sun. The Earth experiences a potential field around the sun holding mass $M_\odot$ where $M << M_\odot$ and assume that $\gamma_m^2=1$. Now the approximate potential $f_{approx}$ equals
\begin{align*}
f_{approx} = & \abs{\frac{w_2+w_1}{2}}^2+
\abs{\frac{w_2-w_1}{2}}^2 \approx \abs{\frac{w_2+w_1}{2}}^2,
\end{align*}
and
\begin{align*}
\frac{M\gamma_m^2}{2}f_{approx}=\frac{M}{2}f_{approx}=\frac{M}{2}\abs{\frac{w_2+w_1}{2}}^2 = \frac{GMM_\odot}{r},
\end{align*}
where $G$ is the gravitational constant. Applying Proposition \eqref{CORR1} then \eqref{l:DRFT4} becomes
\begin{align}
\label{l:HAMILT42}
\begin{split}
\h_2=\h_1=&
\frac{M}{2}\bigg(\abs{\frac{v_2+v_1}{2}}^2+\frac{1}{\gamma_m^2}
\abs{\frac{v_2-v_1}{2}}^2\bigg) + V
\\
\approx &\frac{M}{2}\abs{\frac{v_2+v_1}{2}}^2 + \frac{M\gamma_m^2}{2}f_{approx}
\\
= & \frac{M}{2}\abs{\frac{dx(t)}{dt}}^2 + \frac{GMM_\odot}{\abs{x-x_\odot}}.
\end{split}
\end{align}
This is the classical Newtonian Hamiltonian for mass $M$ orbiting the sun with mass $M_\odot$ placed in position $x_\odot$.

But now assume that there is a diffusion perturbation so that $\abs{v_2-v_1}^2/\gamma_m^2$ has become noticeable and has to be taken into account. In that case following equation \eqref{l:SCHROD5} the approximation of the perturbed orbit of particle $M$ would be provided by Schr\"{o}dinger's equation
\begin{align*}
i\eta\frac{\partial\psi}{\partial t}&=-\frac{\eta^2}{2M}\Delta\psi + \frac{GMM_\odot}{\abs{x(t)-x_\odot}}\psi,
\\
\rho&=\,|\psi\,|^2,
\end{align*}
or using polar coordinates for two-dimensions this reduces to
\begin{align*}
i\eta\frac{\partial\psi}{\partial t}&=-\frac{\eta^2}{2M}
\left(\frac{\partial^2 \psi}{\partial r^2} + \frac{1}{r}\frac{\partial \psi}{\partial r}\right)+\frac{GMM_\odot}{r}\psi,
\\
\rho&=\,|\psi\,|^2,
\end{align*}
using polar coordinates where $r$ is the distance to the sun $M_\odot$ and where $\eta=M\sigma^2/\gamma_m=M\sigma^2$. In this case the second dimension is reasonable as almost all planets except Pluto tend to lead a very similar flat orbit. This equation is an approximation as the polar coordinate solution parts for the angle $\theta$ have not been included. Also notice that the time-dependency $\Psi$ is not included because Newton's attraction to the sun is not time-dependent.

The perturbation $\eta$ generated by $\abs{v_2-v_1}^2/\gamma_m^2$ is not necessarily small but is due to the small fluctuations in planet density, the random presence of gas in the solar system and the fluctuations of the sun's gravity as the sun produces small mass distribution changes. Notice that because this is the  Schr\"{o}dinger's equation the solution to this equation produces a fixed set of solutions of planets bound by the sun's gravity. This means a finite number of planets in the solar system and their equilibrium distance to the sun is well defined. Any masses with excess energy will escape from the sun's orbit and disappear into space. Given the number of planets and the size of their radial distance to the sun provides the possibility of estimating the diffusion constant $\sigma^2$. This approach can not be used for the solar system in its infancy as the original gas cloud that constituted our solar system before the sun started its nuclear process some 5 billion years ago. In that case it would be wiser to use a gas approximation inside the potential that experienced its own gravity, see Penrose~\cite{PENROSE1}

\medskip

\textbf{Constructing Quantum Mechanics.}
The Schr\"{o}dinger equation \eqref{l:SCHROD1} using the diffusion model of \eqref{l:DRFTDEF0a} and \eqref{l:DRFTDEF0b} was simplified by assuming that the main particle is modelled with the potential $V$ representing the effect of the incident particle $m$. Hence the main particle has equations of motion while the incident particles are being provided randomly. The following example shows the iterative process in the case that there is one constant potential between $t$ and $t+\tau_2$.

In this case Theorem \eqref{l:THEOREM6} implies that the wave function solves
\begin{align*}
i\eta\frac{\partial \psi}{\partial t}&=-\frac{\eta^2}{2M}\Delta\psi + \frac{mc_w^2}{2}\psi = -\frac{\eta^2}{2M}\Delta\psi + \overline{e}\psi,
\end{align*}
where $c_w$ is the speed of the incident particles and the energy of the incident particle equals $2\overline{e}=mc_w^2$. So then
\begin{align}
\label{l:ENERG16}
\psi&=Ae^{i\big(\frac{p^Tx - Et}{\eta}\big) + i\frac{\overline{e}}{2\eta}t},
\end{align}
where the energy $2E = M\abs{p}^2-mc_w^2=M\abs{p}^2-2\overline{e}$ for all times $t\leq t \leq t+\tau_2$.

Imagine that $\overline{e}=mc^2_w/2$ equals the minimum basis energy of the electron or photon and choose the main particle energy $M\abs{p_k}^2/2=k\overline{e}=E_k$ then the energy term $E_k$ becomes $E_k=(k-1)\overline{e}$ which means that \eqref{l:ENERG16} simplifies to
\begin{align*}
\psi_k=\psi_k(x,t)=Ae^{i\big(\frac{p^Tx - E_kt}{\eta}\big)}=Ae^{i\big(\frac{p^Tx - E_kt}{\eta}\big)},
\end{align*}
with the main part energy $E_k \geq -mc_w^2/2=-\overline{e}$ coupling the movement of $M$ and the energy of $m$ in one expression. Notice that it is possible that the main part energy $E_k,k \geq 0$ is negative in which case $E_0=-\overline{e}$ or it can be assumed that $k>0$ to avoid this problem. Choosing a fraction of $\overline{e}$ as the electron or photon energy creates the possibilities that the energy for the main particle can become substantially negative in which case tunnelling needs to be taken into account.

The main particle $M$ collides at time $t$ and then moves on the next collision with a new incident particle at time $t+\tau_2$. So between the collision between $M$ and the incident particle $m$ at time $t$ and $t+\tau_2$ the probability density $\rho$ applies and when the next collision occurs this probability density is the initial condition for the new wave function of the main particle. The new incident particle may have a completely different set of initial conditions or it may interact with the main particle. If the collisions are random the main particle has proper equations of motion while the incident particles are being provided randomly. However, if there are arrangements between the main and incident particles the relationship needs to be evaluated using \eqref{l:HAMILT42}. The following example shows a free particle colliding with a set of free incident particles.

\medskip

\textbf{Example of stationary behavior.}
In a heatbath where the main particle is represented as provided by Proposition \eqref{l:EQMOT1P} it can be shown that the main particle has a correlation to the incident particle depending on the speed it is moving. Let $v_1,v_2$ and $w_1,w_2$ be three-dimensional then Proposition \eqref{l:EQMOT1P} shows for the main particle the after-collision speed becomes
\begin{align}
\label{APP15}
v_2 =\left(I-\gamma_m\sin\left(\theta\right)P(\phi)\right)v_1+\gamma_m\sin\left(\theta\right)P(\phi)w_1,
\end{align}
so both the momentum and energy of the main particle will be affected by the incident particle. First the case of no correlation between $v_1$ and $w_1$ is considered and then the correlation is calculated assuming that $w_1$ and $w_2$ carry exactly the same energy.

First specify the correlation $\rho$ between $v_1$ and $w_1$ with respect to their relative energies as follows. Since
\begin{align*}
E[v_1^Tw_1]^2 \leq E\abs{v_1}^2E\abs{w_1}^2,
\end{align*}
the correlation can be defined as
\begin{align*}
\rho=\frac{E[v_1^Tw_1]}{\sqrt{E\abs{v_1}^2E\abs{w_1}^2}},
\end{align*}
since in that case $-1 \leq \rho \leq 1$.

For stationary processes the correlation is proportional to $E[v_1^Tw_1]$, so using $E[v_1^Tw_1]=0$ and $E[P(\phi)]=I$ (where $\phi \in \Re^3$, I being the unit matrix) the expectation of equation \eqref{APP15} becomes
\begin{align}
\begin{split}
\label{l:HAMILT7}
E[v_2] =
&
E[v_1] +
\gamma_m\sin{\theta}\left(E[w_1] - E[v_1]\right),
\\
=
&
E[v_1] + \frac{2\gamma_m^2}{(1+\gamma_m^2)}\left(E[w_1] -
E[v_1]\right)
\\
\approx
&
\left(1 - 2\gamma_m^2\right)E[v_1] + 2\gamma_m^2E[w_1],
\end{split}
\end{align}
which means the main particle $M$ experiences acceleration or deceleration between $v_!$ and $v_2$ depending on the average movement of the mean of $w_1, w_2$. So with $E[v_1^Tw_1]=0$ the energy calculation for $v_2$ becomes
\begin{align}
\label{APP19}
\begin{split}
E\abs{v_2}^2
= &
E\abs{v_1}^2 - 2\gamma_m\sin(\theta)E\abs{v_1}^2 +\gamma_m^2\sin^2(\theta) E\abs{v_1}^2
+\gamma_m^2\sin^2(\theta)E\abs{w_1}^2
\\
& +
2\gamma_m\sin(\theta)(1-\gamma_m\sin(\theta))E[v_1^Tw_1]
\\
= &
E\abs{v_1}^2 +\gamma_m\sin(\theta)\left(- 2 +\gamma_m\sin(\theta)\right) E\abs{v_1}^2
+\gamma_m^2\sin^2(\theta)E\abs{w_1}^2,
\end{split}
\end{align}
and then the denominator becomes
\begin{align*}
E\abs{v_2}^2 - E\abs{v_1}^2
& \approx
-4\gamma_m^2\left((1-\gamma_m^2)E\abs{v_1}^2
-\gamma_m^2E\abs{w_1}^2\right)
\\ & \approx
-4\gamma_m^2\left(E\abs{v_1}^2
-\gamma_m^2E\abs{w_1}^2\right).
\end{align*}
Clearly the momentum and energy are in equilibrium if approximately
\begin{align*}
E\left[v_1\right]
& \approx
E\left[w_1\right],
\\
ME\abs{v_1}^2
& \approx
M\gamma_m^2c_w^2
=
mc_w^2,
\end{align*}
where again $c_w^2 = E\abs{w_1}^2$. This means that for stationary processes the main particle assumes the incident particle speed and assumes the incident particle energy.

So for the case with correlation assume that the energy is conserved in the incident particles during the particle collision. From the same equation \eqref{l:EQMOT1P} it is clear that
\begin{align}
\label{APP26}
w_2 =\frac{\sin\left(\theta\right)}{\gamma_m}P(\phi)v_1  + \left(I - \frac{\sin\left(\theta\right)}{\gamma_m}\right)P(\phi)w_1,
\end{align}
hence if the incident vector always maintains the same energy there will be an effect on the main particle. In this case for $v_1,v_2$ and $w_1,w_2$ then
\begin{align}
\label{APP20}
\begin{split}
\abs{w_2}^2
= &
\abs{\frac{\sin\left(\theta\right)}{\gamma_m}P(\phi)v_1+
\left(I-\frac{\sin\left(\theta\right)}{\gamma_m}P(\phi)\right)w_1}^2
\\
= &
\abs{w_1}^2 -
2\frac{\sin(\theta)}{\gamma_m}w_1^TP(\phi)w_1 +\frac{\sin^2(\theta)}{\gamma_m^2} w_1^TP^T(\phi)P(\phi)w_1
+\frac{\sin^2(\theta)}{\gamma_m^2}v_1^TP^T(\phi)P(\phi)v_1
\\
& +2\frac{\sin(\theta)}{\gamma_m}v_1^TP^T(\phi)\left(1-\frac{\sin(\theta)}{\gamma_m}P(\phi)\right)w_1.
\end{split}
\end{align}
With the usual interaction averaged so that $E[P(\phi)]=I(\text{unit matrix})$ and $P^T(\phi)P(\phi)=P(\phi)P(\phi)=P(\phi)$ then
\begin{align}
\label{APP27}
\begin{split}
E\abs{w_2}^2
= &
E\abs{w_1}^2 - 2\frac{\sin(\theta)}{\gamma_m}E\abs{w_1}^2 +\frac{\sin^2(\theta)}{\gamma_m^2} E\abs{w_1}^2
+\frac{\sin^2(\theta)}{\gamma_m^2}E\abs{v_1}^2
\\
& +
2\frac{\sin(\theta)}{\gamma_m}\left(1-\frac{\sin(\theta)}{\gamma_m}\right)E[v_1^Tw_1],
\end{split}
\end{align}
so then the average energy of the incident velocities remains unchanged if
\begin{align}
\begin{split}
\label{l:DRFT21}
E[v_1^Tw_1]= &
\frac{
\left(2 -\frac{\sin(\theta)}{\gamma_m}\right) E\abs{w_1}^2 - \frac{\sin(\theta)}{\gamma_m}E\abs{v_1}^2
}
{
2\left(1-\frac{\sin(\theta)}{\gamma_m}\right)
}
=
\frac{E\abs{v_1}^2-\gamma^2_m E\abs{w_1}^2}{1-\gamma_m^2},
\end{split}
\end{align}
and the correlation becomes
\begin{align*}
\rho = \frac{\sqrt{E\abs{v_1}^2}}{\left(1-\gamma_m^2\right)\sqrt{E\abs{w_1}^2}}
-\frac{\gamma^2_m \sqrt{E\abs{w_1}^2}}{\left(1-\gamma_m^2\right)\sqrt{E\abs{v_1}^2}}.
\end{align*}
which for small $\gamma_m << 1$ reduces to
\begin{align*}
\rho \approx \frac{\sqrt{E\abs{v_1}^2}}{\sqrt{E\abs{w_1}^2}} \approx
\frac{E\abs{v_1}}{E\abs{w_1}},
\end{align*}
assuming that $m << M$, $E\abs{v_1}^2 \approx E\abs{v_1}E\abs{v_1}$ and $E\abs{w_1}^2 \approx E\abs{w_1}E\abs{v_1}$. With the correlation $\rho$ between the main and incident particle velocities assures that the main particle average velocity $E\abs{v_1}$ will obtain a fraction $\rho$ of the velocity of the incident particle $E\abs{w_1}$. The mass in this case does not affect the result.

\medskip

\textbf{Is there Black Matter?} A current important matter in cosmology and in fact for physics is the apparent lack of existence of mass in the universe as the extrapolated mass from the observed starlight seems to insufficiently describe the radial motion of stars around the galactic centres. Observations on the radial star velocity distributions around the center of the universe should show a significant downward slope as the distance to the center increases but the observed pattern shows an almost flat horizontal slope. This suggests either (much) more dark mass in our galaxy than their light indicates or the amount of mass extrapolated from starlight and radiation is severely under-estimated.

Since 1986 a non-relativistic potential theory for gravity is considered which differs from the Newtonian theory, see Milgrom \cite{MILGROM}. The MOND theory is built on the basic assumption of the modified dynamics shown earlier to reproduce dynamical properties of galaxies and galaxy aggregates without having to assume the existence of hidden mass. Equation \eqref{l:DRFT4} shows that the MOND approach may have more validity as long as the diffusion implied by the incident star systems is noticeable.

To be specific take the motion of a radially symmetric galactic configuration and determine that the incident particles consist of circling stars and galactic dust. Equation \eqref{l:THEOREM5} is generally applicable to all Hamiltonian forces and therefore applies to Newtonian gravity as well. The $w_1, w_2$ part represents the force part of the equation generating the potential $V$ and then Proposition \eqref{l:HAMILT12} and Theorem \eqref{l:THEOREM6} suggest that the distribution of a solar system moving around the galaxy employing satisfies a Schr\"{o}dinger equation with an appropriate variance. The motion of the main sun needs to be determined as a function of the distance to the galactic centre $r$ so then equation \eqref{l:DRFT18} needs to be adjusted to account for gravity and the osmotic gravity term resulting from the diffusion of the dynamic star motion.

Consider for a galaxy that their mass is concentrated at the centre of the galaxy which will be assumed equal to $M_R$. Then assume that the mass is concentrated in the centre the Newtonian energy of the system equals
\begin{align*}
\frac{w_2+w_1}{2} = \left(\frac{GM_R}{r}\right)^{\frac{1}{2}},
\end{align*}
so that
\begin{align*}
\abs{\frac{w_2+w_1}{2}}^2 = \frac{GM_R}{r}.
\end{align*}
In addition, the galactic acceleration of the main particle in $x(t)$ as a function of the radial centre $r$ demands that
\begin{align*}
\frac{\partial x(t)}{\partial r} = \frac{w_2-w_1}{2} = -\frac{GM_R}{r^2},
\end{align*}
so that
\begin{align*}
\abs{\frac{w_2-w_1}{2}}^2 = \frac{\chi(r)^2G^2M_R^2}{r^4},
\end{align*}
which explains the osmotic velocity term in equation \eqref{l:APP21} except for term $\chi(r)$. This term is the result of the lengthening paths of galactic mass being randomized while rotating the axis. The function $\chi(r)$ can ameliorate the terms of the acceleration making the profile less profound for large $r$ and can make the $r$ dependence closer to $1/r$ for smaller $r$. The profile can be made equal to the typical Newtonian by choosing $\chi(r)$ equal to zero. With these choices the potential subject to the mass of galactic stars equals
\begin{align}
\label{l:APP21}
\begin{split}
f_{approx}(x,t)
= & E\left[\abs{\frac{w_2+w_1}{2}}^2+\gamma_m^2\abs{\frac{w_2-w_1}{2}}^2\bigg|x(t)=x\right]
\\
= & E\left[\frac{GM_R}{r}+\gamma_m^2\frac{\chi^2(r)G^2M_R^2}{r^4}\bigg| x(t)=x\right].
\end{split}
\end{align}
This assumes that all mass is concentrated in the galactic centre which is not true for spiral galaxies but may be more accurate for globular clusters.

A solution has been provided by Theorem \eqref{l:THEOREM6} once the correct constants have been provided. First $\gamma_m^2=M/M_R$ where $M_R$ is the galactic mass of the galaxy under consideration and $\eta=M_R\sigma^2/\gamma_m$. The diffusion in this equation is not related to quantum mechanics but to the diffusive motion the collective star systems and gas impart to a mass $M$. Notice also that the diffusion attribution $\chi(r)=\chi$ has been chosen a constant throughout the galaxy. Now the Schr\"{o}dinger equation \eqref{l:SCHROD5} substituting the value of the potential \eqref{l:APP21} so that
\begin{align*}
i\eta\frac{\partial\psi}{\partial t}
& = -\frac{\eta^2}{2M}\left(\frac{\partial^2 \psi}{\partial r^2} + \frac{1}{r}\frac{\partial \psi}{\partial r}\right)\psi + MM_R\left(\frac{G}{r}+\gamma_m^2\frac{\chi^2G^2M_R}{r^4}\right)\psi
\\
& = -\frac{\eta^2}{2M}\left(\frac{\partial^2 \psi}{\partial r^2} + \frac{1}{r}\frac{\partial \psi}{\partial r}\right)\psi + MM_R\left(\frac{G}{r}+\frac{\chi^2G^2M}{r^4}\right)\psi,
\\
\rho&=\,|\psi\,|^2.
\end{align*}
Clearly this potential serves one of Migrom's requirements that the attractive force decreases as $r$ increases but it does not satisfy the requirement for short $r$ until the function $\chi(r)$ can be adjusted accordingly. Obviously the radial speed distribution will be adjusted when $\eta$ is adjusted so the diffusion $\eta$ can be adjusted to the radial velocity pattern observed.

The Schr\"{o}dinger equation has been studied in different contexts but a solution for the equation above does not exist \cite{ROBERTSHAW1}, \cite{HARRISON1} though the Schr\"{o}dinger equation with Newton potential is studied with a potential satisfying a Poisson equation. This set of equations were introduced by Penrose \cite{PENROSE1} to the system of equations consisting of the Schr\"{o}dinger equation for a wave-function moving in a potential $\phi$, where $\phi$ is solved from the Poisson equation. This Schr\"{o}dinger-Newton equation describes a particle that moves under its own gravitational field describing quantum properties while gravity remains classical even at the fundamental level. The system has been studied analytically by Tod \cite{TOD1}.

\medskip

\textbf{Example of Minkowski transaction.}
This example shows that some form of the Minkowski surface condition by equations \eqref{l:EQMOT9} using Theorem \eqref{l:THEOREM2}. It is very interesting to see that the Theorem shows a stochastic error in the Minkowski equation.
\begin{thm}
\label{l:THEOREM9}
The requirements on Theorem \eqref{l:THEOREM2} specifies the mean motion of the mean collision mean $a$ and the interactions $b=w_1-v_1$, $b^\bot=b+\frac{1}{\gamma_m}\Phi$ by manipulating equation \eqref{l:EQMOT4} to derive
\begin{align}
\label{APP23}
\begin{split}
b & = \frac{w_1-v_1}{(1+\gamma_m^2)},
\\
b^\bot & = \frac{v_2-w_2}{(1+\gamma_m^2)},
\\
a & = \frac{Mv_2+mw_2}{M+m} = \frac{Mv_1+mw_1}{M+m},
\end{split}
\end{align}
where $b^\bot=b + \frac{1}{\gamma_m}\Phi$ and $\abs{b}^2 = \abs{b^\bot}^2$. The function $\Phi$ is defined in Theorem \eqref{l:THEOREM2}. Then
\begin{align}
\label{APP21}
\abs{w_2-a}^2 - \abs{v_2 - a}^2 = \abs{w_1-a}^2 - \abs{v_1 - a}^2,
\end{align}
so that the motion of all particles is constrained vis-a-vis the average velocity $a$. If the main particle on average does not loose or gain energy or if the correlation between $a$ and $b+b^\bot$ is zero then
\begin{align}
\label{APP22}
E\abs{w_2}^2 - E\abs{v_2}^2 = E\abs{w_1}^2 - E\abs{v_1}^2.
\end{align}
For small $m$ where $m<<M$ it is clear that equation \eqref{APP21} is very close to \eqref{APP22} and for massless photons as incident particles these equations are identical. In these cases the average motion of the difference between the two particles randomizes against the behaviour of $b+b^\bot$.

\begin{proof}
Using the collision representation \eqref{l:EQMOT4} the following can be obtained quite easily
\begin{align*}
\abs{w_2-a}^2 - \abs{v_2-a}^2
= \abs{-b^\bot}^2 - \abs{\gamma_m^2b^\bot}^2
= (1-\gamma_m^4)\abs{b^\bot}^2,
\end{align*}
and
\begin{align*}
\abs{w_1-a}^2 - \abs{v_1-a}^2
= \abs{b}^2-\abs{-\gamma_m^2b}^2
= (1-\gamma_m^4)\abs{b}^2.
\end{align*}
Then
\begin{align*}
\abs{w_2-a}^2 - \abs{v_2-a}^2 -\left(
\abs{w_1-a}^2 - \abs{v_1-a}^2\right)
=
(1-\gamma_m^4)(\abs{b^\bot}^2-\abs{b}^2)
=0
\end{align*}
because $\abs{b}^2=\abs{b^\bot}^2$. Notice that the expressions above show that the mean average energy is always positive as long as $\gamma_m^4 < 1$ so the $v_2-a,v_1-a$ is always less than the $w_2-a,w_1-a$ energies.

To remove the dependence on the vector $a$ remove the vector from \eqref{APP22} then it is clear that
\begin{align*}
& \abs{w_2-a}^2 - \abs{v_2 - a}^2 - \left( \abs{w_1-a}^2 - \abs{v_1 - a}^2\right)
\\
& = \abs{w_2}^2 - \abs{v_2}^2 - \left(\abs{w_1}^2 - \abs{v_1}^2\right) + 2a^T\left(v_2-w_2 -(v_1-w_1)\right)
\\
& = \abs{w_2}^2 - \abs{v_2}^2 - \left(\abs{w_1}^2 - \abs{v_1}^2\right) +2(1+\gamma_m^2)a^T\left(b+b^\bot\right),
\end{align*}
so that
\begin{align*}
0 & = E\abs{w_2-a}^2 - E\abs{v_2 - a}^2 - \left( E\abs{w_1-a}^2 - E\abs{v_1 - a}^2\right)
\\
& = E\abs{w_2}^2 - E\abs{v_2}^2 - \left(E\abs{w_1}^2 - E\abs{v_1}^2\right) +2(1+\gamma_m^2)E\left[a^T\left(b+b^\bot\right)\right],
\end{align*}
since $b^\bot$ depends on the position of the main particle, the incident particle and the impact $b$.

To show this relationship is valid when the energy in the main particle does not change write from equation \eqref{l:EQMOT4}
\begin{align*}
\abs{v_1}^2 & = \abs{a}^2-\gamma_m^2a^Tb+\abs{b}^2,
\\
\abs{v_2}^2 & = \abs{a}^2+\gamma_m^2a^Tb^\bot+\abs{b^\bot}^2,
\end{align*}
then the amount of energy gained or lost equals
\begin{align*}
2\Delta E = M\left(\abs{v_2}^2-\abs{v_1}^2\right) = \gamma_m^2a^Tb^{\bot}+\abs{b}^2+\gamma_m^2a^Tb-\abs{b^\bot}^2=\gamma_m^2a^T\left(b^{\bot}+b\right).
\end{align*}
This means that $\gamma_m^2a^T\left(b^{\bot}+b\right)$ shows the amount of energy leaving or arriving in the system so for a static system with $\Delta E=0$ this becomes
\begin{align*}
E\abs{w_2}^2 - E\abs{v_2}^2 = E\abs{w_1}^2 - E\abs{v_1}^2.
\end{align*}
It is interesting to see that for small $m$ where $m<<M$ it is clear that $\Delta E$ is very close to zero and for particles that do not have any mass like photons that $\gamma_m=0$ so that indeed $\Delta E=0$.

In addition, the correlation between $a$ and $b+b^\bot$ is equal to
\begin{align*}
\rho=\frac{E\left[a^T(b+b^\bot)\right]}{\sqrt{E\abs{a}^2
E\abs{b+b^\bot}^2}},
\end{align*}
so that
\begin{align*}
E\abs{w_2}^2 - E\abs{v_2}^2 - \left(E\abs{w_1}^2 - E\abs{v_1}^2\right) =-4(1+\gamma_m^2)\rho\sqrt{E\abs{a}^2
E\abs{b+b^\bot}^2},
\end{align*}
hence
\begin{align*}
E\abs{w_2}^2 - E\abs{v_2}^2 = E\abs{w_1}^2 - E\abs{v_1}^2,
\end{align*}
if $\rho$ equals zero. Assuming this equation being not stochastic it is clear that $E\abs{w_1}^2 \approx \overline{w_1}^2, E\abs{w_1}^2 \approx \overline{w_1}^2$ and $E\abs{w_1}^2 \approx \overline{w_1}, E\abs{w_1}^2 \approx \overline{w_1}^2$ which applies to $v_2,w_2$ as well so finally
\begin{align*}
\abs{E[w_2]}^2 - \abs{E[v_2]}^2 = \abs{E[w_1]}^2 - \abs{E[v_1]}^2.
\end{align*}
Notice that for large mass where $M>>m$ it is clear that $a \approx v_1 \approx v_2$ and $b+b^\bot \approx w_1-w_2$ so then $E\left[a^T(b+b^\bot)\right] \approx E[v(w_2-w_1)] \approx vE[(w_2-w_1)]$. Inertial masses moving through the universe at a constant speed not loosing or gaining energy would have $E\left[a^T(b+b^\bot)\right]=0$ and equation \eqref{APP22} holds true.
\end{proof}
\end{thm}

The conservation law \eqref{APP21} always applies no matter what values the masses $m$ or $M$ are but \eqref{APP22} is the relativistic Minkowski space time as long as $a$ and $b+b^\bot$ are uncorrelated or if the main particle does not accept or loose energy. Energy conservation means that $E[a^T(b+b^\bot)]=0$ so that the momentum of the main particle does not loose or gain momentum / energy upon collision on average. In gravity the incident particle consists out of a heatbath of light so then clearly equation \eqref{APP22} applies.

Interestingly the joint energy of the particles equals
\begin{align*}
\frac{2}{M}\h_1= \frac{2}{M}\h_1 = \abs{v_1}^2+\gamma_m^2\abs{w_1}^2 = \abs{v_2}^2+\gamma_m^2\abs{w_2}^2=(1+\gamma_m^2)\abs{a}^2+\gamma_m^2(1+\gamma_m^2)\abs{b}^2,
\end{align*}
and Einstein solved the relativistic case in gravity in equation \eqref{APP22} by varying the time to compensate for equation \eqref{APP22} as the incident particle speed (photons) always has the same speed setting $\abs{w_1}^2=c_w^2\tau_1^2, \abs{w_2}^2=c_w^2\tau_2^2$. Using the pre - and post collision vector of the macroscopic object speeds $v_1$ and $v_2$ then with $c_w^2$ as the constant speed of light equation \eqref{APP22} becomes
\begin{align*}
c_w^2\tau_2^2 - \abs{v_2}^2\tau_2^2 = c_w^2\tau_1^2- \abs{v_1}^2\tau_1^2.
\end{align*}
This requirement determines the behaviour of the times $\tau_1$ and $\tau_2$ by setting
\begin{align}
\label{APP24}
\frac{\tau_2}{\tau_1} = \frac{\sqrt{1-\frac{\abs{v_1}^2}{c_w^2}}}{\sqrt{1-\frac{\abs{v_2}^2}{c_w^2}}},
\end{align}
and then equation \eqref{APP22} can be solved as follows
\begin{align*}
c_w\tau_2 & = c_w\tau_1\cosh{(\xi)} - v_1\tau_1\sinh{(\xi)},
\\
v_2\tau_2 & = -c_w\tau_1\sinh{(\xi)} + v_1\tau_1\cosh{(\xi)},
\end{align*}
which fits \eqref{APP22} because $\cosh{(\xi)}^2-\sinh{(\xi)}^2 = 1$. If however $\gamma_m>0$ and the correlation is not zero then equation \eqref{APP21} holds and there should be a small correction to the time estimates in equation \eqref{APP24}.

\section{Double Quantum Mechanics}
The two particle Hamiltonian describing the elastic collision in Theorem \eqref{l:THEOREM5} has four terms where the first two represent the main particle energy and the third and fourth terms represent the energy of the incident particle. In the previous examples the incident particles in the third and fourth term were approximated by a static potential (Proposition \eqref{l:HAMILT12}) and the minimum calculated using Schr\"{o}dinger equation (Theorem \eqref{l:THEOREM6}). However, this dynamic approximation does not allow the $w_1,w_2$ velocity dependence and does not represent the actual incident particle motion. The easiest way to incorporate the incoming particle behaviour into \eqref{l:DRFT4} is to apply the differential equations \eqref{l:DRFT2aa} and \eqref{l:DRFT2ab} in Theorem \eqref{l:TIMEERG2} for the incident particle as was done for the main particle.

Hence using Theorem \eqref{l:TIMEERG2} for the incident particle assume that the incident particle satisfies \eqref{l:DRFTDEF0a} and \eqref{l:DRFTDEF0a} with drifts $b^+=b^+(x(t),y(t),t)$, $b^-=b^-(x(t),y(t),t)$ with a drift of $\sigma$ for the stochastic differential equation. So it is assumed that
\begin{subequations}
\begin{align}
\label{l:DRFTDEF1a}
d^+
\begin{pmatrix}
x(t)
\\
y(t)
\end{pmatrix}
= &
\begin{pmatrix}
x(t+\tau_2)-x(t)
\\
y(t+\tau_2)-y(t)
\end{pmatrix}
= b^+dt + \sigma dz^+(t)
\\
\label{l:DRFTDEF1b}
d^-
\begin{pmatrix}
x(t)
\\
y(t)
\end{pmatrix}
= &
\begin{pmatrix}
x(t)-x(t-\tau_1)
\\
y(t)-y(t-\tau_1)
\end{pmatrix}
= b^-dt + \sigma d^-z(t)
\\
\label{l:DRFTDEF1c}
\begin{pmatrix}
b^-
\\
b^-
\end{pmatrix}
= &
\begin{pmatrix}
b^+
\\
b^+
\end{pmatrix}
-\sigma \nabla \log(\rho),
\end{align}
where $d^+z(t),d^-z(t)$ is the 6 dimensional Gaussian process where $d^+z(t) = z(t+\tau_2) -z(t)$ and $d^-z(t) = z(t) - z(t-\tau_1)$ and $\nabla^T=\left(\frac{\partial}{\partial
x_1},...,\frac{\partial}{\partial x_3},\frac{\partial}{\partial y_1},...,\frac{\partial}{\partial y_3}\right)$. Here $\rho=\rho(x,y,t)$ is the associated probability density where $x$ refers to the main particle variables and $y$ refers to the incident particle variables.

From this it follows that
\begin{align}
\label{l:DRFT7aa}
\frac{\partial \rho}{\partial t}&= - \nabla^T\left(b^+\rho\right)+\frac{\sigma^2}{2} \Delta\rho,
\\
\label{l:DRFT7ab}
\frac{\partial \rho}{\partial t}&= -\nabla^T\left(b^-\rho\right)
- \frac{\sigma^2}{2} \Delta\rho,
\end{align}
so that
\begin{align}
\label{l:DRFT7ac}
\begin{split}
\frac{\partial \rho}{\partial t} & = -
\nabla^T\left(\left(\frac{b^++b^-}{2}\right)\rho\right),
\end{split}
\end{align}
\end{subequations}
where $\Delta=\sum_{i} \frac{\partial^2}{\partial x_i^2} +\sum_{i} \frac{\partial^2}{\partial y_i^2}$.

By this assumption the sequence of incoming incident particles in time is modelled in the form of the time-dependent continuous stochastic process also experienced by the main particle. It is possible to add correlations between the main particle and use a different variant for the incident particles but for the moment keep the calculations simple. The density $\rho=\rho(x,y,t)$ is assumed to deal with the joint distribution of main particles $v_1, v_2$ and incident particles $w_1, w_2$. Using \eqref{l:DRFT7aa}, \eqref{l:DRFT7ab} and \eqref{l:DRFT7ac} change Proposition \eqref{l:SCHROD1} as follows
\begin{thm}
\label{l:THEOREM7}
Let $b^+_j=b^+_j(x,y,t)$, $b^-_j=b^-_j(x,y,t)$ for $1\leq j \leq 6$ be the forward and backward drift for the main particle $x$ and incident particle $y$ then the expectation of the total energy in $E\left[\h_1\right]=E\left[\h_2\right]$ can be written as
\begin{align}
\label{l:HAMILT39}
\begin{split}
E\left[\h_1\right]=E\left[\h_2\right]
=&\frac{M}{2}\sum_{k=1}^{k=3}\left(E\left[\abs{\frac{b_j^++b_j^-}{2}}^2\right]
+\frac{\eta^2}{4M^2}E\left[\abs{\frac{\nabla_j \rho}{\rho}}^2\right]\right)
\\
&+\frac{m}{2}\sum_{k=4}^{k=6}\left(E\left[\abs{\frac{b_j^++b_j^-}{2}}^2\right]
+\frac{\eta_w^2}{4m^2} E\left[\abs{\frac{\nabla_j \rho}{\rho}}^2\right]\right)
\\
&+\frac{3(\eta+\eta_w)}{\sin{(\theta)}\overline{\tau}}
\end{split}
\end{align}
where
\begin{align}
\label{l:APP15}
\eta = M\sigma^2/\gamma_m,
\eta_w = M\gamma_m^3\sigma^2=m\gamma_m\sigma^2
\end{align}
and $\nabla_j=\frac{\partial}{\partial x_j}$ with $1 \leq j\leq 6$. The expectation $E[.]$ calculates an average over the forward, backward fluctuations $\rho=\rho(x,y,t)$ of $v_1,v_2$ and $w_1,w_2$ and over the position $x$.
\end{thm}
\begin{proof}
The first two terms in \eqref{l:HAMILT39} in a similar fashion in Proposition \eqref{l:HAMILT12} and the terms that require an explanation are the expectation of the new kinetic energy term for the incident particle and the new resident variance $\eta_w$. Using the expressions for backward and forward drift processes as defined in \eqref{l:HAMILT39} expectation over \eqref{l:DRFT4} reduces to
\begin{align*}
&E\left[\abs{\frac{v_2+v_1}{2}}^2 + \frac{1}{\gamma_m^2}\abs{\frac{v_2-v_1}{2}}^2\right]+
\gamma_m^2E_w\left[\abs{\frac{w_2+w_1}{2}}^2 + \gamma_m^2
\abs{\frac{w_2-w_1}{2}}^2\right]
\\
& \quad =\sum_{j=1}^{j=3}E\left[\abs{\frac{b_j^++b_j^-}{2}+ \frac{1}{2}\sigma\frac{\Delta^+z_j}{\tau_2}+\frac{1}{2}
\sigma\frac{\Delta^-z_j}{\tau_1} }^2\right]
+
\\
& \qquad +\frac{1}{\gamma_m^2} \sum_{j=1}^{j=3}E\left[\abs{\frac{b_j^+-b_j^-}{2}
+\frac{1}{2}\sigma\frac{\Delta^+z_j}{\tau_2}-\frac{1}{2}
\sigma\frac{\Delta^-z_j}{\tau_1} }^2\right]
\\
& \qquad +\gamma_m^2\sum_{j=4}^{j=6}E\left[\abs{\frac{b_j^++b_j^-}{2}+
\frac{1}{2}\sigma\frac{\Delta_j^+z}{\tau_2}+\frac{1}{2}
\sigma\frac{\Delta_j^-z}{\tau_1} }^2\right]
\\
& \qquad +\gamma_m^4 \sum_{j=4}^{j=6}E\left[\abs{\frac{b_j^+-b_j^-}{2}
+\frac{1}{2}\sigma\frac{\Delta_j^+z}{\tau_2}-\frac{1}{2}
\sigma\frac{\Delta_j^-z}{\tau_1} }^2\right],
\end{align*}
which can be simplified to
\begin{align}
\label{l:HAMILT5}
\begin{split}
&\sum_{j=1}^{j=3}E\left[\abs{\frac{b_j^++b_j^-}{2}}^2
+\frac{1}{\gamma_m^2}\abs{\frac{b_j^+-b_j^-}{2}}^2\right]
+\frac{3\sigma^2}{2}\left(1+\frac{1}{\gamma_m^2}
\right)E\left[\frac{1}{\tau}\right]
\\
&\quad+\gamma_m^2\sum_{j=4}^{j=6}E\left[\abs{\frac{b_j^++b_j^-}{2}}^2
+\gamma_m^2\abs{\frac{b_j^+-b_j^-}{2}}^2\right]
+\frac{3\sigma^2\gamma_m^2(1+\gamma_m^2)}{2}E\left[\frac{1}{\tau}\right]
\\
&=\sum_{j=1}^{j=3}E\left[\abs{\frac{b_j^++b_j^-}{2}}^2
+\frac{\sigma^4}{4\gamma_m^2}\abs{\frac{\nabla_j \rho}{\rho}}^2\right]
+\frac{6\sigma^2}{\gamma_m\sin(\theta)\overline{\tau}}
\\
&\quad+\gamma_m^2\sum_{j=4}^{j=6}E\left[\abs{\frac{b_j^++b_j^-}{2}}^2
+\frac{\gamma_m^2\sigma^4}{4}\abs{\frac{\nabla_j \rho}
{\rho}}^2\right]
+\frac{6\gamma_m^3\sigma^2}{\sin(\theta)\overline{\tau}},
\end{split}
\end{align}
assuming that $E[1/\tau] = 2 / \overline{\tau}$, $\nabla_j=\frac{\partial}{\partial x_j}$ and using the definition for $\sin(\theta)=2\gamma_m/(1+\gamma_m^2)$ provided above.

These constants can be estimated as follows using the definition \eqref{l:APP15} so that
\begin{align}
\label{l:APP19}
\begin{split}
\frac{\sigma^4}{4\gamma_m^2} & =\frac{\eta^2}{4M^2}
\\
\frac{\gamma_m^2\sigma^4}{4} & = \frac{\eta_w^2}{4\gamma_m^4M^2}=\frac{\eta_w^2}{4m^2}
\end{split}
\end{align}
so that $\sigma^2/\gamma_m = \eta/M$ and $\sigma^2\gamma_m = \eta_w/m$ so clearly then
\begin{align*}
E[\h_1] = E[\h_2]
=&\frac{M}{2}\sum_{j=1}^{j=3}E\left[\abs{\frac{b_j^++b_j^-}{2}}^2
+\frac{\eta^2}{4M^2}\abs{\frac{\nabla_j \rho}{\rho}}^2\right]
+\frac{3\eta}{\sin(\theta)\overline{\tau}}
\\
&+\frac{m}{2}\sum_{j=4}^{j=6}E\left[\abs{\frac{b_j^++b_j^-}{2}}^2
+\frac{\eta_w^2}{4m^2}\abs{\frac{\nabla_j \rho}
{\rho}}^2\right]
+\frac{3\eta_w}{\sin(\theta)\overline{\tau}},
\end{align*}
which settles the proof.
\end{proof}

Equation \eqref{l:HAMILT39} show the energy embedded in the main and incident particle both modeled using stochastic processes but since there are no other effects on the two particles \eqref{l:HAMILT39} must be constant in time. In the previous example a wave function was employed to find the minimum energy but the Theorem below shows a new wave function emerges. This specifies exactly how the density for the space variable specified by $\rho=\rho(x,y,t)$ is defined.
\begin{thm}
\label{l:THEOREM8}
Assume that $\eta = M\sigma^2/\gamma_m$, $\eta_w =m\gamma_m\sigma^2$, $\Delta_x=\left(\frac{\partial^2}{\partial
x^2_1}+...+\frac{\partial^2}{\partial x^2_3}\right)$ and $\Delta_y=\left(\frac{\partial^2}{\partial
y^2_{1}}+...+\frac{\partial^2}{\partial y_{3}^2}\right)$. Also define $R_1=R_1(x,t), S_1=S_1(x,t)$ and let $R_2=R_2(y,t), S_2=S_2(y,t)$ as suitable functions with $R=R_1(x,t)+R_2(y,t), S=S_1(x,t)+S_2(y,t)$. Then in the absence of angular momenta or spin the Hamiltonian in \eqref{l:HAMILT39} (or averaged \eqref{l:DRFT13}) is
time independent if the main particle has probability density $\rho=\rho(x,y,t)=\rho_1(x,t)\rho_2(y,t)=\rho_1\rho_2$ for the main particle position process $x(t)$ derived from the wave functions $\psi_1=\psi_1(x,t)$, $\psi_2=\psi_2(y,t)$ where
\begin{align}
\label{l:SCHROD7}
\begin{split}
\psi_1 &=\psi_1(x,t)=
\exp\left[\frac{M\left(R_1+i
S_1\right)}{\eta}\right],
\end{split}
\end{align}
and
\begin{align}
\label{l:SCHROD8}
\begin{split}
\psi_2 &=\psi_2(y,t)=
\exp\left[\frac{m\left(R_2+i
S_2\right)}{\eta_w}\right].
\end{split}
\end{align}
The Schr\"{o}dinger's equations satisfy
\begin{align}
\label{l:SCHROD10}
\begin{split}
i\eta\frac{\partial \psi_1(x,t)}{\partial t} & = -\frac{\eta^2}{2M}\Delta_x\psi_1(x,t)+\overline{V_1(x,t)}\psi_1(x,t),
\\
i\eta_w\frac{\partial \psi_2(y,t)}{\partial t}&=-\frac{\eta_w^2}{2m}\Delta_y\psi_2(y,t)-\overline{V_2(y,t)}\psi_2(y,t),
\\
\rho(x,y,t)&=\rho_1(x,t)\rho_2(y,t)=\,|\psi_1(x,t)\psi_2(y,t)\,|^2,
\end{split}
\end{align}
for any six-dimensional potential $V=V(x,y,t)$ where
\begin{align*}
\overline{V_1(x,t)} & = \int \rho_2(y,t)\left(V_1(x,y,t)+\Psi_1(x,y,t)\right)dy,
\\
\overline{V_2(y,t)} & = \int \rho_1(x,t)\left(V_2(x,y,t)+\Psi_2(x,y,t)\right)dx,
\end{align*}
and
\begin{align*}
\frac{\partial V_1(x,y,t)}{\partial t} & = \big(\nabla S\big)^T \nabla\Psi_1(x,y,t),
\\
\frac{\partial V_2(x,y,t)}{\partial t} & = \big(\nabla S\big)^T \nabla\Psi_2(x,y,t),
\end{align*}
with
\begin{align*}
\sum_{j=1}^{6}\bigg(\gamma_m^2\int \rho_1 \nabla_jS_1 \nabla_j \big(\overline{V_1(x,t)}\big)dx
- \int\rho_2 \nabla_jS_2 \nabla_j \big(\overline{V_2(y,t)}\big)dy\bigg)=0.
\end{align*}
The total (average) energy of the main and incident particle now equals
\begin{align}
\label{l:HAMILT6}
\begin{split}
E\left[\h_1\right] & = E\left[\h_2\right]=
\frac{\eta^2}{2M}E\left[ \frac{\abs{\nabla\psi_1}^2}{\abs{\psi_1}^2}\right]+
\frac{\eta_w^2}{2m}E\left[ \frac{\abs{\nabla\psi_2}^2}{\abs{\psi_2}^2}\right]+ \frac{3\left(\eta+\eta_w\right)}{\sin(\theta)\overline{\tau}}
\\
& = \frac{\eta^2}{2M}\int \rho\abs{\nabla\psi_1}^2 dx +
\frac{\eta_w^2}{2m}\int \rho\abs{\nabla\psi_2}^2 dx + 
\frac{3\left(\eta+\eta_w\right)}{\sin(\theta)\overline{\tau}},
\end{split}
\end{align}
\begin{proof}
The proof is presented in Appendix \eqref{l:APP14}. This proof uses Appendix \eqref{l:APP5} to apply the result to both the main and incidently particle simultaneously. Instead of entering an arbitrary force the interaction uses both drifts as forces to stop the energy criterium from gaining or loosing energy.
\end{proof}
\end{thm}

\section{Conclusions}

The paper considers elastic collisions between classical two-particle ignoring angular momentum, spin and additional particle properties. Theorem \eqref{l:THEOREM5} shows that the conserved energy for an elastic collision between a main and incident particle allows to be converted into a sum of four terms where two terms relate to the non-conserved main particle energy and two terms relate to the incident particle energy. The first two terms for the main particle satisfy a weighted form of Nelson's measure while the remaining terms also use Nelson's form but have a different weighting due to the mass differences. The weighting on the four forms here are unique but were absent from the original paper presented by Nelson~\cite{ENELSON1}. Two of the terms are average energies while the second difference terms in the main and incident particles are referred to as the main and incident osmotic terms and relate to the diffusion process. The proof of this conserved energy is relatively straightforward from the definition of the input $(v_1,w_1)$ / output $(v_2,w_2)$ particle velocities. Theorem \eqref{l:THEOREM5} has been derived assuming a one-dimensional particle interaction at the time of the collision and using projections to complete the elastic collision.

This unique classical collision energy distribution is classical and can be applied to equilibrium problems in quantum mechanics, statistical mechanics, cosmology, Dark Matter, Dark Energy or mathematical Finance. Presumably the method can be extended to multiple collisions and there seems the possibility of extending the approach to a field form (by employing a continuous number of dimensions). First to start modelling the elastic energy equation the incident particle energy is condensed to a potential and the motion of the main particle is provided by a stochastic differential equation concentrating the two-particle Hamiltonian into a mean-average motion for the main particle. So the incident particle motion is now encompassed by a non-stochastic potential ignoring the diffusion effects. Proposition \eqref{l:HAMILT12} shows the proper form for the collision energy for the main particle and Theorem \eqref{l:THEOREM6} proves how minimizing the energy expression requires that the position of the main particle satisfies Schr\"{o}dinger's equation. Because of the weighting it is clear that the diffusion is small for small incident particle mass ratios $m << M$ and large for cases where $m >> M$ for the main particle.

The resulting Schr\"{o}dinger equation is a classical result as illustrated in this paper but the analysis here does not relate to quantum mechanics as the variance parameter of the Schr\"{o}dinger equation is defined as $\eta=M\sigma^2/\gamma_m=M^{3/2}\sigma^2/m^{1/2}$. In principle quantum mechanics can be easily derived out of assuming that $\eta=\hbar$ but that does not completely define the quantal process. There are other objections to this approach as Wallstrom~\cite{WALLSTROM1} argued that additional conditions on the discreteness of angular momentum does not follow from the stochastic derivation. Though this is a clear requirement on the theory it is also clear that the original paper from Nelson nor this paper deal with angular momentum requirements in the main particle energy representation. Angular representation can be included relatively easily into Theorem \eqref{l:THEOREM5} as it consists of angular moment terms and quadratic angular terms. By explicit inclusion angular momentum and even spin can be included easily by construction.

Using the stochastic differential equation to derive Schr\"{o}dinger's equation is very useful but the Gaussian stochastic differential equations provide quite un-physical paths. Though the formulation provides proper distributions for main particle positions the associated velocity is infinite at all times an Hamiltonian theory actively using Gaussian processes will obtain unpleasant surprises. A better approach could be to use more direct distribution transition assumptions employing binomial, multi-nomial or equivalent processes that satisfy the proper statistical layout as specified in Proposition \eqref{l:HAMILT12} or Theorem \eqref{l:THEOREM2}. The forward steps are defined as they are here in equation \eqref{l:DRFTDEF0a} while the backward distribution \eqref{l:DRFTDEF0b} is defined with an application of Bayes' Theorem on multi-nomials. In that case a new form for the continuity equation from \eqref{l:DRFT2ab} for the multi-nomial distribution form needs to be found. The advantage of multi-nomial distributions is that they represent interaction times and define velocities appropriately so that for instance QED developments using Feynman's~\cite{FEYNMAN1} ideas could be represented more directly. To avoid the re-normalisation procedures there are many areas that require a better representation .

Notice that generating Schr\"{o}dinger's equation in Theorem \eqref{l:THEOREM6} is slightly different from the standard used in quantum mechanics. Equation \eqref{l:SCHROD5} uses a speed adjustment term $\Psi$ defined in equation \eqref{l:APP18} and is only non-zero when the potential term depends on time explicitly. A time-dependent potential between the main and incident particles needs to account for the behaviour of the time dependence of $\int \rho V dx$ where $V=V(x,t)$ represents the potential. Also notice that the drift terms in equation \eqref{l:DRFT15} has the mass ratio incorporated into the potential width term since $M\gamma_m^2=m$. Finally the energy terms in \eqref{l:HAMILT4} are presented proportional to $\eta^2$ according to regular quantum mechanics. Though the weighting of the osmotic term exists in the drift term the wave-function and energy are exactly in line with quantum mechanics.

The first example in this paper models the behaviour of the main particle assuming that all interacting incident particles have a multiple of energies proportional to $\overline{e}=mc_w^2/2$. Clearly the Schr\"{o}dinger equation now demands that all energy solutions (associated energy Eigenvectors) for the main and incident particles produce multiples of $\overline{e}$. The joint energy for the main particle and incident system becomes $(k mc_w^2/2=k\overline{e}, k\geq 0)$ and it is possible that main system does not have any energy while the incident particle has energy $\overline{e}$. It demonstrates the start of a quantum system of main particles residing in a heatbath of particles with energies proportional to $\overline{e}$. Presumably it can be projected onto quantum mechanics if the energy assumes the appropriate distribution of energy assuming that $\hbar = \overline{e}\overline{\tau}$ where $\overline{\tau}$ is the usual average quantal interaction time. This argument will apply to photons as well though their masses are zero but the momentum $p=mc_w$ (photons $p=\hbar k$) and energy $\overline{e}$ (photon $\abs{k}c$) are well defined.

The two-particle energy in Theorem \eqref{l:THEOREM5} can be solved via Proposition \eqref{l:HAMILT12} and the Schr\"{o}dinger equation if the potential can be determined directly from Newtonian physics as witnessed for instance in the search for Dark Matter. The first term for the incident particles $V(x,t)$ was estimated from Newton while the incident osmotic term used acceleration. For diffusion then the Schr\"{o}dinger equation was employed to deal with two potentials simultaneously. The first term was estimated from Newtonian mechanics as per assumption that all incident particles - in this case randomly moving stars within the galactic cluster - are moving according to Newton. Then for the osmotic term the axial motion of the stellar system $r$ parsecs will be $\abs{\frac{w_2+w_1}{2}}^2 = \frac{GM_R}{r}$ where $G$ is the gravitational constant and $M_R$ is the estimated central mass of stars at the galactic cluster. A similar estimate for the incident particle osmotic term will look at the acceleration and then multiply with the stellar system speed to find that $\abs{\frac{w_2-w_1}{2}}^2 = \frac{G^2M_R^2}{r^4}$. Finally an approximation of the probability distribution is provided via the Schr\"{o}dinger equation in Proposition \eqref{l:HAMILT12}.

This analysis suggests that the Newtonian model for galaxies is not quite right as the osmotic terms in the incident term tends to change the radial attraction potential considerably as a result of diffusion. But Milgrom~\cite{MILGROM} is quite right to assume that the Newtonian model is not sufficient if there is substantial diffusion caused by incident particles. The model here is clearly showing additional attraction at smaller radius to the centre of the galaxy though the assumption that all galaxy mass is centred at the centre of the galaxy is clearly inaccurate. The mass distribution calculated around the galaxy has quite a different profile due to its own gravity and many authors have taken this into account, see Penrose~\cite{PENROSE1}. Notice that it is possible to adjust the theory presented by a weighting factor though the user has to be careful to adjust $\chi(r)$ in a self-consistent manner. The MOND theory is built on the basic assumptions of the modified dynamics, which were shown earlier to reproduce dynamical properties of galaxies and galaxy aggregates without having to assume the existence of hidden mass. Whether additional terms in the potential employing Schr\"{o}dinger's equation solves the Dark Matter absence remains to be seen.

There is another solution to the velocity equation in Theorem \eqref{l:THEOREM3} as presented in Theorem \eqref{l:THEOREM2} using Eigenvectors and Eigenvalues. The solution provides the average velocity field and restricts the size of the $b$ (determines $v_1, w_1$) and $b^\bot$ (determines $v_2, w_2$) solutions. An interesting application is Theorem \eqref{l:THEOREM3} which shows that all fluctuations around the average motion are protected before and after the collision. Once the average velocity $a$ has been subtracted the energy expression for both colliding particle has a term that depends on $\gamma_m^2a^T(b+b^\bot)$. If the average speed $a$ is not correlated with $w_1, w_2$ then the relation shows the Minkowski surface for relativistic particles. What is really quite interesting is that the Minkowksi equation holds immediately if the mass tends to zero so in other words when we are dealing with photons. But at positive mass where small amounts of energy will be exchanged with the main particle fall under quantal effects but the solution is described in terms of the equation involving the average speed $a$.

The preserved energy for the main and incident particle in \eqref{l:DRFT4} represents four weighted terms showing the average motion and the osmotic contributions with the incident particle's contribution so far represented as a fixed potential, However, given the form of the incident particle it is clear that it can be approximated in the same way as the main particle so there should be two Schr\"{o}dinger equations for both particles at the same time. As Theorem \eqref{l:THEOREM8} shows there are two Schr\"{o}dinger equations represent the fluctuations that are being sent and received but the two equations have different diffusion constants. The quantal process providing energy has a much smaller diffusion constant since $\eta_w =\gamma_m^4\eta$ plus this Schr\"{o}dinger has a negative potential term showing the energy the incident particle is providing to the main solution. The main particle Schr\"{o}dinger equation shows a similar potential receiving energy from the incident particle while the incident particle is delivering energy.

Notice that for the Schr\"{o}dinger equations in Theorem \eqref{l:THEOREM8} the variance $\eta=M^{3/2}\sigma^2/m^{1/2}$ describes the main particle with mass $M$ receiving energy from the incident particle with mass $m$ while $\eta_w=m^{3/2}\sigma^2/M$ shows the incident particle mass $m$ providing energy to the main particle $M$. So far in this paper no effort has been made to prove existence of the solutions nor have any examples of this case provided. Ultimately, these two equations show a 6 dimensional equation for both solutions at the same time and Theorem \eqref{l:THEOREM8} shows the form of the equations together with the simple energy condition. This line of research might be able to prove Dirac's equation while pursuing the relativistic ideas of Serva~\cite{SERVA1} or Guerra~\cite{GUERRA1}.



\bibliographystyle{phcpc}
\bibliography{xbib}

\begin{thebibliography}{10}

\bibitem{EINST1}
Einstein, A.,
\newblock {\em Investigations on the Theory of The Brownian Movement},
\newblock Dover, New York, 1956.

\bibitem{GOLDSTEIN1}
Goldstein, S., Lebowitz, J., and Ravishankar, K.,
\newblock Communications in Mathematical Physics {\bf 85} (1982) 419.

\bibitem{GARBA1}
Garbaczewski, P.,
\newblock arXiv:cond-mat/9809288  (22Sep 1998).

\bibitem{POSILICANO1}
Posilicano, A. and Ugolini, S.,
\newblock arXiv:math.PR/0212020  (2Dec 2002).

\bibitem{KAMPEN1}
van Kampen, N.~G.,
\newblock {\em Stochastic Processes in Physics and Chemistry},
\newblock North Holland, Amsterdam, 1992.

\bibitem{GAMBA1}
Gamba, I.~M., Rjasanow, S., and Wagner, W.,
\newblock Mathematical and Computer Modelling {\bf 42} (2005) 683.

\bibitem{ENELSON3}
Nelson, E.,
\newblock Physical Review {\bf 150} (1966) 1079.

\bibitem{ECARL2}
Carlen, E.,
\newblock Stochastic mechanics: A look back and a look ahead,
\newblock Dedicated to Professor Edward Nelson on his Seventy Second Birthday,
  2005.

\bibitem{ENELSON1}
Nelson, E.,
\newblock {\em Quantum Fluctuations},
\newblock Princeton University Press, 1985.

\bibitem{ENELSON2}
Nelson, E.,
\newblock {\em Dynamical Theories of Brownian Motion},
\newblock Princeton University Press, 1967.

\bibitem{WALLSTROM1}
Wallstrom, T.~J.,
\newblock Physyal Review A {\bf 3} (1994) 1613.

\bibitem{PENROSE1}
Penrose, R.,
\newblock Phil.Trans.R.Soc. (Lond) A {\bf 356} (1998) 1927.

\bibitem{MIGROM1}
Migrom, M.,
\newblock The Astrophysical Journal {\bf 270} (1983) 365.

\bibitem{MILGROM}
Milgrom, M.,
\newblock Astrophysical Journal, Part 1 (ISSN 0004-637X) {\bf 286} (1984) 7.

\bibitem{ECARL1}
Carlen, E.,
\newblock {\em Existence of Stochastic Processes in Stoch. Mech.},
\newblock PhD thesis, Princeton University,Department of Physics, 1985.

\bibitem{ECARL3}
Carlen, E.,
\newblock Communications in Mathematical Physics {\bf 94} (1984) 293.

\bibitem{GUERRA1}
Guerra, F.,
\newblock {\em The Problem of the Physical Interpretation of Nelson Stochastic
  Mechanics as a Model for Quantum Mechanics}, volume The Foundation of Quantum
  Mechanics,
\newblock Kluwer, Amsterdam, 1994.

\bibitem{GUERRA2}
Guerra, F.,
\newblock The Foundation of Quantum Mechanics {\bf 17} (1985) 305.

\bibitem{ROBERTSHAW1}
Robertshaw, O.,
\newblock arxiv.org/abs/math-ph/0509066  (28Sep 2018).

\bibitem{HARRISON1}
R.~Harrison, I.~Moroz, K.~T.,
\newblock Nonlinearity {\bf 16} (2003) 101.

\bibitem{TOD1}
I.M.~Moroz, K.~T.,
\newblock Nonlinearity {\bf 12} (1999) 201.

\bibitem{FEYNMAN1}
Feynman, R.,
\newblock {\em QED: The Strange Theory of Light and Matter},
\newblock Princeton University Press, 1985.

\bibitem{SERVA1}
Serva, M.,
\newblock Annales de l'I.H.P.,Section A, {\bf 49 no. 4} (1988) 415.

\end{thebibliography}

\appendix
\section{}
\label{l:AppendixA} \noindent {\bf Proof of Theorem
\eqref{l:THEOREM3}}
\begin{proof}
Figure \ref{COLL3} clearly shows the elastic collision at point $x$ for $v_1, w_1, v_2, w_2$ dividing the particle collision along $\phi$ with the remaining motion orthogonal to $\phi$. This proof assumes a perfect collision along the $\phi$ axis and assumes that the remaining momenta and energies are unaltered.

For the proof separate the motion of n-dimensional velocities $v_1, w_1, v_2, w_2$ along $\phi$ with $P(\phi)v_1$, $P(\phi)w_1$, $P(\phi)v_2$ and $P(\phi)w_2$ and present its perpendicular motion using $(I-P(\phi))v_1$, $(I-P(\phi))w_1$, $(I-P(\phi))v_2$ and $(I-P(\phi))w_2$. Then for the collision clearly
\begin{align}
\label{l:ENERG4}
\begin{split}
v_1 & = (\phi\phi^T)v_1 + (v_1 - (\phi\phi^T)v_1)=P(\phi)v_1+\left(I-P(\phi)\right)v_1,
\\
w_1 & = (\phi\phi^T)w_1 + (w_1 - (\phi\phi^T)w_1)=P(\phi)w_1+\left(I-P(\phi)\right)w_1,
\\
v_2  & = (\phi\phi^T)v_2 + (v_2 - (\phi\phi^T)v_2)=P(\phi)v_2+\left(I-P(\phi)\right)v_2,
\\
w_2  & = (\phi\phi^T)w_2 + (w_2 - (\phi\phi^T)w_2)=P(\phi)w_2+\left(I-P(\phi)\right)w_2,
\end{split}
\end{align}
and the total momentum before and after the collision can be written as
\begin{align}
\label{l:ENERG1}
\begin{split}
\mathcal{M}_1(x) & = Mv_1 + mw_1
\\
& = MP(\phi)v_1+M(1-P(\phi))v_1 +
mP(\phi)w_1+m(1-P(\phi))w_1,
\\
\mathcal{M}_2(x) & = Mv_2 + mw_2
\\
& = MP(\phi)v_2+M(1-P(\phi))v_2 +
mP(\phi)w_2+m(1-P(\phi))w_2.
\end{split}
\end{align}

The pre-collision energy $\h_{1}(x)$ is given by $M\abs{v_1}^2$ (multiplied with factor 2) which decomposed in $P(\phi)$ and $(I-P(\phi))$ becomes
\begin{align}
\label{l:ENERG15}
\begin{split}
2\h_{1}(x) = & M\abs{v_1}^2 + m\abs{w_1}^2
\\
= & M\abs{P(\phi)v_1+(I-P(\phi))v_1}^2 +  m\abs{P(\phi)w_1+(I-P(\phi))w_1}^2
\\
= & Mv_1^TP(\phi)P(\phi)v_1 +Mv_1^T(I-P(\phi))(I-P(\phi))v_1
\\
& + mw_1^TP(\phi)P(\phi)w_1 +mw_1^T(I-P(\phi))(I-P(\phi))w_1,
\\
 = & Mv_1^TP(\phi)v_1 +Mv_1^T(I-P(\phi))v_1 +
 mw_1^TP(\phi)w_1 +mw_1^T(I-P(\phi))w_1,
\end{split}
\end{align}
because $P(\phi)$ and $(I-P(\phi))$ are orthogonal projections so that $P(\phi)P(\phi)=P(\phi)$ and
$(I-P(\phi))(I-P(\phi))=(I-P(\phi))$. Similarly using the same projections the after-collision energy $\h_2(x)$ becomes
\begin{align}
\label{l:ENERG2}
\begin{split}
2\h_{2}(x) & = M\abs{v_2}^2+m\abs{w_2}^2
\\
& = Mv_2^TP(\phi)v_2+Mv_2^T(1-P(\phi))v_2 +
mw_2^TP(\phi)w_2+mw_2^T(1-P(\phi))w_2.
\end{split}
\end{align}

Since the motions $v_1-(\phi\phi^T)v_1 = (1-P(\phi))v_1$, $w_1-(\phi\phi^T)w_1=(1-P(\phi))w_1$
are orthogonal to the collision surface this part of the motion must remain unchanged during the collision. Hence
\begin{align}
\label{l:ENERG17}
\begin{split}
(1-P(\phi))v_2 & = (1-P(\phi))v_1,
\\
(1-P(\phi))w_2 & = (1-P(\phi))w_1,
\end{split}
\end{align}
while $\phi^Tv_1, \phi^Tw_1$ collide with each other along $\phi$ yielding $\phi^Tv_2, \phi^Tw_2$. Putting \eqref{l:ENERG17} back into $\eqref{l:ENERG4}$ results in
\begin{align*}
v_1 & = P(\phi)v_1 + (I - P(\phi))v_1,
\\
w_1 & = P(\phi)w_1 + (I - P(\phi))w_1,
\\
v_2  & = P(\phi)v_2 + (I - P(\phi))v_1,
\\
w_2  & = P(\phi)w_2 + (I - P(\phi))w_1,
\end{align*}
and that means that \eqref{l:ENERG1},  \eqref{l:ENERG15} and \eqref{l:ENERG2} change to
\begin{align*}
\mathcal{M}_1(x) & = Mv_1 + mw_1
\\
& = MP(\phi)v_1+ mP(\phi)w_1 + M(1-P(\phi))v_1
+m(1-P(\phi))w_1,
\\
\mathcal{M}_2(x) & = Mv_2 + mw_2
\\
& = MP(\phi)v_2 +
mP(\phi)w_2 + M(1-P(\phi))v_1 + m(1-P(\phi))w_1,
\\
2\h_{1}(x) & = M\abs{v_1}^2+m\abs{w_1}^2
\\
& = Mv_1^TP(\phi)v_1 + mw_1^TP(\phi)w_1 +
Mv_1^T(1-P(\phi))v_1 + mw_1^T(1-P(\phi))w_1,
\\
2\h_{2}(x) & = M\abs{v_2}^2 + m\abs{w_2}^2
\\
& = Mv_2^TP(\phi)v_2 +
mw_2^TP(\phi)w_2 + Mv_1^T(1-P(\phi))v_1 +mw_1^T(1-P(\phi))w_1,
\end{align*}
and with $\mathcal{M}_1(x) = \mathcal{M}_2(x)$ and $\h_1(x)=\h_2(x)$ this expression becomes
\begin{align*}
\begin{split}
MP(\phi)v_1+ mP(\phi)w_1 & =
MP(\phi)v_2 + mP(\phi)w_2,
\\
Mv_1^TP(\phi)v_1 + mw_1^TP(\phi)w_1
 & = Mv_2^TP(\phi)v_2 + mw_2^TP(\phi)w_2,
\end{split}
\end{align*}
or using $P(\phi)P(\phi)=P(\phi)=\phi^T\phi$ and reducing $P(\phi)$ to the direction $\phi$ becomes
\begin{align}
\label{l:APP3}
\begin{split}
M\phi^Tv_1+ m\phi^Tw_1 & =
M\phi^Tv_2 + m\phi^Tw_2,
\\
M(\phi^Tv_1)^2 + m(\phi^Tw_1)^2
 & = M(\phi^Tv_2)^2 + m(\phi^Tw_2)^2.
\end{split}
\end{align}

This means balancing the initial momenta $M(\phi^Tv_1)$, $m(\phi^Tw_1)$ and initial energies $M(\phi^Tv_1)^2$, $m(\phi^Tw_1)^2$ against the final momenta $M(\phi^Tv_2)$, $m(\phi^Tw_2)$ and final energies $M(\phi^Tv_2)^2$, $m(\phi^Tw_2)^2$ while the orthogonal motion does not change with the collision. This is a one-dimensional collision and the following proposition shows the form for $\phi^Tv_2, \phi^Tv_2$.

\begin{prop}
\label{l:PROP1} The solution to the two equations in \eqref{l:APP3} equals
\begin{align}
\label{l:EQMOT1}
\begin{split}
\begin{pmatrix}
\phi^Tv_2 \\
\phi^Tw_2
\end{pmatrix}
&= \begin{pmatrix}
\cos(\theta) & \gamma_m\sin(\theta) \\
\frac{\sin(\theta)}{\gamma_m} & -\cos(\theta)
\end{pmatrix}
\begin{pmatrix}
\phi^Tv_1\\
\phi^Tw_1
\end{pmatrix}
= \Gamma_\theta
\begin{pmatrix}
\phi^Tv_1\\
\phi^Tw_1
\end{pmatrix},
\end{split}
\end{align}
defining $\Gamma_\theta$ with $\gamma_m^2=m/M, \sin\left(\theta\right) =2\gamma_m/\left(1+\gamma_m^2\right)$, $\cos\left(\theta\right)=\left(1-\gamma_m^2\right)/\left(1+\gamma_m^2\right)$ as defined in Theorem
\eqref{l:THEOREM3} above.

\begin{proof}
From the conservation of momentum using equations \eqref{l:APP3} it is clear that
\begin{align*}
MP(\phi)v_2- MP(\phi)v_1=mP(\phi)w_1- mP(\phi)w_2,
\end{align*}
so that
\begin{align*}
M\phi^Tv_2- M\phi^Tv_1=m\phi^Tw_1- m\phi^Tw_2.
\end{align*}
Dividing this by the mass $M$ yields
\begin{align}
\label{l:REL13}
\phi^T(v_2 - v_1) = - \gamma_m^2\phi^T(w_2 - w_1),
\end{align}
so that
\begin{align*}
\phi^Tv_2 + \gamma_m^2\phi^Tw_2 =\phi^Tv_1 + \gamma_m^2\phi^Tw_1,
\end{align*}
using the definition of the mass ratio $\gamma_m$ in the Theorem.

From the energy distribution \eqref{l:APP3} it is clear that
\begin{align*}
Mv_2^TP(\phi)v_2 + mw_2^TP(\phi)w_2
=Mv_1^TP(\phi)v_1 + mw_1^TP(\phi)w_1,
\end{align*}
so that
\begin{align*}
M(\phi^Tv_2)^2+m(\phi^Tw_2)^2
=M(\phi^Tv_1)^2+m(\phi^Tw_1)^2,
\end{align*}
which is equivalent to
\begin{align*}
M((\phi^Tv_2)^2-(\phi^Tv_1)^2)
=-m((\phi^Tw_2)^2-(\phi^Tw_1)^2),
\end{align*}
or
\begin{align*}
M(\phi^Tv_2-\phi^Tv_1)(\phi^Tv_2+\phi^Tv_1)
=-m(\phi^Tw_2-\phi^Tw_1)(\phi^Tw_2+\phi^Tw_1).
\end{align*}
Dividing by $M$ then provides
\begin{align*}
(\phi^Tv_2-\phi^Tv_1)(\phi^Tv_2+\phi^Tv_1)
=-\gamma_m^2(\phi^Tw_2-\phi^Tw_1)(\phi^Tw_2+\phi^Tw_1).
\end{align*}

However, the term $\gamma_m^2\phi^T(w_2-w_1)$ is given by \eqref{l:REL13} above and can be substituted into the righthand side. Hence
\begin{align*}
(\phi^Tv_2-\phi^Tv_1)(\phi^Tv_2+\phi^Tv_1)
=(\phi^Tv_2-\phi^Tv_1)(\phi^Tw_2+\phi^Tw_1),
\end{align*}
and
\begin{align*}
(\phi^Tv_2-\phi^Tv_1)\left((\phi^Tv_2+\phi^Tv_1) - (\phi^Tw_2+\phi^Tw_1)\right)=0,
\end{align*}
establishing
\begin{align}
\label{l:REL14}
\phi^Tv_2+\phi^Tv_1 = \phi^Tw_2+\phi^Tw_1,
\end{align}
because this is a one-dimensional collision with $\phi^Tv_2 \neq \phi^Tv_1$. If $\phi^Tv_2 = \phi^Tv_1$ then $P(\phi)^Tv_2 = P(\phi)^Tv_1$ which means that $v_2=v_1$ by \eqref{l:ENERG17} and therefore $w_2=w_1$ so there is no interaction.

Combining equations \eqref{l:REL13} and \eqref{l:REL14} it is clear that
\begin{align*}
\phi^Tv_2 + \gamma_m^2\phi^Tw_2 & =\phi^Tv_1 + \gamma_m^2\phi^Tw_1,
\\
\phi^Tv_2-\phi^Tw_2& =-\phi^Tv_1  +\phi^Tw_1,
\end{align*}
or in matrix form
\begin{align*}
 \begin{pmatrix}
1 & \gamma_m^2 \\
1 & -1
\end{pmatrix}
\begin{pmatrix}
\phi^Tv_2
\\
\phi^Tw_2
\end{pmatrix}
=
 \begin{pmatrix}
1 & \gamma_m^2 \\
-1 & 1
\end{pmatrix}
\begin{pmatrix}
\phi^Tv_1
\\
\phi^Tw_1
\end{pmatrix}.
\end{align*}
Inverting the first matrix and multiplying the matrices shows
\begin{align*}
\begin{pmatrix}
\phi^Tv_2
\\
\phi^Tw_2
\end{pmatrix}
=&
 \begin{pmatrix}
\frac{1-\gamma_m^2}{1+\gamma_m^2} & \frac{2\gamma_m^2}{1+\gamma_m^2} \\
\frac{2}{1+\gamma_m^2} & -\frac{1-\gamma_m^2}{1+\gamma_m^2}
\end{pmatrix}
\begin{pmatrix}
\phi^Tv_1
\\
\phi^Tw_1
\end{pmatrix}
=
\begin{pmatrix}
\cos(\theta) & \gamma_m\sin(\theta) \\
\frac{\sin(\theta)}{\gamma_m} & -\cos(\theta)
\end{pmatrix}
\begin{pmatrix}
\phi^Tv_1
\\
\phi^Tw_1
\end{pmatrix},
\end{align*}
using the definitions from the Theorem \eqref{l:THEOREM3}. This proves Proposition \eqref{l:PROP1}.
\end{proof}
\end{prop}

It is relatively straightforward to extend this to the n-dimensional velocities $v_1$, $w_1$, $v_2$ and $w_2$ using the parts of the equation not involved in the collision. To show this multiply the equations \eqref{l:EQMOT1} with the vector $\phi$ to derive
\begin{align*}
\begin{pmatrix}
P(\phi)v_2 \\
P(\phi)w_2
\end{pmatrix}
&= \begin{pmatrix}
\cos(\theta) & \gamma_m\sin(\theta) \\
\frac{\sin(\theta)}{\gamma_m} & -\cos(\theta)
\end{pmatrix}
\begin{pmatrix}
P(\phi)v_1\\
P(\phi)w_1
\end{pmatrix}
=
\Gamma_\theta
\begin{pmatrix}
P(\phi)v_1\\
P(\phi)w_1
\end{pmatrix}.
\end{align*}
Then add $(1-P(\phi))v_2$, $(1-P(\phi))w_2$ on both sides and add the factors $(1-P(\phi))v_1-(1-P(\phi))v_1=0$ and $(1-P(\phi))w_1-(1-P(\phi))w_1=0$ under the vector yielding the following multi-dimensional version
\begin{align*}
\begin{pmatrix}
P(\phi)v_2 + (1-P(\phi))v_2\\
P(\phi)w_2 + (1-P(\phi))w_2
\end{pmatrix}
&=
\begin{pmatrix}
(1-P(\phi))v_2\\
(1-P(\phi))w_2
\end{pmatrix}
+
\\
& \begin{pmatrix}
\cos(\theta) & \gamma_m\sin(\theta) \\
\frac{\sin(\theta)}{\gamma_m} & -\cos(\theta)
\end{pmatrix}
\begin{pmatrix}
P(\phi)v_1+(1-P(\phi))v_1-(1-P(\phi))v_1\\
P(\phi)w_1+(1-P(\phi))w_1-(1-P(\phi))w_1
\end{pmatrix},
\end{align*}
which by collision rules and $v_1=P(\phi)v_1+(1-P(\phi))v_1$, $w_1=P(\phi)w_1+(1-P(\phi))w_1$ is equivalent to
\begin{align}
\label{l:REL17}
\begin{split}
\begin{pmatrix}
v_2\\
w_2
\end{pmatrix}
& =
\begin{pmatrix}
(1-P(\phi))v_1\\
(1-P(\phi))w_1
\end{pmatrix}
+
\begin{pmatrix}
\cos(\theta) & \gamma_m\sin(\theta) \\
\frac{\sin(\theta)}{\gamma_m} & -\cos(\theta)
\end{pmatrix}
\begin{pmatrix}
v_1\\
w_1
\end{pmatrix}
-
\begin{pmatrix}
\cos(\theta) & \gamma_m\sin(\theta) \\
\frac{\sin(\theta)}{\gamma_m} & -\cos(\theta)
\end{pmatrix}
\begin{pmatrix}
(1-P(\phi))v_1\\
(1-P(\phi))w_1
\end{pmatrix}
\\
& =
\begin{pmatrix}
\cos(\theta) & \gamma_m\sin(\theta) \\
\frac{\sin(\theta)}{\gamma_m} & -\cos(\theta)
\end{pmatrix}
\begin{pmatrix}
v_1\\
w_1
\end{pmatrix}
+
\begin{pmatrix}
\gamma_m\sin(\theta) & -\gamma_m\sin(\theta)
\\
-\frac{\sin(\theta)}{\gamma_m} & \frac{\sin(\theta)}{\gamma_m}
\end{pmatrix}
\begin{pmatrix}
(1-P(\phi))v_1\\
(1-P(\phi))w_1
\end{pmatrix}
\\
& =
\begin{pmatrix}
\cos(\theta) & \gamma_m\sin(\theta) \\
\frac{\sin(\theta)}{\gamma_m} & -\cos(\theta)
\end{pmatrix}
\begin{pmatrix}
v_1\\
w_1
\end{pmatrix}
+
\begin{pmatrix}
\gamma_m\sin(\theta) (1-P(\phi))(v_1-w_1)\\
-\frac{\sin(\theta)}{\gamma_m} (1-P(\phi))(v_1-w_1)
\end{pmatrix}
\\
& =
\begin{pmatrix}
\cos(\theta) & \gamma_m\sin(\theta) \\
\frac{\sin(\theta)}{\gamma_m} & -\cos(\theta)
\end{pmatrix}
\begin{pmatrix}
v_1\\
w_1
\end{pmatrix}
+
\begin{pmatrix}
\gamma_m\Phi\\
-\frac{1}{\gamma_m}\Phi
\end{pmatrix}
\\
& =
\Gamma_\theta
\begin{pmatrix}
v_1\\
w_1
\end{pmatrix}
+(I-\Gamma_\theta)
\begin{pmatrix}
\gamma_m\Phi\\
-\frac{1}{\gamma_m}\Phi
\end{pmatrix},
\end{split}
\end{align}
combining the $(1-P(\phi))v_1$, $(1-P(\phi))w_1$ terms and using $1-\cos(\theta)=\gamma_m\sin(\theta)$ as well as $1+\cos(\theta)=\sin(\theta)/\gamma_m$. From the third line of this statement also
\begin{align}
\label{l:DRFT20}
\begin{split}
\begin{pmatrix}
v_2\\
w_2
\end{pmatrix}
& =
\begin{pmatrix}
\cos(\theta) & \gamma_m\sin(\theta)
\\
\frac{\sin(\theta)}{\gamma_m} & -\cos(\theta)
\end{pmatrix}
\begin{pmatrix}
v_1\\
w_1
\end{pmatrix}
+
\begin{pmatrix}
\gamma_m\sin(\theta) (1-P(\phi))(v_1-w_1)
\\
-\frac{\sin(\theta)}{\gamma_m} (1-P(\phi))(v_1-w_1)
\end{pmatrix}
\\
& =
\begin{pmatrix}
\cos(\theta) + \gamma_m\sin(\theta) (1-P(\phi)) & \gamma_m\sin(\theta)-\gamma_m\sin(\theta) (1-P(\phi))
\\
\frac{\sin(\theta)}{\gamma_m} -\frac{\sin(\theta)}{\gamma_m} (1-P(\phi)) & -\cos(\theta)+\frac{\sin(\theta)}{\gamma_m} (1-P(\phi))
\end{pmatrix}
\begin{pmatrix}
v_1\\
w_1
\end{pmatrix}
\\
& =
\begin{pmatrix}
I - \gamma_m\sin(\theta)P(\phi) & \gamma_m\sin(\theta)P(\phi)
\\
\frac{\sin(\theta)}{\gamma_m}P(\phi) & I-\frac{\sin(\theta)}{\gamma_m}P(\phi)
\end{pmatrix}
\begin{pmatrix}
v_1\\
w_1
\end{pmatrix}.
\end{split}
\end{align}
Both equations \eqref{l:REL17} and \eqref{l:DRFT20} are represented in equation \eqref{l:EQMOT9}

Notice that by definition $\Phi^T(v_2-v_1)=0$ as a result. This is easy to show using the proper definition for $v_2(t)$ in \eqref{l:REL17} and \eqref{l:DRFT20} and noticing that
\begin{align*}
v_2-v_1
& =(\cos(\theta)-1)v_1 + \gamma_m\sin(\theta)w_1+\gamma_m\Phi
\\
& =\gamma_m\sin(\theta)(w_1 -v_1)+\gamma_m\Phi.
\end{align*}
Now if $\Phi^T(v_2-v_1)=0$ then
\begin{align*}
0 =\gamma_m\sin(\theta)\Phi^T(w_1 -v_1)+\gamma_m\abs{\Phi}^2,
\end{align*}
with $\abs{\Phi}^2=\Phi^T\Phi$ but this is correct because
\begin{align*}
\gamma_m \sin(\theta) & \Phi^T(w_1 -v_1)+\gamma_m\abs{\Phi}^2
\\
= & -\gamma_m\sin^2(\theta)(v_1 -w_1)^T(1-P(\phi))(v_1 -w_1)
\\
& +\gamma_m\sin^2(\theta)(v_1 -w_1)^T(1-P(\phi))(1-P(\phi))(v_1 -w_1)
\\
= & -\gamma_m\sin^2(\theta)(v_1 -w_1)^T(1-P(\phi))(v_1 -w_1)
\\
& +\gamma_m\sin^2(\theta)(v_1 -w_1)^T(1-P(\phi))(v_1 -w_1)
=0.
\end{align*}
This establishes the proof for Theorem \eqref{l:THEOREM3}.
\end{proof}

\section{}
\label{l:AppendixB} \noindent {\bf Proof of Theorem
\eqref{l:THEOREM2}
\begin{proof}
Equation \eqref{l:EQMOT9} is a linear equation and its solution can be easily converted into Eigenvectors $(a,a)$, $(-\gamma_m^2b,b)$ with associated Eigenvalues $1,-1$. Since $v_1,v_2,w_1$ and $w_2$ are n-dimensional the velocity vectors $(a,a)$, $(-\gamma_m^2b,b)$ are 2n-dimensional. Then substitute
\begin{align*}
\begin{pmatrix}
v_1
\\
w_1
\end{pmatrix}=
\begin{pmatrix}
a
\\
a
\end{pmatrix}
+
\begin{pmatrix}
-\gamma_m^2b
\\
b
\end{pmatrix},
\end{align*}
so that
\begin{align*}
\begin{pmatrix}
v_2 \\
w_2
\end{pmatrix}
&=
\begin{pmatrix}
\cos(\theta) & \gamma_m\sin(\theta) \\
\frac{\sin(\theta)}{\gamma_m} & -\cos(\theta)
\end{pmatrix}
\begin{pmatrix}
v_1\\
w_1
\end{pmatrix}
+
\begin{pmatrix}
\gamma_m \Phi
\\
-\frac{1}{\gamma_m} \Phi
\end{pmatrix}
\\
&=
\begin{pmatrix}
a
\\
a
\end{pmatrix}
-
\begin{pmatrix}
-\gamma_m^2b
\\
b
\end{pmatrix}
+
\begin{pmatrix}
\gamma_m \Phi
\\
-\frac{1}{\gamma_m} \Phi
\end{pmatrix}
=
\begin{pmatrix}
a
\\
a
\end{pmatrix}
+
\begin{pmatrix}
\gamma_m^2(b+\frac{1}{\gamma_m}\Phi)
\\
-(b+\frac{1}{\gamma_m}\Phi)
\end{pmatrix}
=\begin{pmatrix}
a+\gamma_m^2b^\bot
\\
a-b^\bot
\end{pmatrix},
\end{align*}
which shows equation \eqref{l:EQMOT4}. In addition, consider the total energy $\h_1, \h_2$ conserved through the collision
\begin{align*}
\frac{2}{M}\h_1=(1+\gamma_m^2)\abs{a}^2 +\gamma_m^2(1+\gamma_m^2)\abs{b}^2 = (1+\gamma_m^2)\abs{a}^2 +\gamma_m^2(1+\gamma_m^2)\abs{b^\bot}^2=\frac{2}{M}\h_2,
\end{align*}
then clearly
\begin{align*}
\gamma_m^2(1+\gamma_m^2)\abs{b}^2 = \gamma_m^2(1+\gamma_m^2)\abs{b^\bot}^2,
\end{align*}
and $\abs{b}^2=\abs{b^\bot}^2$. This proves Theorem \eqref{l:THEOREM2}.
\end{proof}

\section{}
\label{l:AppendixC} \noindent {\bf Proof of Proposition \eqref{l:THEOREM4}
\begin{proof}
Applying equation \eqref{l:EQMOT9} in Theorem \eqref{l:THEOREM3} for velocity $v_2$ it is clear that
yields
\begin{align*}
v_2  & = (\phi\phi^T)v_2 + (v_1 - (\phi\phi^T)v_1)
\\
& = (v_1 - (\phi\phi^T)v_1) + \phi\left(\cos\left(\theta\right) (\phi^Tv_1)
+ \gamma_m \sin\left(\theta\right)(\phi^Tw_1)\right)
\\
& = v_1 + \phi\left((\cos\left(\theta\right)-1)(\phi^Tv_1)
+ \gamma_m \sin\left(\theta\right)(\phi^Tw_1)\right)
\\
& = v_1 + \phi\gamma_m \sin\left(\theta\right)
\left(-(\phi^Tv_1) + (\phi^Tw_1)\right)
\\
& = v_1 + \gamma_m \sin\left(\theta\right)
\left(P(\phi)w_1 - P(\phi)v_1\right)=v_1 +
\gamma_m\sin\left(\theta\right)P(\phi)(w_1-v_1),
\end{align*}
and
\begin{align*}
w_2  & = (\phi\phi^T)w_2 + (w_1 -
(\phi\phi^T)w_1)
\\
& = (w_1 - (\phi\phi^T)w_1) + \phi\left(\frac{\sin\left(\theta\right)}{\gamma_m} (\phi^Tv_1)
- \cos\left(\theta\right)(\phi^Tw_1)\right)
\\
& = w_1+ \phi\left(\frac{\sin\left(\theta\right)}{\gamma_m}
(\phi^Tv_1) - (\cos\left(\theta\right)+1)(\phi^Tw_1)\right)
\\
& = w_1+ \phi\frac{\sin\left(\theta\right)}{\gamma_m}\left(
(\phi^Tv_1) - (\phi^Tw_1)\right)
\\
& = w_1+ \frac{\sin\left(\theta\right)}{\gamma_m}\left(
P(\phi)v_1 - P(\phi)w_1\right)= w_1-
\frac{1}{\gamma_m}\sin\left(\theta\right)P(\phi)(w_1-v_1),
\end{align*}
combining $\phi\phi^Tv_1=P(\phi)v_1$,
$\phi\phi^Tw_1=P(\phi)w_1$. This proves Proposition \eqref{l:THEOREM4}.
\end{proof}

\section{}
\label{l:AppendixE} \noindent {\bf Proof of Theorem
\eqref{l:THEOREM5}}
\begin{proof}
This equation can of course be verified by direct substitution however the following simple argument is more intuitive. First notice that from the definitions $\sin(\theta) = 2\gamma_m/(1+\gamma_m^2)$ and $\cos(\theta)=(1-\gamma_m^2)/(1+\gamma_m^2)$ it is straighforward that
\begin{align*}
\sin\left(\theta/2\right)=\frac{\gamma_m}{\sqrt{1+\gamma_m^2}},
\cos\left(\theta/2\right)=\frac{1}{\sqrt{1+\gamma_m^2}},
\end{align*}
so that
\begin{align*}
\frac{\sin\left(\theta/2\right)}{\cos\left(\theta/2\right)}=\gamma_m.
\end{align*}

To consider the sum of the squared n-dimensional differences $v_2-v_1$ and $v_2+v_1$ for the main particle consider first the $v_2-v_1$ difference. Using \eqref{l:REL17} of Appendix\eqref{l:AppendixA} it is found that
\begin{align*}
\frac{v_2-v_1}{2}
&=\frac{1}{2} (\cos(\theta)-1)v_1+\frac{1}{2}\gamma_m\sin(\theta)w_1
+\frac{\gamma_m}{2}\Phi
\\
&=\sin\left(\theta/2\right)\big(-\sin\left(\theta/2\right)v_1
+\gamma_m\cos\left(\theta/2\right)w_1\big)
+\frac{\gamma_m}{2}\Phi
\\
&=\sin\left(\theta/2\right)A_p
+\frac{\gamma_m}{2}\Phi,
\end{align*}
with
\begin{align*}
A_p & = -\sin\left(\theta/2\right)v_1
+\gamma_m\cos\left(\theta/2\right)w_1
\\
& = \sin\left(\theta/2\right)\left(-v_1
+w_1\right).
\end{align*}
But we know that $\Phi^T(v_2-v_1)=0$ by Theorem
\eqref{l:THEOREM3} so then
\begin{align*}
\sin\left(\theta/2\right)\Phi^TA_p
+\frac{\gamma_m}{2}\abs{\Phi}^2=0,
\end{align*}
or dividing by $\gamma_m$
\begin{align*}
\cos\left(\theta/2\right)\Phi^TA_p
+\frac{1}{2}\abs{\Phi}^2=0.
\end{align*}
As a result
\begin{align*}
\frac{v_2-v_1}{2\gamma_m}
&=\cos\left(\theta/2\right)A_p +\frac{1}{2}\Phi,
\end{align*}
and
\begin{align}
\label{l:APP8}
\begin{split}
\abs{\frac{v_2-v_1}{2\gamma_m}}^2
&=\cos^2\left(\theta/2\right)\abs{A_p}^2
+\cos\left(\theta/2\right)\Phi^TA_p
+\frac{1}{4}\abs{\Phi}^2
\\
&=\cos^2\left(\theta/2\right)\abs{A_p}^2 -\frac{1}{4}\abs{\Phi}^2.
\end{split}
\end{align}

To get a similar expression for the average positive difference use \eqref{l:REL17} from Appendix\eqref{l:AppendixA} again to show that
\begin{align*}
\frac{v_2+v_1}{2}
& =\frac{1}{2} (\cos(\theta)+1)v_1+\frac{1}{2}\gamma_m\sin(\theta)w_1
+\frac{\gamma_m}{2}\Phi
\\
& = \cos\left(\theta/2\right)\big(\cos\left(\theta/2\right)v_1
+\gamma_m\sin\left(\theta/2\right)w_1\big) + \frac{\gamma_m}{2}\Phi
\\
& = \cos\left(\theta/2\right)B_p + \frac{\gamma_m}{2}\Phi,
\end{align*}
with
\begin{align*}
B_p & =\cos\left(\theta/2\right)v_1
+\gamma_m\sin\left(\theta/2\right)w_1
\\
& =\cos\left(\theta/2\right)\left(v_1
+\gamma_m^2w_1\right).
\end{align*}
Hence the difference squared becomes
\begin{align}
\label{l:APP9}
\abs{\frac{v_2+v_1}{2}}^2
& = \cos\left(\theta/2\right)^2\abs{B_p}^2 +
\gamma_m\cos\left(\theta/2\right)\Phi^TB_p+
\frac{\gamma_m^2}{4}\abs{\Phi}^2.
\end{align}

To get the sum of the squared differences add equation \eqref{l:APP8} and equation \eqref{l:APP9} to find
\begin{align*}
E_1 & =\abs{\frac{v_2+v_1}{2}}^2+\frac{1}{\gamma_m^2}
\abs{\frac{v_2-v_1}{2}}^2
\\
& =\cos\left(\theta/2\right)^2\left(\abs{A_p}^2 + \abs{B_p}^2\right)
+\gamma_m\cos\left(\theta/2\right)\Phi^TB_p+
\frac{\gamma_m^2-1}{4}\abs{\Phi}^2
\\
&=\cos\left(\theta/2\right)^2
\left(\abs{v_1}^2+\gamma_m^2\abs{w_1}^2\right)
+\gamma_m\cos\left(\theta/2\right)\Phi^TB_p+
\frac{\gamma_m^2-1}{4}\abs{\Phi}^2,
\end{align*}
because
\begin{align*}
\abs{A_p}^2 + \abs{B_p}^2
= &
\sin^2\left(\theta/2\right)\left(\abs{v_1}^2-2v_1^Tw_1+\abs{w_1}^2\right)
\\
& +\cos^2\left(\theta/2\right)\left(\abs{v_1}^2+2\gamma_m^2v_1^Tw_1+\gamma_m^4\abs{w_1}^2\right)
\\
= &\abs{v_1}^2+\left(\sin\left(\theta/2\right)^2 +
\gamma_m^4\cos\left(\theta/2\right)^2\right)\abs{w_1}^2
\\
= & \abs{v_1}^2+\gamma_m^2\abs{w_1}^2.
\end{align*}
So clearly then
\begin{align*}
E_1 & = \abs{\frac{v_2+v_1}{2}}^2+\frac{1}{\gamma_m^2}
\abs{\frac{v_2-v_1}{2}}^2
\\
& =
\frac{2\h_1}{M}\cos\left(\theta/2\right)^2 +\gamma_m\cos\left(\theta/2\right)\Phi^TB_p+
\frac{\gamma_m^2-1}{4}\abs{\Phi}^2
\\
& =
\frac{2\h_1}{M(1+\gamma_m^2)}+\gamma_m\cos\left(\theta/2\right)\Phi^TB_p+
\frac{\gamma_m^2-1}{4}\abs{\Phi}^2
\\
& =\frac{2\h_1}{M_T}
+\gamma_m\cos\left(\theta/2\right)\Phi^TB_p+
\frac{\gamma_m^2-1}{4}\abs{\Phi}^2,
\end{align*}
where $M_T=(M+m)$.

Now this exercise is repeated for $w_1, w_2$ to find the proper squared differences. First use \eqref{l:REL17} of Appendix\eqref{l:AppendixA} for the third time to show that
\begin{align*}
\gamma_m\frac{w_2-w_1}{2}
& =
\gamma_m\left(\frac{\sin(\theta)}{2\gamma_m}v_1
-\frac{1}{2}\left(\cos(\theta)+1\right)w_1
-\frac{1}{2\gamma_m}\Phi\right)
\\
& = \gamma_m\cos^2\left(\theta/2\right)\left(v_1
-w_1\right)
-\frac{1}{2}\Phi
\\
&=\cos\left(\theta/2\right)C_p
-\frac{1}{2}\Phi,
\end{align*}
with
\begin{align*}
C_p & = \gamma_m\cos\left(\theta/2\right)v_1
-\gamma_m\cos\left(\theta/2\right)w_1
\\
& = \sin\left(\theta/2\right)\left(v_1
-w_1\right) =-A_p.
\end{align*}
But we know that $\Phi^T(w_2-w_1)=0$ so (by momentum conservation $\Phi^T(v_2-v_1)=\gamma_m^2\Phi^T(w_1-w_2)=0$) then
\begin{align*}
\cos\left(\theta/2\right)\Phi^TC_p
-\frac{1}{2}\abs{\Phi}^2=0.
\end{align*}
As a result
\begin{align}
\label{l:APP10}
\begin{split}
\gamma_m^2\abs{\frac{w_2-w_1}{2}}^2
& = \cos^2\left(\theta/2\right)\abs{C_p}^2
-\cos\left(\theta/2\right)\Phi^TC_p
+\frac{1}{4}\abs{\Phi}^2
\\
&=\cos^2\left(\theta/2\right)\abs{C_p}^2 -\frac{1}{4}\abs{\Phi}^2.
\end{split}
\end{align}

For the positive difference for the $w_2, w_1$ use equation \eqref{l:REL17} in Appendix\eqref{l:AppendixA} one more time from which follows that
\begin{align*}
\frac{w_2+w_1}{2}
& =\frac{\sin(\theta)}{2\gamma_m}v_1-\frac{1}{2}\left(\cos(\theta)-1\right)w_1
-\frac{1}{2\gamma_m}\Phi
\\
& = \cos\left(\theta/2\right)\big(\cos\left(\theta/2\right)v_1
+\gamma_m\sin\left(\theta/2\right)w_1\big) - \frac{1}{2\gamma_m}\Phi
\\
& = \cos\left(\theta/2\right)D_p - \frac{1}{2\gamma_m}\Phi,
\end{align*}
with
\begin{align*}
D_p & =\cos\left(\theta/2\right)v_1
+\gamma_m\sin\left(\theta/2\right)w_1
\\
& =\cos\left(\theta/2\right)\left(v_1
+\gamma_m^2w_1\right) = B_p.
\end{align*}
Hence
\begin{align}
\label{l:APP11}
\abs{\frac{w_2+w_1}{2}}^2
& = \cos\left(\theta/2\right)^2\abs{D_p}^2 -
\frac{1}{\gamma_m}\cos\left(\theta/2\right)\Phi^TD_p+
\frac{1}{4\gamma_m^2}\abs{\Phi}^2.
\end{align}

Now adding the squares using \eqref{l:APP10} and
\eqref{l:APP11} results in
\begin{align*}
& \abs{\frac{w_2+w_1}{2}}^2+\gamma_m^2\abs{\frac{w_2-w_1}{2}}^2
\\
& = \cos\left(\theta/2\right)^2\left(\abs{C_p}^2 + \abs{D_p}^2\right)
-\frac{1}{\gamma_m}\cos\left(\theta/2\right)\Phi^TD_p +
\frac{1}{4}\left(\frac{1}{\gamma_m^2}-1\right)\abs{\Phi}^2
\\
& = \cos\left(\theta/2\right)^2
\left(\abs{v_1}^2+\gamma_m^2\abs{w_1}^2\right)
-\frac{1}{\gamma_m}\cos\left(\theta/2\right)\Phi^TD_p+
\frac{1}{4}\left(\frac{1}{\gamma_m^2}-1\right)\abs{\Phi}^2,
\end{align*}
because
\begin{align*}
\abs{C_p}^2 + \abs{D_p}^2 = \abs{A_p}^2 + \abs{B_p}^2
= \abs{v_1}^2+\gamma_m^2\abs{w_1}^2.
\end{align*}
So clearly then for the sum $E_2$
\begin{align*}
E_2 & =\abs{\frac{w_2+w_1}{2}}^2+\gamma_m^2
\abs{\frac{w_2-w_1}{2}}^2
\\
& =
\frac{2\h_1}{M}\cos\left(\theta/2\right)^2 -\frac{1}{\gamma_m}\cos\left(\theta/2\right)\Phi^TD_p+
\frac{1}{4}\left(\frac{1}{\gamma_m^2}-1\right)\abs{\Phi}^2
\\
& =
\frac{2\h_1}{M(1+\gamma_m^2)}-\frac{1}{\gamma_m}\cos\left(\theta/2\right)\Phi^TD_p+
\frac{1}{4}\left(\frac{1}{\gamma_m^2}-1\right)\abs{\Phi}^2
\\
& =\frac{2\h_1}{M_T}
-\frac{1}{\gamma_m}\cos\left(\theta/2\right)\Phi^TD_p+
\frac{1}{4}\left(\frac{1}{\gamma_m^2}-1\right)\abs{\Phi}^2,
\end{align*}
and finally
\begin{align*}
E_1+\gamma_m^2E_2 = \frac{2\h_1(1+\gamma_m^2)}{M_T}=\frac{2\h_1}{M},
\end{align*}
as the second and third term cancel. This proves
Theorem \eqref{l:THEOREM5}.
\end{proof}

\section{}
\label{l:APP5} \noindent {\bf Proof of Theorem \eqref{l:THEOREM6}}
\begin{proof}
This argument is a variation on the proofs presented by Nelson~\cite{ENELSON1} and Carlen~\cite{ECARL1}. The point of this Theorem is to find a combination of the density $\rho(x,t)$ the forward drift $b^+$ and the backward drift $b^-$ for the stochastic process \eqref{l:DRFT2aa} so that \eqref{l:HAMILT1} is time invariant for the given potential $V(x,t)$. Once $b^+$ is determined equation \eqref{l:DRFT2ab} derives the backward drift $b^-$ and the probability density can be calculated from equations \eqref{l:DRFT3c}.  Beyond the potential $V(x,t)$ there are no other forces working on the particles.

First simplify the n-dimensional drifts $b^+,b^-$ by introducing the sufficiently smooth functions $R=R(x,t)$, $S=S(x,t)$ as follows
\begin{subequations}
\begin{align}
\label{l:APP13a}
& 2\gamma_m\nabla R = \sigma^2\frac{\nabla\rho}{\rho} =b^+-b^-,
\\
\label{l:APP13b}
& 2\nabla S =\left(b^++b^-\right),
\end{align}
\end{subequations}
or equivalently
\begin{align}
\label{l:APP1}
\begin{split}
b^+ & =\nabla S + \gamma_m\nabla R,
\\
b^- & =\nabla S - \gamma_m\nabla R.
\end{split}
\end{align}
Notice that \eqref{l:APP13b} implies that $b^++b^-$ can be represented in a gradient and \eqref{l:APP13a} is derived from \eqref{l:DRFT2aa} while suggesting that $\rho=\exp{{\left[\frac{2 \gamma_m R}{\sigma^2}\right]}}$ where the function $R$ is scaled to normalize the density.

Using \eqref{l:APP13a} and \eqref{l:APP13b} the expected functional \eqref{l:DRFT13} or functional
\eqref{l:HAMILT1} can be represented as
\begin{align}
\label{l:HAMILT37}
\begin{split}
\frac{2}{M}E\left[\h_1\right]
& = E
\begin{matrix}
\abs{\nabla S}^2
\end{matrix}
+
E\begin{matrix} \abs{\nabla R}^2
\end{matrix}
+\frac{2}{M}\int\rho Vdx
\\
& =
\int \rho S_{x_j}S_{x_j}dx
+
\int\rho R_{x_j}R_{x_j}+\frac{2}{M}\int\rho Vdx,
\end{split}
\end{align}
using Einstein's summation convention. The $\gamma_m$ in equation \eqref{l:APP13a} was clearly chosen for this equation to remove the $1/\gamma_m^2$ weighting of the functional \eqref{l:HAMILT1}. To be time independent this functional must have a zero derivative with respect to time. However, the time derivative of this functional creates time derivatives of $R_{x_j}$, $S_{x_j}$ with the potential $V=V(x,t)$ and the partial derivative of $\rho=\rho(x,t)$ satisfying the continuity equation \eqref{l:DRFT3c}.

Starting with time behaviour of the probability density it becomes clear from equation
\eqref{l:DRFT3c} that
\begin{align}
\label{l:HAMILT43}
\frac{\partial \rho}{\partial t}=\rho_t=-\nabla^T\left(\frac{b^++b^-}{2}\rho
\right)=-\nabla^T\left(\rho\nabla S \right),
\end{align}
from which follows
\begin{align}
\label{l:HAMILT41}
\begin{split}
\frac{\partial R}{\partial t} = &R_t=-\left(\frac{\sigma^2}{2\gamma_m}S_{x_jx_j} +
S_{x_j}R_{x_j}\right),
\\
\frac{\partial^2 R}{\partial t\partial x_k}=&R_{tx_k}=-\left(\frac{\sigma^2}{2\gamma_m}
S_{x_jx_j} + S_{x_j}R_{x_j}\right)_{x_k},
\end{split}
\end{align}
using Einstein's convention for the summation of indices again. This determines the time behaviour for $\rho$, $R$ and $R_{x_j}$.

The time derivative of the functional \eqref{l:HAMILT1} yields
\begin{align*}
&\int\rho_t
\begin{pmatrix}
S_{x_k}S_{x_k} + R_{x_k}R_{x_k}
\end{pmatrix}
dx  && \text{(a)}
\\
&+2\int\rho
\begin{pmatrix}
S_{x_k}S_{x_kt} +R_{x_k}R_{tx_k}
\end{pmatrix}
dx && \text{(b)}
\\
&+\frac{2}{M}\frac{d}{dt}\int\rho(x,t)V(x,t)dx=0, &&
\text{(c)}
\end{align*}
where $\rho_t$ is the time derivative of the probability density $\rho=\rho(x,t)$. As there are no other processes in this collision the time derivative has to be zero. The three individual terms are examined in order starting with $(a)$.

Substituting $\rho=\exp\left[\frac{2 \gamma_m R}{\sigma^2}\right]$ using equation \eqref{l:HAMILT41} the first term $(a)$ becomes
\begin{align*}
\text{(a)} = &\int\frac{2\gamma_m}{\sigma^2}\rho R_t
\begin{pmatrix}
S_{x_k}S_{x_k} + R_{x_k}R_{x_k}
\end{pmatrix}
dx
=\frac{2\gamma_m}{\sigma^2}\int\rho
\begin{pmatrix}
\left(-\frac{\sigma^2}{2\gamma_m} S_{x_jx_j} -
S_{x_j}R_{x_j}\right)
\\
\begin{pmatrix}
S_{x_k}S_{x_k} +R_{x_k}R_{x_k}
\end{pmatrix}
\end{pmatrix}
dx
\\
=&-\int \rho S_{x_jx_j}
\begin{pmatrix}
S_{x_k}S_{x_k} + R_{x_k}R_{x_k}
\end{pmatrix}
dx
-\frac{2\gamma_m}{\sigma^2}\int\rho S_{x_j}R_{x_j}
\begin{pmatrix}
S_{x_k}S_{x_k} + R_{x_k}R_{x_k}
\end{pmatrix}
dx.
\end{align*}
So with one partial integral on the first term this reduces to
\begin{align*}
\text{(a)}
=&\int\rho S_{x_j}
\begin{pmatrix}
S_{x_k}S_{x_k} + R_{x_k}R_{x_k}
\end{pmatrix}_{x_j}
dx
+\int \rho_{x_j} S_{x_j}
\begin{pmatrix}
S_{x_k}S_{x_k} + R_{x_k}R_{x_k}
\end{pmatrix}
dx
\\
&-\frac{2\gamma_m}{\sigma^2}\int\rho S_{x_j}R_{x_j}
\begin{pmatrix}
S_{x_k}S_{x_k} + R_{x_k}R_{x_k}
\end{pmatrix}
dx,
\end{align*}
or
\begin{align*}
\text{(a)} =&\int \rho S_{x_j}
\begin{pmatrix}
S_{x_k}S_{x_k} + R_{x_k}R_{x_k}
\end{pmatrix}_{x_j}
dx
+\frac{2\gamma_m}{\sigma^2}\int \rho R_{x_j} S_{x_j}
\begin{pmatrix}
S_{x_k}S_{x_k} + R_{x_k}R_{x_k}
\end{pmatrix}
dx
\\
&-\frac{2\gamma_m}{\sigma^2}\int\rho
R_{x_j}S_{x_j}
\begin{pmatrix}
S_{x_k}S_{x_k} + R_{x_k}R_{x_k}
\end{pmatrix}
dx,
\end{align*}
so that finally
\begin{align*}
\text{(a)}=&\int \rho S_{x_j}
\begin{pmatrix}
S_{x_k}S_{x_k} + R_{x_k}R_{x_k}
\end{pmatrix}_{x_j}
dx.
\end{align*}

To reduce the $(b)$ term use equation \eqref{l:HAMILT41} again for $R_{tx_k}$ to yield
\begin{align*}
\text{(b)}= & 2\int\rho
\begin{pmatrix}
S_{x_k}S_{x_kt}
- R_{x_k}\left(\frac{\sigma^2}{2\gamma_m}S_{x_jx_j} + S_{x_j}R_{x_j}\right)_{x_k}
\end{pmatrix}
dx
\\
=&2\int\rho
\begin{pmatrix}
S_{x_k}S_{x_kt} + R_{x_kx_k}
\begin{pmatrix}
\frac{\sigma^2}{2\gamma_m} S_{x_jx_j} + S_{x_j}R_{x_j}
\end{pmatrix}
\\
+\frac{2 \gamma_m}{\sigma^2}R_{x_k} R_{x_k}
\begin{pmatrix}
\frac{\sigma^2}{2\gamma_m} S_{x_jx_j} + S_{x_j}R_{x_j}
\end{pmatrix}
\end{pmatrix}
dx
\\
=&2\int\rho
\begin{pmatrix}
S_{x_k}S_{x_kt}
+\frac{\sigma^2}{2\gamma_m}R_{x_kx_k} S_{x_jx_j} + R_{x_kx_k}S_{x_j}R_{x_j}
\\
+ R_{x_k}R_{x_k}S_{x_jx_j} + \frac{2\gamma_m}{\sigma^2}R_{x_k}R_{x_k}S_{x_j}R_{x_j}
\end{pmatrix}
dx,
\end{align*}
using one partial integral on the second term. Rearranging the terms this can be written as
\begin{align*}
\text{(b)}= &2\int\rho
\begin{pmatrix}
S_{x_k}S_{x_kt} + S_{x_j}R_{x_j}\left(R_{x_kx_k} +\frac{2\gamma_m}{\sigma^2}R_{x_k}
R_{x_k}\right)
\end{pmatrix}
dx
\\
&+2\int\rho S_{x_jx_j}
\begin{pmatrix}
\frac{\sigma^2}{2\gamma_m}R_{x_kx_k} + R_{x_k} R_{x_k}
\end{pmatrix}
dx,
\end{align*}
and with one partial integral against the last term
this becomes
\begin{align*}
\text{(b)}= &2\int\rho
\begin{pmatrix}
S_{x_k}S_{x_kt} + S_{x_j}R_{x_j}\left(R_{x_kx_k} +\frac{2\gamma_m}{\sigma^2}R_{x_k}
R_{x_k}\right)
\end{pmatrix}
dx
\\
&-2\int\rho
S_{x_j}
\begin{pmatrix}
\frac{\sigma^2}{2\gamma_m}R_{x_kx_k} + R_{x_k} R_{x_k}
\end{pmatrix}_{x_j}
dx
\\
&-2\int\rho S_{x_j}R_{x_j}
\begin{pmatrix}
R_{x_kx_k} +\frac{2\gamma_m}{\sigma^2} R_{x_k} R_{x_k}
\end{pmatrix}
dx.
\end{align*}
Clearly the second term and the fourth term cancel so $(b)$ reduces to
\begin{align*}
\text{(b)}= &2\int\rho
\left(S_{x_k}S_{x_kt} - S_{x_j}
\begin{pmatrix}
\frac{\sigma^2}{2\gamma_m}R_{x_kx_k} + R_{x_k} R_{x_k}
\end{pmatrix}_{x_j}
\right)
dx.
\end{align*}

So then $(a)+ (b)$ becomes
\begin{align*}
\text{(a)}+\text{(b)}&
=\int \rho S_{x_j}
\begin{pmatrix}
S_{x_k}S_{x_k} + R_{x_k}R_{x_k}
\end{pmatrix}_{x_j}
dx
\\
&+2\int\rho
\left(S_{x_k}S_{x_kt} - S_{x_j}
\begin{pmatrix}
\frac{\sigma^2}{2\gamma_m}R_{x_kx_k} + R_{x_k} R_{x_k}
\end{pmatrix}_{x_j}
\right)
dx,
\end{align*}
which reduces to
\begin{align*}
\begin{split}
\text{(a)}+\text{(b)}
= \int\rho S_{x_j}
\begin{pmatrix}
2S_t+S_{x_k}S_{x_k} -R_{x_k}R_{x_k} - \frac{\sigma^2}{\gamma_m}R_{x_kx_k}
\end{pmatrix}_{x_j}
dx,
\end{split}
\end{align*}
once the summation over $k$ is written in $j$.

Assume that there is an appropriate function $\Psi=\Psi(x,t)$ capturing the time behaviour of the potential then for the $(c)$ term
\begin{align*}
\frac{\partial V}{\partial t} = \big(\nabla S\big)^T \nabla\Psi,
\end{align*}
and with one partial integral
\begin{align*}
\frac{d}{dt} \int \rho V(x,t) dx
& = \int \rho_t V dx + \int \rho \frac{\partial V}{\partial t}dx
\\
& = -\int \nabla^T(\rho\nabla S) V dx +
\int \rho\big(\nabla S\big)^T \nabla\Psi dx
\\
&= \int \rho S_{x_j} V_{x_j} dx +\int \rho S_{x_j}\Psi_{x_j} dx.
\end{align*}
Now adding the second term $\text{(c)}$ to the results $(a) + (b)$ the time change of the potential becomes
\begin{align}
\label{l:DRFT131}
\begin{split}
\text{(a)}+\text{(b)}+\text{(c)}
=&\int \rho S_{x_j}
\begin{pmatrix}
2 S_t + S_{x_k}S_{x_k} - \frac{\sigma^2}{\gamma_m}R_{x_kx_k} - R_{x_k}R_{x_k}
\end{pmatrix}_{x_j}
dx
\\
&+\int\rho S_{x_j}
\frac{2}{M}(V+\Psi)_{x_j}dx,
\end{split}
\end{align}
including the derivative of the potential energy $V$ term. Clearly then a sufficient condition for the Hamiltonian to become time-independent is that the integrand in \eqref{l:DRFT131} is to equal zero. In other words
\begin{align}
\label{l:SCHROD6}
S_{t}-\frac{1}{2}\abs{\nabla R}^2
+\frac{1}{2} \abs{\nabla S}^2 -\frac{\eta}{2M}\Delta R
+\frac{1}{M}\left(V+\Psi\right)=0.
\end{align}
The following proposition shows that this is equivalent to demanding that the probability for the main particle position $\rho=\rho(x,t)$ is derived from a wave function satisfying Schr\"{o}dinger's equation. The following Proposition shows this the case.

\begin{prop}
\label{l:prp22}
Define the wave function using proper functions $R=R(x,t)$ nd $S=S(x,t)$ then
\begin{align*}
\psi=\psi(x,t)=\exp\left[M\left(\frac{R+iS}{\eta}\right)\right],
\end{align*}
with $\rho=\rho(x,t)=\,|\psi(x,t)\,|^2=\,|\psi\,|^2$ as the probability of $x(t)$ and $\eta = M\sigma^2/\gamma_m$. Then the functional \eqref{l:HAMILT1} represented as
\begin{align*}
\begin{split}
E\left[\h_1\right] = E\left[\h_2\right]
= \frac{M}{2}
\left(E\abs{\nabla S}^2 + E\abs{\nabla R}^2\right)
+\int\rho\left(V+\Psi\right)dx,
\end{split}
\end{align*}
is time invariant if and only if
\begin{align}
\label{app:POT13}
i\eta\psi_t=-\frac{\eta^2}{2M}\Delta
\psi+(V+\Psi)\psi.
\end{align}
\end{prop}
\begin{proof}
The proof involves a straightforward verification of equation \eqref{l:DRFT131} by brute force.  Expanding the derivatives on the righthand side of this wave function
\begin{align*}
\frac{\eta}{iM}\psi_t&=\psi\left( -i R_t+S_t\right),
\\
\frac{\eta}{2M}\psi_{x_j}&=\frac{1}{2}\psi\left( R_{x_j}+iS_{x_j}\right),
\\
\frac{\eta^2}{2M^2}\psi_{x_jx_j}&=\frac{1}{2}\psi\left(R_{x_j}+iS_{x_j}\right)^2 + \frac{\eta}{2M}\psi\left( R_{x_jx_j}+i S_{x_jx_j}\right),
\end{align*}
which combined becomes
\begin{align*}
\frac{\eta}{iM}\psi_t-\frac{\eta^2}{2M^2}\psi_{x_jx_j} =
&-i\psi\left( R_t + R_{x_j}S_{x_j}+\frac{\eta}{2M}S_{x_jx_j}\right)
\\
&\quad + \psi\left( S_t-\frac{1}{2}\abs{\nabla R}^2 +
\frac{1}{2}\abs{\nabla S}^2
-\frac{\eta}{2M}R_{x_jx_j}\right),
\end{align*}
because of \eqref{l:SCHROD6}. Using the definition of $\eta$ the first term disappears and applying \eqref{l:SCHROD6} this reduces to
\begin{align*}
\frac{\eta}{iM}\psi_t=\frac{\eta^2}{2M^2}\psi_{x_jx_j}
-\frac{\left(V+\Psi\right)}{M}\psi
\end{align*}
which after multiplication with the mass $M$ and the complex number $i^2=-1$ yields Schr\"{o}dinger's equation \eqref{app:POT13} with $\eta = M\sigma^2/\gamma_m$. Notice that
\begin{align*}
\rho(x,t)=\abs{\psi(x,t)}^2=e^{\frac{2M R}{\eta}}=e^{\frac{2\gamma_m R}{\sigma^2}},
\end{align*}
which is consistent with our original assumption.
\end{proof}
An immediate difference with quantum mechanics is that the mean acceleration of the motion is guided by the potential $V(x,t)+\Psi(x,t)$. This means that the Schr\"{o}dinger's equation does not by itself preserve the energy as defined by \eqref{l:HAMILT1} if the potential is time dependent. The usual properties of quantum mechanics apply though.

\begin{prop}
\label{l:PROP2} Some straightforward derivatives show that
\begin{align*}
\frac{d^2}{dt^2}E[x(t)]=-\frac{1}{M}\int\rho\nabla(V+\Psi)
dx.
\end{align*}
\end{prop}
\begin{proof}
Using the continuity equation \eqref{l:HAMILT37} with one partial derivative
\begin{align*}
\frac{d}{dt}E[x_k(t)]&=\frac{d}{dt}\int x_k \rho(x,t)dx =
\int x_k \rho_tdx = -\int x_k \left(S_{x_j}\rho\right)_{x_j} dx
\\
&=\int S_{x_k}\rho dx,
\end{align*}
so that
\begin{align*}
\frac{d^2}{dt^2}E[x_k(t)]
&=
\int S_{x_kt}\rho dx
+\int S\rho_t dx
=
\int S_{x_kt}\rho dx
-\int S_{x_k} \left(\rho S_{x_j}\right)_{x_j} dx
\\
&=
\int S_{x_kt}\rho dx +
\int S_{x_jx_k} S_{x_j}\rho dx
=
\int S_{x_kt}\rho dx +
\frac{1}{2}\int \left(S_{x_j} S_{x_j}\right)_{x_k}\rho dx.
\end{align*}

Using the Schr\"{o}dinger's equation \eqref{l:SCHROD5} and extending to all dimensions this equation becomes
\begin{align*}
\frac{d^2}{dt^2}E[x_k(t)] =
&\int\rho \nabla\left(S_t+\frac{1}{2}\abs{\nabla
S}^2\right) dx
\\
=&\int\rho
\nabla\left(\frac{1}{2}\abs{\nabla R}^2+\frac{\eta^2}{2M}\Delta
R- \frac{(V+\Psi)}{M}\right) dx
\\
=&-\frac{1}{M}\int\rho\nabla \left(V+\Psi\right)dx,
\end{align*}
as clearly
\begin{align*}
\int\rho
\nabla\left(\frac{1}{2}\abs{\nabla R}^2+\frac{\eta^2}{2M}\Delta
R\right) dx = 0,
\end{align*}
settling the proof.
\end{proof}
\end{proof}

\section{}
\label{l:APP14} \noindent {\bf Proof of Theorem
\eqref{l:THEOREM8}}
\begin{proof}
This argument is an extension of the proof produced in \eqref{l:APP5} where now both the main and the incident particle are represented by the equation of motion \eqref{l:DRFTDEF1a} and \eqref{l:DRFTDEF1b}. The easiest way of achieving the constant energy for the combined energy \eqref{l:HAMILT1} is to consider both sets of stochastic differential equation with independent drift functions $b^+=b^+(x,y,t),b^-=b^-(x,y,t)$ and the variances $\eta,\eta_w$. Not surprisingly in this case both the main particle $v_1,v_2$ and incident particle $w_1,w_2$ particle must be described by Schr\"{o}dinger's equation. This equation is the same representation as presented in Theorem \eqref{l:TIMEERG2} so the minimum \eqref{l:HAMILT39} is written as a double term plus a variance represented as a constant.

To unify the notation to two dimensions introduce the sufficiently smooth function for the main particle $S_1=S_1(x,t)$, $R_1=R_1(x,t)\in \Re^{n}$ with $x=(x_1,...,x_n)$ and the incident particle $S_2=S_2(y,t)$, $R_2=R_2(y,t)\in \Re^{n}$ with $y=(x_{n+1},...,x_{2n})$ and define constants 
\begin{align}
\label{l:APP20}
\begin{split}
\alpha'(j) = & 1, \beta'(j) = \gamma_m, \quad\quad\quad\quad 1 \leq j \leq n,
\\
\alpha'(j) = & 1/\gamma_m, \beta'(j) = 1/\gamma_m^2, \quad n < j \leq 2n.
\end{split}
\end{align}
Define the $R$ and $S$ functions and assume that
\begin{align}
\label{l:APP22}
\begin{split}
R(x,y,t) & = R(x_1,...,x_n,x_{n+1}...,x_{2n},t) = R_1(x_1,...,x_n,t) + R_2(x_{n+1},...,x_{2n},t)
\\
& = R_1(x,t) + R_2(y,t),
\\
S(x,y,t) & = S(x_1,...,x_n,x_{n+1}...,x_{2n},t) = S_1(x_1,...,x_n,t) + S_2(x_{n+1},...,x_{2n},t)
\\
& = S_1(x,t) + S_2(y,t),
\end{split}
\end{align}
which simplifies the probability distribution. Now assume that for all $j=1,...,2n$ the drifts can be separated as follows
\begin{subequations}
\begin{align}
\label{l:APP14a}
2\alpha'(j)\nabla_j S & = \left(b_j^++b_j^-\right),
\\
\label{l:APP14b}
2\beta'(j)\nabla_j R & = \sigma^2\frac{\nabla_j\rho}{\rho} =b_j^+-b_j^-,
\end{align}
\end{subequations}
for $\rho=\rho(x,y,t),x,y\in\Re^{2n}$ and $\nabla_j=\frac{\partial}{\partial x_j},j=1,...,2n$. Then the drifts can be calculated as
\begin{align*}
b_j^+ & =\alpha'(j)\nabla_j S +\beta'(j)\nabla_j R,
\\
b_j^- & =\alpha'(j)\nabla_j S -\beta'(j)\nabla_j R,
\end{align*}
Notice that \eqref{l:APP14a} are derived from \eqref{l:DRFT2aa} while \eqref{l:APP14b} implies that $b^++b^-$ can now be represented in a gradient of the function $S$. Specifically, equations \eqref{l:APP14b} suggests that for all $j$
\begin{align}
\label{l:APP17}
\begin{split}
\nabla_{j} R &= \frac{\sigma^2}{2\beta'(j)} \frac{\nabla_j\rho}{\rho}=\frac{\sigma^2}{2\beta'(j)} \nabla_j\log(\rho),
\end{split}
\end{align}
which can be solved by $\rho=\exp(2(\beta_1R_1 + \beta_2R_2)/\sigma^2)=\exp(2\beta_1R_1/\sigma^2)\exp(2\beta_2R_2/\sigma^2)$ so that the density for $\rho(x,y,t)=\rho_1(x_1,...,x_n)\rho_2(x_{n+1},...,x_{2n})=\rho_1(x,t)\rho_2(y,t)$ is scaled by $R_1,R_2$. The constant is not included meaning that $R_1,R_2$ have to be scaled to accommodate the probability density unity.

To obtain the full Hamiltonian \eqref{l:HAMILT39} for the joint energy for both the main and incident particle using \eqref{l:APP14a}, \eqref{l:APP14b} equals
\begin{align}
\label{l:HAMILT40}
\begin{split}
\frac{2}{M}E\left[\h_1\right] = & \frac{2}{M}E\left[\h_2\right]
\\
= & \sum_{j=1}^{j=n}\int \rho \abs{\frac{b_j^++b_j^-}{2}}^2dx + \frac{1}{\gamma_m^2}\sum_{j=1}^{j=n}\int \rho \abs{\frac{b_j^+-b_j^-}{2}}^2dx
\\
& + \gamma_m^2\sum_{j=n+1}^{j=2n}\int \rho \abs{\frac{b_j^++b_j^-}{2}}^2dx + \gamma_m^4\sum_{j=n+1}^{j=2n}\int \rho \abs{\frac{b_j^+-b_j^-}{2}}^2dx
\\
= & \sum_{j=1}^{j=n}\alpha'(j)^2\int \rho \abs{\nabla_j S}^2dx + \frac{1}{\gamma_m^2}\sum_{j=1}^{j=n}\beta'(j)^2\int \rho \abs{\nabla_j R}^2dx,
\\
& + \gamma_m^2\sum_{j=n+1}^{j=2n}\alpha'(j)^2\int \rho \abs{\nabla_j S}^2dx + \gamma_m^4\sum_{j=n+1}^{j=2n}\beta'(j)^2\int \rho \abs{\nabla_j R}^2dx,
\end{split}
\end{align}
which becomes
\begin{align*}
\frac{2}{M}E\left[\h_1\right] & =\frac{2}{M}E\left[\h_2\right] = \sum_{j=1}^{j=2n}\int \rho \abs{\nabla_j S}^2dx + \sum_{j=1}^{j=2n}\int \rho \abs{\nabla_j R}^2dx,
\end{align*}
for the constants $\alpha'(j),\beta'(j)$ scaled in \eqref{l:APP20} with $\gamma_m^2 = m/M$ and $E[.] = \int.\rho(x,y,t)dx$. This is Theorem \eqref{l:THEOREM5} again but now there are appropriate choices of stochastic drifts relating to $w_1,w_2$.

To solve the time dependence of \eqref{l:HAMILT40} some time relationships need to be determined and employing equations \eqref{l:APP14a}, \eqref{l:APP14b} and equation using
\begin{align*}
\frac{\partial \rho}{\partial t}=-\sum_{j=1}^{j=2n}\nabla_j\left(\rho\nabla_j S\right),
\end{align*}
then the time dependence of $\rho$ using \eqref{l:APP17} becomes
\begin{align}
\label{l:APP16}
\frac{\partial \rho}{\partial t} = &  -\rho\sum_{j=1}^{j=2n}\left(\nabla_{jj}S +\frac{2\beta'(j)\nabla_jS\nabla_jR}{\sigma^2}\right).
\end{align}
Using \eqref{l:APP17} again with equation \eqref{l:APP16} it is found that for all $j$ the time dependence becomes
\begin{align}
\label{l:APP18}
2\beta'(j)\nabla_{j} R_t= \sigma^2 \nabla_j\left(\frac{\partial \log(\rho)}{\partial t}\right)=-\sigma^2\nabla_j\sum_{k=1}^{k=2n}\left(\nabla_{kk}S +\frac{2\beta'(k)\nabla_kS\nabla_kR}{\sigma^2}\right),
\end{align}
which will be used below.

To be time independent this functional must have a zero derivative with respect to time. So taking the time derivative of the functional \eqref{l:HAMILT40} for the main and incident particles yields
\begin{align*}
\frac{2}{M}\frac{\partial E[\h_1]}{\partial t} & =\frac{2}{M}\frac{\partial E[\h_2]}{\partial t} =
\\
& \sum_{j=1}^{j=2n}\int\frac{\partial \rho}{\partial t} \abs{\nabla_j S}^2dx + \sum_{j=1}^{j=2n}\int\frac{\partial \rho}{\partial t} \abs{\nabla_j R}^2
dx
&& \text{(a)}
\\
& +2\sum_{j=1}^{j=2n}\int\rho \nabla_jS \nabla_j S_tdx + 2\sum_{j=1}^{j=2n}\int\rho \nabla_jR \nabla_jR_t
dx.
&& \text{(b)}
\end{align*}
To determine the first term $(a)$ substitute \eqref{l:APP16} and changing the variable $j$ to $k$ the first term becomes
\begin{align*}
\text{(a)} = & -\sum_{j=1}^{j=2n}\int\rho \left(\nabla_{jj}S +\frac{2\beta'(j)\nabla_jS\nabla_jR}{\sigma^2}\right)
\sum_{k=1}^{k=2n}
\begin{pmatrix}
\abs{\nabla_kS}^2 + \abs{\nabla_kR}^2
\end{pmatrix}
dx
\\
= & -\sum_{j=1}^{j=2n}\int\rho \nabla_{jj}S
\sum_{k=1}^{k=2n}
\begin{pmatrix}
\abs{\nabla_kS}^2 + \abs{\nabla_kR}^2
\end{pmatrix}
dx
\\
&-\sum_{j=1}^{j=2n}\frac{2\beta'(j)}{\sigma^2}\int\rho \nabla_jS\nabla_jR
\sum_{k=1}^{k=2n}
\begin{pmatrix}
\abs{\nabla_kS}^2 + \abs{\nabla_kR}^2
\end{pmatrix}
dx,
\end{align*}
and with one partial integral on the first term this reduces to
\begin{align*}
\text{(a)} = & \sum_{j=1}^{j=2n}\int\rho \nabla_{j}S \nabla_j\left(
\sum_{k=1}^{k=2n}
\begin{pmatrix}
\abs{\nabla_kS}^2 + \abs{\nabla_kR}^2
\end{pmatrix}
\right)
dx,
\end{align*}
as the second term of the partial integral cancels the last term.

To reduce the (b) term using equation \eqref{l:APP18} first the derivatives with time are removed via partial integration so that
\begin{align*}
\text{(b)}= & 2\sum_{j=1}^{j=2n}\int\rho \nabla_jS \nabla_j S_tdx + 2\sum_{j=1}^{j=2n}\int\rho \nabla_jR \nabla_jR_tdx
\\
= & 2\sum_{j=1}^{j=2n}\int\rho \nabla_jS \nabla_j S_tdx -
\sigma^2\sum_{j=1}^{j=2n}\int\rho \nabla_jR \frac{\nabla_j\sum_{k=1}^{k=2n}\left(\nabla_{kk}S +\frac{2\beta'(k)\nabla_kS\nabla_kR}{\sigma^2}\right)}{\beta'(j)}dx
\\
= & 2\sum_{j=1}^{j=2n}\int\rho \nabla_jS \nabla_j S_tdx +
\sigma^2\sum_{j=1}^{j=2n}\int\rho \nabla_{jj}R \frac{\sum_{k=1}^{k=2n}\left(\nabla_{kk}S +\frac{2\beta'(k)\nabla_kS\nabla_kR}{\sigma^2}\right)}{\beta'(j)}dx
\\
& \qquad +\sigma^2\sum_{j=1}^{j=n}\int\nabla_j\rho \nabla_jR\frac{\sum_{k=1}^{k=2n}\left(\nabla_{kk}S +\frac{2\beta'(k)\nabla_kS\nabla_kR}{\sigma^2}\right)}{\beta'(j)}dx
\\
= & 2\sum_{j=1}^{j=2n}\int\rho \nabla_jS \nabla_j S_tdx +
\sigma^2\sum_{j=1}^{j=2n}\int\rho \nabla_{jj}R \frac{\sum_{k=1}^{k=2n}\left(\nabla_{kk}S +\frac{2\beta'(k)\nabla_kS\nabla_kR}{\sigma^2}\right)}{\beta'(j)}dx
\\
& \qquad +2\sum_{j=1}^{j=2n}\int\rho \abs{\nabla_jR}^2\sum_{k=1}^{k=2n}\left(\nabla_{kk}S +\frac{2\beta'(k)\nabla_kS\nabla_kR}{\sigma^2}\right)dx,
\end{align*}
using \eqref{l:APP17} and \eqref{l:APP14a} again.

Collecting terms in this equation yields
\begin{align*}
\text{(b)}= & 2\sum_{j=1}^{j=2n}\int\rho \nabla_jS \nabla_j S_tdx +
\sigma^2\sum_{j,k=1}^{j,k=2n}\int\rho \nabla_{kk}S\left(\frac{\nabla_{jj}R}{\beta'(j)} +\frac{2\abs{\nabla_jR}^2}{\sigma^2}\right)dx
\\
& +2\sum_{j,k=1}^{j,k=2n}\frac{\beta'(k)}{\beta'(j)}\int\rho \nabla_kS\nabla_{jj}R\nabla_kRdx +\frac{4}{\sigma^2}\sum_{j,k=1}^{j,k=2n}\beta'(k)\int\rho \nabla_kS\abs{\nabla_jR}^2\nabla_kR
dx,
\end{align*}
and then taking another partial derivative against the second term removing the second derivative
\begin{align*}
\text{(b)}
= & 2\sum_{j=1}^{j=2n}\int\rho \nabla_jS \nabla_j S_tdx -
\sigma^2\sum_{j,k=1}^{j,k=2n}\int\rho \nabla_kS\nabla_k\left(\frac{\nabla_{jj}R}{\beta'(j)} +\frac{2\abs{\nabla_jR}^2}{\sigma^2}\right)dx
\\
& \qquad -2\sum_{j,k=1}^{j,k=2n}\beta'(k)\int\rho \nabla_kS\nabla_kR\left(\frac{\nabla_{jj}R}{\beta'(j)} +\frac{2\abs{\nabla_jR}^2}{\sigma^2}\right)dx
\\
& \qquad +2\sum_{j,k=1}^{j,k=2n}\frac{\beta'(k)}{\beta'(j)}\int\rho \nabla_kS\nabla_{jj}R\nabla_kRdx
\\
& \qquad + \frac{4}{\sigma^2}\sum_{j,k=1}^{j,k=2n}\beta'(k)\int\rho \nabla_kS\abs{\nabla_jR}^2\nabla_kR
dx
\\
= & 2\sum_{j=1}^{j=2n}\int\rho \nabla_jS \nabla_j S_tdx -
\sigma^2\sum_{j,k=1}^{j,k=2n}\int\rho \nabla_kS\nabla_k\left(\frac{\nabla_{jj}R}{\beta'(j)} +\frac{2\abs{\nabla_jR}^2}{\sigma^2}\right)dx,
\end{align*}
as four terms cancel in the equation. Applying the terms finally results in
\begin{align*}
\text{(a+b)}= &\sum_{j=1}^{j=2n}\int\rho \nabla_jS \nabla_j
\begin{pmatrix}
2\frac{\partial S}{\partial t} + \sum_{k=1}^{k=2n}\abs{\nabla_kS}^2 - \sum_{k=1}^{k=2n}\abs{\nabla_kR}^2
-\sum_{k=1}^{k=2n}\frac{\sigma^2}{\beta'(k)}\nabla_{kk}R
\end{pmatrix}dxdy.
\end{align*}

To simplify this equation notice that this expression is equivalent to the sum of two Schr\"{o}dinger equations employing different variances $\eta$ and $\eta_w$. Following \eqref{l:DRFT131} imagine that there are two potentials $-V_1(x,y,t)$ and $V_2(x,y,t)$ and the parameter separation suggested in \eqref{l:APP22} then this equation becomes
\begin{align*}
\text{(a+b)}=
&\sum_{j=1}^{j=n}\int\rho \nabla_jS_1 \nabla_j
\begin{pmatrix}
2\frac{\partial S_1}{\partial t} + \sum_{k=1}^{k=n}\abs{\nabla_kS_1}^2 - \sum_{k=1}^{k=n}\abs{\nabla_kR_1}^2
-\frac{\eta}{M}\sum_{k=1}^{k=n}\nabla_{kk}R_1
\end{pmatrix}
dxdy
\\
& + \sum_{j=n+1}^{j=2n}\int\rho \nabla_jS_2 \nabla_j
\begin{pmatrix}
2\frac{\partial S_2}{\partial t} + \sum_{k=n+1}^{k=2n}\abs{\nabla_kS_2}^2 - \sum_{k=n+1}^{k=2n}\abs{\nabla_kR_2}^2
\\
-\frac{\eta_w}{m}\sum_{k=n+1}^{k=2n}\nabla_{kk}R_2
\end{pmatrix}
dydx
\\
= &\frac{2}{M}\sum_{k=1}^{k=n}\int \rho(x,y,t) \nabla_jS_1 \nabla_j
\begin{pmatrix}
\left(V_1(x,y,t)+\Psi_1(x,y,t)\right)
\end{pmatrix}dxdy
\\
&\qquad - \frac{2}{m}\sum_{k=n+1}^{k=2n}\int \rho(x,y,t) \nabla_jS_2 \nabla_j
\begin{pmatrix}
\left(V_2(x,y,t)+\Psi_2(x,y,t)\right)
\end{pmatrix}dxdy,
\end{align*}
since by equation \eqref{l:APP19}
\begin{align*}
\frac{\sigma^2}{\beta'(k)} = & \frac{\sigma^2}{\gamma_m}=\frac{\eta}{M},\quad 1\leq k \leq n,
\\
\frac{\sigma^2}{\beta'(k)} = & \sigma^2\gamma_m=\frac{\eta_w}{m},\quad n < k \leq 2n,
\end{align*}
while $\eta = M\sigma^2/\gamma_m$, $\eta_w = m\sigma^2\gamma_m$ and
\begin{align}
\label{APP28}
\begin{split}
\frac{\partial V_1(x,y,t)}{\partial t} & = \big(\nabla S\big)^T \nabla\Psi_1(x,y,t),
\\
\frac{\partial V_2(x,y,t)}{\partial t} & = \big(\nabla S\big)^T \nabla\Psi_2(x,y,t).
\end{split}
\end{align}

Since it was assumed that
\begin{align*}
\rho(x,y,t)=\rho_1(x,t)\rho_2(y,t)=\abs{\psi_1(x,t)}^2\abs{\psi_2(y,t)}^2,
\end{align*}
it is finally shown that
\begin{align}
\label{APP25}
\begin{split}
\text{(a+b)}=
&\frac{2}{M}\sum_{k=1}^{k=n}\int \rho_1(x,t) \nabla_jS_1 \nabla_j \big(\overline{V_1(x,t)}\big)dx
\\
& - \frac{2}{m}\sum_{k=n+1}^{k=2n}\int \rho_2(y,t) \nabla_jS_2 \nabla_j \big(\overline{V_2(y,t)}\big)dy,
\end{split}
\end{align}
where
\begin{align*}
\overline{V_1(x,t)} & = \int \rho_2(y,t)\left(V_1(x,y,t)+\Psi_1(x,y,t)\right)dy,
\\
\overline{V_2(y,t)} & = \int \rho_1(x,t)\left(V_2(x,y,t)+\Psi_2(x,y,t)\right)dx.
\end{align*}
Hence the energy $(a+b)$ will be zero if the wave functions
\begin{align*}
\psi_1(x,t) & = \exp\left(\frac{M(R_1(x,t)+iS_1(x,t))}{\eta}\right),
\\
\psi_2(y,t) & = \exp\left(\frac{m(R_2(y,t)+iS_2(y,t))}{\eta_w}\right),
\end{align*}
satisfy Schr\"{o}dinger equations with potentials $\overline{V_1(x,t)}=\int\rho_2(y,t)\left(V_1(x,y,t)+\Psi_1(x,y,t)\right)dy$ and $\overline{V_2(x,t)}=\int\rho_1(x,t)\left(V_2(x,y,t)+\Psi_2(x,y,t)\right)dx$ satisfying \eqref{APP28} and \eqref{APP25}.

\end{proof}

\end{document}